\documentclass{article}
\usepackage{amsmath,amssymb,amsthm,times,url}
\usepackage{mathtools}
\usepackage{color}

\def\d{{\rm d}}
\def\e{{\rm e}}

\def\la{\lambda}
\def\a{\alpha}
\def\b{\beta}
\def\ga{\gamma}
\def\Ph{\Phi_{\la}}

\def\ph{\varphi}
\def\e{\varepsilon}
\def\ep{\varepsilon}
\def\th{\vartheta}

\def\W{\mathcal{W}}

\def\N{\mathbb{N}}
\def\R{\mathbb{R}}
\def\Z{\mathbb{Z}}
\newcommand{\deriv}[3][]{\frac{\d^{#1}{#2}}{{\d{#3}}^{#1}}}
\def\PI{\hbox{\rm P$_{\rm I}$}}
\def\PII{\hbox{\rm P$_{\rm II}$}}

\def\PIV{\hbox{\rm P$_{\rm IV}$}}

\def\dPI{\hbox{\rm dP$_{\rm I}$}}

\newcommand{\WhitD}[1]{D_{#1}}

\newtheorem{theorem}{Theorem}[section]

\newtheorem{lemma}[theorem]{Lemma}
\newtheorem{corollary}[theorem]{Corollary}
\theoremstyle{definition}

\newtheorem{remark}[theorem]{Remark}
\newtheorem{remarks}[theorem]{Remarks}

\def\ds{\displaystyle}
\def\erf{\mathop{\rm erf}\nolimits}

\numberwithin{figure}{section}
\numberwithin{equation}{section}
\numberwithin{table}{section}
\def\ww#1{w(#1;t)}
\def\p{Painlev\'{e}}
\def\peq{Painlev\'{e} equation}
\def\peqs{Painlev\'{e} equations}
\def\bts{B\"ack\-lund transformations}

 \def\O{\mathcal{O}}
\def\SS{\widetilde{S}}
\def\QQ{\widetilde{Q}}
\def\bb{\widetilde{\b}}
\def\hb{\widehat{\b}}
\def\comment#1{}
\begin{document}
\title{Properties of Generalized Freud Polynomials}

\author{Peter A. Clarkson$^{1}$ and Kerstin Jordaan$^{2}$\\[2.5pt]
$^{1}$ School of Mathematics, Statistics and Actuarial Science,\\ University of Kent, Canterbury, CT2 7FS, UK\\ {\tt P.A.Clarkson@kent.ac.uk}\\[2.5pt]
$^{2}$ Department of Decision Sciences,\\
University of South Africa, Pretoria, 0003, South Africa\\ 
{\tt jordakh@unisa.ac.za}}

\maketitle
\begin{abstract}
\noindent We consider the semiclassical generalized Freud weight function 
\[w_{\lambda}(x;t) = |x|^{2\lambda+1}\exp(-x^4 +tx^2),\qquad \lambda>-1,\quad x\in\mathbb{R}.\] We analyse the asymptotic behaviour of the sequences of monic polynomials that are orthogonal with respect to $w_{\lambda}(x;t)$, as well as the asymptotic behaviour of the recurrence coefficient, when the degree, or alternatively, the parameter $t$, tend to infinity. We also investigate existence and uniqueness of positive solutions of the nonlinear discrete equation satisfied by the recurrence coefficients and prove properties of the zeros of the generalized Freud polynomials.
\end{abstract}

 \section{Introduction}
The study of polynomials orthogonal on unbounded intervals with respect to general exponential-type weights\break $\exp\{-Q(x)\}$, with $Q(x)$ a polynomial of the form $Q(x)=|x|^{\alpha}$, with $\alpha\in \N$, began with G\'{e}za Freud in the 1960's (for details see \cite{refFreud1971,refFreud1976,refNevai1984a,refNevai1986} as well as the monographs by Levin and Lubinsky \cite{refLevinLubinsky01} and Mhaskar \cite{refMhaskar}). Earlier Freud \cite{refFreud1976,refFreud1986} investigated the asymptotic behaviour of the recurrence coefficients for special classes of weights by a technique giving rise to an infinite system of nonlinear equations called Freud equations for these coefficients, cf.~\cite{refMagnus86,refMagnus95}. If the monic orthogonal polynomials $\{ p_n(x) \}_{n=0}^{\infty}$ satisfying the three-term recurrence relation
\begin{equation}\label{orthogonalrecur}
p_{n+1}(x)= xp_{n}(x)-\b_n p_{n-1}(x), 
\end{equation} with $p_{-1}(x)=0$ and $p_0(x)=1$,
are related to the weight $w(x)=\exp(-x^4)$ on the whole real line, then the Freud equations are reduced to
(cf.~\cite{refBauldryMN88,refFreud1986,refKMcLaughlin,refMagnus85,refNevai1973,refNevai1983,refNevai1986})
\begin{subequations}\label{nonlinear}
\begin{equation}\label{nonlinear1}
4\b_{n}\left( \b_{n-1} + \b_{n} + \b_{n+1}\right) =n, 
\end{equation} with initial conditions
\begin{equation}\label{nonlinear2} 
\b_0=0,\quad \b_{1} = \frac{\int_{-\infty}^\infty x^2 \exp(-x^4)\, dx}{\int_{-\infty}^\infty \exp(-x^4)\, dx}
= \frac{\Gamma(\frac{3}{4})}{\Gamma(\frac{1}{4})}.
\end{equation} \end{subequations}
We remark that equation \eqref{nonlinear1} was first derived by Shohat \cite[equation (39), p.\ 407]{refShohat39}.
Nevai \cite{refNevai1983} proved that there is a unique positive solution to the problem \eqref{nonlinear}.

Freud \cite{refFreud1976}, via the Freud equations, conjectured that the asymptotic behaviour of recurrence coefficients $\b_n$ in the recurrence relation \eqref{orthogonalrecur} satisfied by the polynomials $\{ p_n (x) \}_{n=0}^{\infty}$ orthogonal with respect to the weight
 \begin{equation}\label{conj}w(x) = |x|^{\la}\exp(-|x|^{m}), \qquad m\in \N,\end{equation}
 with $\la>-1$, could be described by
\begin{equation}\label{Freudconj}
\ds \lim_{n\rightarrow \infty} \frac{\b_n}{n^{2/m}}= \left[ \frac{\Gamma(\tfrac{1}{2}m) \Gamma(1+\tfrac{1}{2}m)}{\Gamma(m+1)}\right]^{2/m}.
\end{equation}
We note that Freud \cite{refFreud1976} proved the result for orthonormal polynomials whilst \eqref{Freudconj} is for monic orthogonal polynomials.
Freud showed that if the limit exists for $m\in2\Z$, then it is equal to the expression in \eqref{Freudconj} but could only prove the existence of the limit \eqref{Freudconj} for $m=2,4,6$. Significant progress in the study of orthogonal polynomials associated with Freud weights was made when Magnus \cite{refMagnus85} proved the validity of Freud's conjecture for the recurrence coefficients when $m$ is an even positive integer and weight 
\begin{equation} \label{Freudweight}w(x)=\exp\{-Q(x)\}, \end{equation}
where $Q(x)$ is an even degree polynomial with positive leading coefficient. A more general proof of Freud's conjecture of the recursion coefficients for exponential weights is due to Lubinsky, Mhaskar, and Saff \cite{refLubinskyMS}; see also \cite{refDamelin,refFreud1971,refFreud1976,refNevai1986}. Deift \textit{et al}.\ \cite{refDKMVZ} discuss the asymptotics of orthogonal polynomials with respect to the weight \eqref{Freudweight} using a Riemann-Hilbert approach.

Bauldry, M\'{a}t\'{e}, and Nevai \cite{refBauldryMN88} showed that the convergent solutions of a system of smooth recurrence equations, whose Jacobian matrix satisfies a certain non-unimodularity condition, can be approximated by asymptotic expansions and they provide an application to approximate the recurrence coefficients associated with polynomials orthogonal with respect to the weight \eqref{Freudweight}, where $Q(x)$ is an even degree polynomial with positive leading coefficient.
Further, Bauldry and Zaslavsky obtained asymptotic expansions for the recurrence coefficients of a larger class
of orthogonal polynomials with exponential-type weights, cf.~\cite[Theorem 1, p.\ 496]{refMNZ85} and \cite[Theorem 5.1, p.\ 223]{refBauldryMN88}.

In a more general setting, a function of the form \eqref{Freudweight}
is called a \textit{Freud-type weight} if $Q(x)$ is an even, non-negative and continuous real valued function defined on the real line that satisfies certain conditions involving its derivatives of first and second order. Orthogonal polynomials with Freud-type weights as well as generalisations of the weight \eqref{Freudweight} in the form
\begin{equation} \label{Freudtype} w_{\ga}(x)=|x|^{\ga}\exp\{-Q(x)\}, \end{equation}
for $\ga >-1$, were studied by Levin and Lubinsky \cite{refLevinLubinsky01}. 
Lubinsky \cite{refLubinsky00}, see also \cite{refLubinsky87,refLubinsky93}, explored various types of asymptotics for polynomials orthogonal on finite and unbounded intervals, which includes a special treatment of polynomials in the Freud class. Levin and Lubinsky \cite{refLevinLubinsky01, refLevinLubinsky05} obtained many interesting properties of polynomials orthogonal with respect to the weight function \eqref{Freudtype} 
on the interval $[0,k)$, where $k \leq \infty$, including infinite-finite range inequalities, estimates for the Christoffel function, estimates for the largest zero, estimates for the spacing between zeros, estimates for the weighted orthogonal polynomials and estimates for the derivatives of the orthogonal polynomials.

Kasuga and Sakai \cite{refKasugaSakai} also considered generalized Freud-type weights of the form \eqref{Freudtype}.
Their results are similar to those for the Freud weight \eqref{Freudweight} 
obtained by Levin and Lubinsky \cite{refLevinLubinsky01}. They also showed that the zeros of the generalized Freud polynomials can be used to construct higher order Hermite-Fejer interpolation polynomials, which have their own applications in approximation theory \cite{refKasugaSakai2005}. Damelin \cite{refDamelin} used Freud equations to obtain the main term in the asymptotic expansion of the recurrence coefficients associated with orthogonal polynomials with respect to the weight \eqref{Freudtype}. The asymptotics of zeros of polynomials orthogonal with respect to the weight \eqref{Freudtype} were also derived by Kriecherbauer and McLaughlin \cite{refKMcLaughlin}. Wong and Zhang \cite{refWongZhang10} discussed the asymptotics of polynomials orthogonal with respect to the weight \eqref{Freudtype} when $Q(x)$ is an even polynomial of degree $2m$.
Using the results of Kriecherbauer and McLaughlin \cite{refKMcLaughlin}, Alfaro \textit{et al.} \cite{refAMPR} derived Mehler-Heine type asymptotic formulae for orthonormal polynomials with respect to the weight
\begin {equation}
w_{\a,m}(x)=x^{2m}\exp\left(-2|x|^\a\right),
\end{equation}
for $m\in\Z^{+}$ and $\a>1$.

Bleher and Its \cite{refBleherIts99,refBleherIts03} found several asymptotic results for semiclassical orthogonal polynomials 
with respect to the weight
\begin{equation} w(x)=\exp\{-NV(x)\},\end{equation}
where $V(x)=\tfrac14gx^4+\tfrac12tx^2$, with $g$ and $t$ parameters,
via a Riemann-Hilbert approach and applied these to prove the universality of the local distribution of eigenvalues in the matrix model with the double-well quartic interaction in the presence of two cuts, see also Wong and Zhang \cite{refWongZhang09}. 

Magnus \cite{refMagnus86} discussed the nonlinear discrete equation satisfied by the recurrence coefficients in the three-term recurrence relations for polynomials orthogonal with respect to exponential weights \eqref{Freudweight} 
and he found the relation of such equations to discrete equations for potentials such as
$Q(x) = x^4$ and $Q(x) = x^6$. Magnus \cite{refMagnus95} showed that the coefficients in the three-term recurrence relation for the Freud weight \cite{refFreud1976}
\begin{equation}\ww{x}=\exp\left(-x^4+tx^2\right),\qquad x\in\R,\label{Freud}\end{equation} with $t\in\R$ a parameter,
can be expressed in terms of simultaneous solutions, $q_n$, of the discrete equation
\begin{equation}\label{eq:dPIf}
q_n(q_{n-1}+q_n+q_{n+1})+2tq_n=n,
\end{equation}
which is {discrete \PI}\ (\dPI) -- see equation \eqref{nonlineardiffer} below for a more general version -- as earlier shown by Bonan and Nevai \cite[p.\ 135]{refBNevai}, and \comment{the differential equation
\begin{equation}\label{eq:PIVf}
\deriv[2]{q_n}{z}= \frac{1}{2q_n}\left(\deriv{q_n}{z}\right)^{2} + \frac{3}{2}q_n^3 + 4z q_n^2 + 2(z^2 +\tfrac12n)q_n- \frac{n^2}{2q_n},
\end{equation}
with $n\in\Z^+$, which is a special case of}%
the fourth \p\ equation (\PIV)
\begin{equation}\label{eq:PIV}
\deriv[2]{q}{z}= \frac{1}{2q}\left(\deriv{q}{z}\right)^{2} + \frac{3}{2}q^3 + 4z q^2 + 2(z^2 - A)q + \frac{B}{q},\end{equation}
where $A=-\tfrac12n$ and $B=-\tfrac12n^2$, with $n\in\Z^+$.
This connection between the recurrence coefficients for the Freud weight (\ref{Freud}) and simultaneous solutions of (\ref{eq:dPIf}) and (\ref{eq:PIV}) is due to Kitaev, see also \cite{refFIKa,refFIZ}.
Subsequently, this relation was studied by Bassom, Clarkson, and Hicks \cite{refBCH95}, who wrote tables of simultaneous solutions of \PIV\ \eqref{eq:PIV} and \dPI\ \eqref{eq:dPIf} in terms of parabolic cylinder functions, and later by Grammaticos and Ramani \cite{refGR98}. The relationship between solutions of \PIV\ \eqref{eq:PIV} and \dPI\ \eqref{eq:dPIf} is reflected in the striking similarity of the results for \PIV\ \eqref{eq:PIV} in \cite{refBCH95,refMurata85,refOkamotoPIIPIV} and those for
\dPI\ \eqref{eq:dPIf} in \cite{refGR98}. 
Bonan and Nevai \cite{refBNevai} proved that there is a unique positive solution of the discrete equation \eqref{eq:dPIf} with initial conditions
\[ \b_0=0,\qquad \b_{1} = \frac{\int_{-\infty}^\infty x^2 \exp(-x^4+tx^2)\, dx}{\int_{-\infty}^\infty \exp(-x^4+tx^2)\, dx}.\]

In \cite{refCJK}, we considered the \textit{generalized Freud weight} \comment{inner product
\begin{align*}\label{innerGF}
\langle p,q\rangle
 = \int_{-\infty}^{\infty} p(x) q(x) |x|^{2\lambda +1} \exp(-x^4+tx^2)\,d{x}, \quad t\in \R,
\end{align*}
and some of the properties of the orthogonal polynomials associated with this inner product.
Consider the positive Borel measure $\,d{\mu}(x)=w_{\la}(x;t)\,d{x} $ where}
\begin{equation}\label{genFreud}
w_{\la}(x;t)=|x|^{2\la +1}\exp(-x^4+tx^2),\qquad x\in\R,
\end{equation}
with $\la>-1$ and $t\in \R$ parameters and gave explicit expressions for the moments of this weight \eqref{genFreud}.
The first moment is
\begin{align}\label{mu0} \mu_0(t;\la)&= \int_{-\infty}^{\infty} |x|^{2\la +1} \exp \left( -x^4 +tx^2\right) \,d{x}\nonumber \\
&= \frac{\Gamma(\la+1)}{2^{(\la+1)/2}}\,\exp\left(\tfrac18t^2\right)\WhitD{-\la-1}\big(-\tfrac12\sqrt2\,t\big),\end{align}
where $\WhitD{v}(\xi)$ is the parabolic cylinder function with integral representation, cf.~\cite[\S12.5(i)]{refNIST}
\[\WhitD{\nu}(\xi) = \frac{\exp(-\tfrac{1}{4}\xi^2)}{\Gamma(-\nu)} \int_{0}^{\infty} s^{-\nu-1} \, \exp\left( -\tfrac{1}{2}s^2 -\xi s \right) d{s},\qquad \mathop{\rm Re}(\nu)<0,\]
and the higher moments are
\begin{align*}
\mu_{2n}(t;\la) &= \int_{-\infty}^{\infty} x^{2n} |x|^{2\la +1} \exp \left( -x^4 +tx^2\right) \,d{x} 
\equiv\mu_0(t;\la+n), \\ 
\mu_{2n-1}(t;\la) &= \int_{-\infty}^{\infty} x^{2n-1} |x|^{2\la +1} \exp \left( -x^4 +tx^2\right) \,d{x} =0,
\end{align*}
for $n\geq1$.
\comment{The polynomials $ p_n(x;t)= k_nx^n + \hdots$, $n\in \N$ for fixed parameter $t\in \R$ are orthogonal with respect to \eqref{genFreud} if
\begin{align*}\label{innerortho}
\langle p_n(x;t), p_m(x;t)\rangle= \delta_{mn},
\end{align*} where $\delta_{mn}$ is the Kronecker delta. In this case the polynomials satisfy the three term recurrence relation
\begin{equation}\label{orthonormalrecur}
xp_n(x;t) = b_{n+1}(t;\la)~p_{n+1}(x;t) + b_{n}(t;\la)~p_{n-1}(x;t), \qquad n=0,1,2,\hdots
\end{equation}
with the recurrence coefficients given by
 $\ds{b_n(t;\la) = \frac{k_{n-1}}{k_{n}}}$, $b_0=p_{-1}(x;t)=0$ and $p_0(x;t) = \left(\ds \int_{-\infty}^{\infty} w_{\lambda}(x;t)\,d{x} \right)^{-\frac{1}{2}}.$}%
The weight function \eqref{genFreud} is (weakly) differentiable on the non-compact support $\R$ and satisfies the distributional equation, known as Pearson equation (see \cite{refVanAssche07}),
\begin{align}\nonumber 
\deriv{}{x}\ln{w_{\la}(x;t)} = \frac{B(x)-A'(x)}{A(x)}=-4x^{3}+2tx + \frac{2\la+1}{x}, \end{align}
with $A(x)$ and $B(x)$ polynomials of minimal degree, so
\[A(x)=x, \quad B(x)= -4x^{4} +2tx^{2} +2\lambda +3. \]
Since $\text{deg}(A)=1$ and $\text{deg}(B)=4$, the polynomial sequence $\{ S_n(x;t)\}_{n=0}^{\infty}$, representing the sequence of monic polynomials orthogonal with respect to \eqref{genFreud}, is said to constitute a family of \textit{semiclassical orthogonal polynomials}\ \cite{refCJ,refHendriksenvR,refIsmail,refIsmail00,refMaroni}.

Monic orthogonal polynomials with respect to the symmetric weight \eqref{genFreud} satisfy the three-term recurrence relation
\begin{equation}xS_n(x;t)=S_{n+1}(x;t)+\b_{n}(t;\la)S_{n-1}(x;t)\label{eq:gfrr}\end{equation} where $\b_{n}(t;\la)>0$, $S_{-1}(x;t)=0$, $S_0(x;t)=1$, $\b_0(t;\la)=0$ and
\begin{align}
\b_1(t;\la)= \frac{\mu_2(t;\la)}{\mu_0(t;\la)}
&=\frac{\int_{-\infty}^{\infty} x^2 |x|^{2\la +1} \exp \left( -x^4 +tx^2\right) d{x}}{\int_{-\infty}^{\infty} |x|^{2\la +1} \exp \left( -x^4 +tx^2\right) d{x}}\nonumber\\
&=\tfrac12t+\tfrac12\sqrt{2}\,\frac{\WhitD{-\la}\big(-\tfrac12\sqrt2\,t\big)}{\WhitD{-\la-1}\big(-\tfrac12\sqrt2\,t\big)},
\label{beta1}\end{align}
see \S2 for further details. Several sequences of monic orthogonal polynomials related to the weight \eqref{genFreud} and its extensions have been studied in the literature. For instance, for $t= 0$, $\la=-\frac 12$, the asymptotic and analytic properties of the corresponding orthogonal polynomials were studied in \cite{refNevai1983}, while the case when $t>0$ and $\la=-\tfrac12$ is discussed in \cite{refAHM2016}.

The recurrence coefficients in the three-term recurrence relations associated with semiclassical orthogonal polynomials can often be expressed in terms of solutions of the \p\ equations and associated discrete \p\ equations. As shown in \cite{refCJK}, the recurrence coefficients $\b_n(t;\la)$ in the three term recurrence relation \eqref{eq:gfrr} are related to solutions of \PIV\ \eqref{eq:PIV} and satisfy the equation
\begin{equation}\label{eq:bn}
\deriv[2]{\b_n}{t}=\frac{1}{2\b_n}\left(\deriv{\b_n}{t}\right)^2+\tfrac32\b_n^3-t\b_n^2+(\tfrac18 t^2-\tfrac12A_n)\b_n+\frac{B_n}{16\b_n},\end{equation}
where the parameters $A_n$ and $B_n$ are given by
\begin{equation*}
\begin{pmatrix}A_{2n}\\ B_{2n}\end{pmatrix} =\begin{pmatrix} -2\la-n-1\\ -2n^2\end{pmatrix} , \qquad
\begin{pmatrix}A_{2n+1}\\ B_{2n+1}\end{pmatrix} =\begin{pmatrix}\la-n\\ -2(\la+n+1)^2\end{pmatrix}
\end{equation*}
as well as the nonlinear discrete equation
\begin{equation}\label{nonlineardiffer}
 4\b_n \left( \b_{n-1} + \b_{n} + \b_{n+1} -\tfrac12{t} \right)= n+ (2\lambda +1)\Delta_n,
 \end{equation}
where $\Delta_{n}:= \tfrac12[1-(-1)^{n}]$, 
which is the general discrete \PI\ (\dPI). We remark that the nonlinear discrete equation \eqref{nonlineardiffer} appears in the paper by Freud \cite[equation (23), p.\ 5]{refFreud1976}; see also \cite[\S2]{refANSVA15} for a historical review of the origin and study of equation \eqref{nonlineardiffer}.

The moments of certain semiclassical weights provide the link between the weight and the associated \p\ equation as shown in \cite{refCJ}. In \cite{refCJK} this was used to obtain the explicit expressions for the recurrence coefficients $\b_n(t;\la)$ in the three term recurrence relation \eqref{eq:gfrr} given by
\begin{subequations}\label{betathreeAAab}
\begin{align}
\b_{2n}(t;\la)
&=\deriv{}{t}\ln\frac{\tau_n(t;\la+1)}{\tau_n(t;\la)},\\
\b_{2n+1}(t;\la) &=\deriv{}{t}\ln\frac{\tau_{n+1}(t;\la)}{\tau_n(t;\la+1)},
\end{align}\end{subequations}
for $n\geq0$,
where $\tau_n(t;\la)$ is the 
Hankel determinant given by
\begin{align*}
\tau_n(t;\la)&
=\det\left[\deriv[j+k]{}{t}\mu_0(t;\la)\right]_{j,k=0}^{n-1},
\end{align*}
with $\tau_0(t;\la)=1$ and $\mu_0(t;\la)$ given by \eqref{mu0}.
\comment{$\phi_{\la}(t)$ given by
\begin{equation}\label{phin}
\phi_{\la}(t)=\mu_0(t;\la)=\frac{\Gamma(\la+1)}{2^{(\la+1)/2}}\,\exp\left(\tfrac18t^2\right)\WhitD{-\la-1}\big(-\tfrac12\sqrt2\,t\big).
\end{equation}}

Following our earlier work in \cite{refCJK}, here we are concerned with the asymptotic behaviour of the recurrence coefficient of the three-term recurrence relation satisfied by the generalized Freud polynomials and the asymptotic properties of the polynomials themselves. We review some pertinent results from \cite{refCJK} in \S2. In \S\ref{sec3} we consider the case where the parameter $t\to\infty$ whilst in \S\ref{sec4} we investigate the asymptotic behaviour as the degree $n$ of the polynomials tends to $\infty$. Existence and uniqueness of positive solutions of the nonlinear discrete equation \eqref{nonlineardiffer} are discussed in \S\ref{sec5} where we prove that unique, positive solutions exist for all $t\in\R$. Properties of the zeros of generalized Freud polynomials are investigated in \S\ref{sec6}.

\section{Generalized Freud polynomials}
The first few recurrence coefficients $\b_n(t;\la)$ are given by
\begin {subequations}\label{def:betan}\begin{align}
\b_1(t;\la)&=\Ph(t),\label{def:beta1}\\
\b_2(t;\la)&=\tfrac12t -\Ph(t)+\frac{\la+1}{2\Ph(t)},\\
\b_3(t;\la)&=-\frac{\la+1}{2\Ph(t)}-\frac{\Ph(t)}{2\Ph^2(t)-t\Ph(t)-\la-1},\\
\b_4(t;\la)&=\frac{t}{2(\la+2)}+\frac{\Ph(t)}{2\Ph^2(t)-t\Ph(t)-\la-1} \nonumber \\ &\qquad
+\frac{(\la+1)\big[(t^2+2\la+4)\Ph(t)+(\la+1)t\big]}{2(\la+2)\big[2(\la+2)\Ph^2(t)-(\la+1) t\Ph(t)-(\la+1)^2\big]},\comment{\\
\b_5(t;\la)&=\frac{(\la+1)t}{2(\la+2)}+\frac{2(\la+2)}{t}
-\frac{(\la+1)\big[(t^2+2\la+4)\Ph(t)+(\la+1)t\big]}{2(\la+2)\big[2(\la+2)\Ph^2(t)-(\la+1)t\Ph(t)-(\la+1)^2\big]} \\
&\qquad -\frac{2\big[(\la+1)t^2+4(\la+2)(2\la+3)\big]\Ph^2(t)-(\la+1)(t^2+8\la+18)t\Ph(t)-(\la+1)^2(t^2+8\la+16)}{2\big[2t^2\Ph^3(t)-(t^2-4\la-6)t\Ph^2(t)-3(\la+1)t^2\Ph(t)-2(\la+1)^2t\big]},\nonumber}
\end{align}\end{subequations}
where
\begin{align}\label{def:Phi}
\Ph(t)&=\deriv{}{t}\ln\left\{\WhitD{-\la-1}\big(-\tfrac12\sqrt2\,t\big)\exp\left(\tfrac18t^2\right)\right\}\nonumber\\
&=\tfrac12t+\tfrac12\sqrt{2}\,\frac{\WhitD{-\la}\big(-\tfrac12\sqrt2\,t\big)}{\WhitD{-\la-1}\big(-\tfrac12\sqrt2\,t\big)}.
\end{align}
It was shown in \cite{refCJK} that as $t\to \infty$
\begin{equation*}
\Ph(t)=\frac{t}{2}+\frac{\la}{t} +\frac{2\la(1-\la)}{t^3}+\frac{4\la(\la-1)(2\la-3)}{t^5}+\mathcal{O}\big(t^{-7}\big).
\end{equation*}
Hence, as $t\to \infty$
\begin{equation*}
\frac{1}{\Ph(t)}=\frac{2}{t} -\frac{4\la}{t^3}+\frac{8\la(2\la-1)}{t^5}+\mathcal{O}\big(t^{-7}\big).
\end{equation*}
Plots of $\b_n(t;\la)$, for $n=1,2,\ldots,10$, with $\la=\tfrac12$ are given in Figure~\ref{fig:betan}. We see that there is completely different behaviour for $\b_n(t;\la)$ as $t\to\infty$, depending on whether $n$ is even or odd, which is reflected in Lemma~\ref{betaasymp}. The different behaviour for $\b_n(t;\la)$ depending on whether $n$ is even or odd can be explained by the fact that
$\b_{2n}(t;\la)$ and $\b_{2n+1}(t;\la)$ satisfy different explicit expressions \eqref{betathreeAAab}, as well as different
differential equations \eqref{eq:bn} and difference equations \eqref{nonlineardiffer}.

\def\fig#1{\includegraphics[width=7cm]{#1}}
\begin{figure}[ht]
\[\begin{array}{c@{\quad}c}
\fig{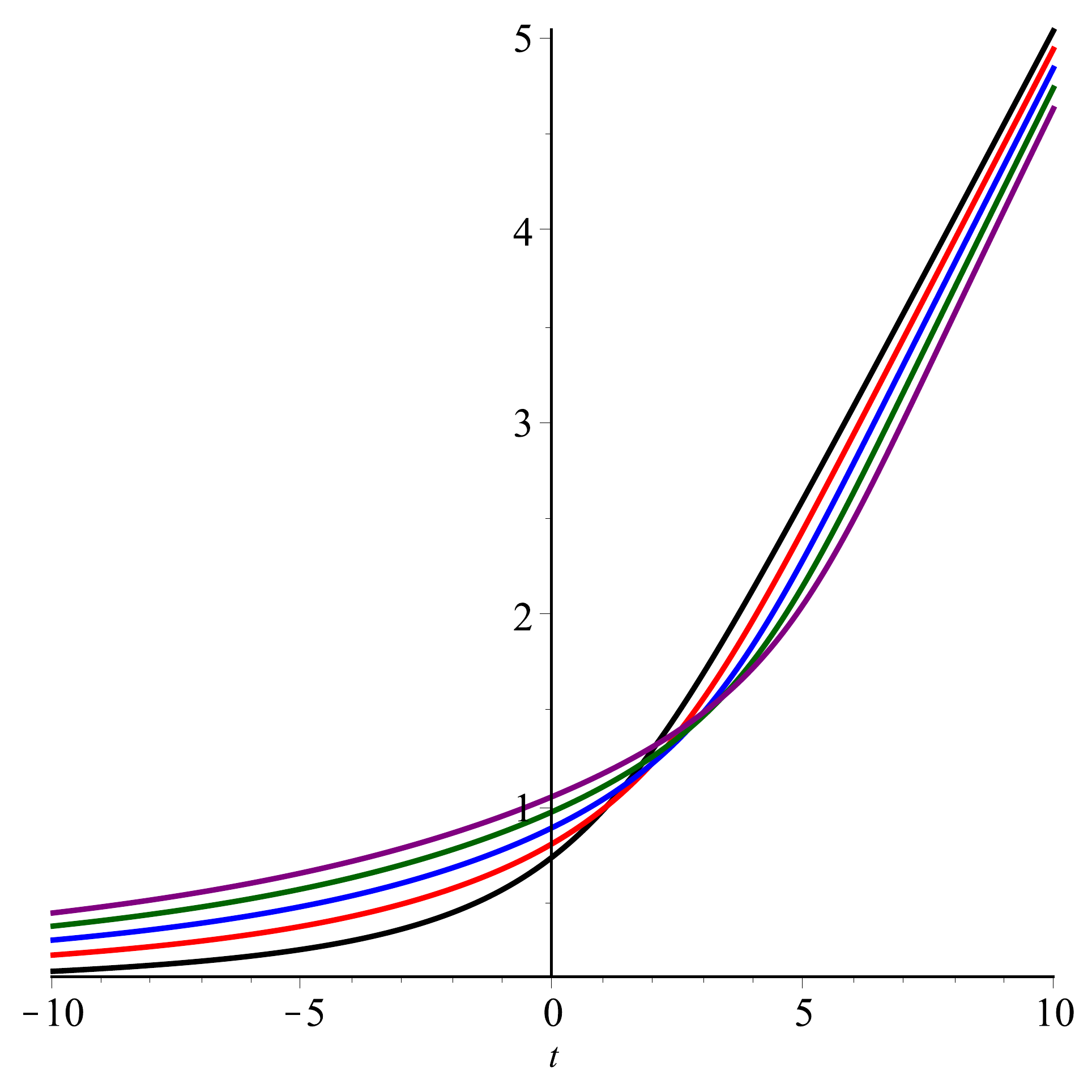} & \fig{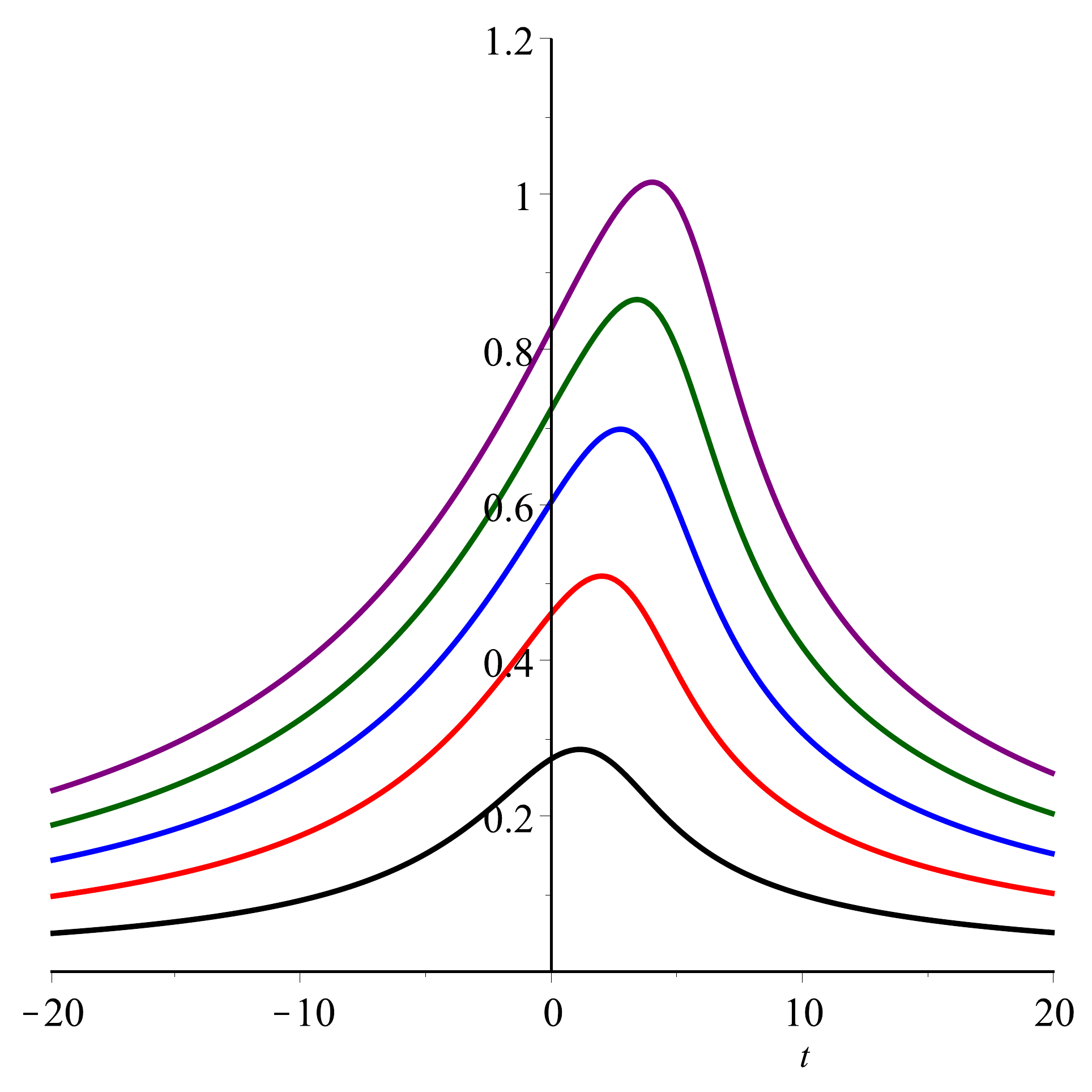}\\
\b_{2n-1}(t;\la),\quad n=1,2,\ldots,5 & \b_{2n}(t;\la),\quad n=1,2,\ldots,5
\end{array}\]
\caption{\label{fig:betan}
Plots of the recurrence coefficients $\b_{2n-1}(t;\la)$ and $\b_{2n}(t;\la)$, for $n=1$ (black), $n=2$ (red), $n=3$ (blue), $n=4$ (green) and $n=5$ (purple), with $\la=\tfrac12$.}
\end{figure}

\begin{lemma}{\label{betaasymp} As $t\to\infty$, the recurrence coefficient $\b_n(t;\la)$ has the asymptotic expansion
\begin{subequations}
\begin{align*}
\b_{2n}(t;\la)&=\frac{n}{t}-\frac{2n(2\la-n+1)}{t^3}+\mathcal{O}\big(t^{-5}\big),
\\ 
\b_{2n+1}(t;\la)&=\frac{t}{2}+\frac{\la-n}{t}-\frac{2(\la^2-4\la n+n^2-\la-n)}{t^3}+\mathcal{O}\big(t^{-5}\big),
\label{bn:asym12}\end{align*}\end{subequations}
for $n\in\N$. Further, as $t\to-\infty$
\begin{subequations}
\begin{align*}
\b_{2n}(t;\la)&=-\frac{n}{t}+\frac{2n (2\la+3n+1)}{t^3} +\mathcal{O}\big(t^{-5}\big),\\
\b_{2n+1}(t;\la)&=-\frac {\la+n+1}{t}+\frac { 2( \la +n+1)(\la+3n+2 )}{t^3} +\mathcal{O}\big(t^{-5}\big).
\end{align*}\end{subequations}
}\end{lemma}
\begin{proof}{See \cite[Lemma 8]{refCJK}}\end{proof}

Using the recurrence relation \eqref{eq:gfrr}, with $\b_n(t;\la)$ given by \eqref{def:betan}, the first few polynomials $S_n(x;t)$ are given by
\begin{subequations}
\begin{align*}
S_1(x;t)&={x},\\
S_2(x;t)&=x^2-\Ph(t),\\
S_3(x;t)&=x^3-\frac {t\Ph(t) +\la+1}{2\Ph(t)}\,x,\\
S_4(x;t)&=x^4-{\frac {2t\Ph^{2}(t)-({t}^{2}+2)\Ph(t) -(\la+1) t}{ 2\left[2\Ph^{2}(t)-t\Ph(t) -\la-1\right]}}\,x^2
-{\frac {2(\la+2)\Ph^2(t)-(\la+1) t\Ph(t)-(\la+1)^2}{2\left[2\Ph^{2}(t)-t\Ph(t) -\la-1\right]}},\comment{\\
S_5(x;t)&=x^5-{\frac {2(\la+3)t\Ph^2(t)-(\la+1) (t^2-2)\Ph-(\la+1)^2t}{2\big[2(\la+2)\Ph^2-(\la+1) t\Ph-(\la+1)^2\big]}}\,x^3\\
&\phantom{=x^5\ } -{\frac {\big[2(\la+2)^2-t^2\big]\Ph^{2}-(\la+1)(\la+4)t\Ph -(\la+1)^2(\la+3)}{2\big[2(\la+2)\Ph^2-(\la+1) t\Ph-(\la+1)^2\big]}}\,x,\nonumber}
\end{align*}\end{subequations}
with $\Ph(t)$ given by \eqref{def:Phi}.
Plots of the polynomials $S_n(x;t)$, $n=3,4,\ldots,8$, with $\la=\tfrac12$, for $t=0,1,\ldots,4$, are given in Figure~\ref{fig:Sn1}.
These show that the magnitude of the roots of $S_n(x;t)$ increases as $t$ increases (see Theorem \ref{zerofreud} for further details and a proof). In fact, as shown in \S\ref{sec3} below, the roots of $S_{2n}(x;t)$ and $S_{2n+1}(x;t)/x$ tend to $\pm(\tfrac12t)^{1/2}$ as $t\to\infty$. Plots of the polynomials $S_n(x;t)$, $n=3,4,\ldots,8$, with $\la=\tfrac12$, at times $t=0,1,\ldots,5$ are given in Figure~\ref{fig:Sn2}, which illustrates the interlacing of the roots of successive polynomials, as discussed in Theorem \ref{zerofreud}.

\begin{figure}[ht]
\[\begin{array}{c@{\quad}c}
\fig{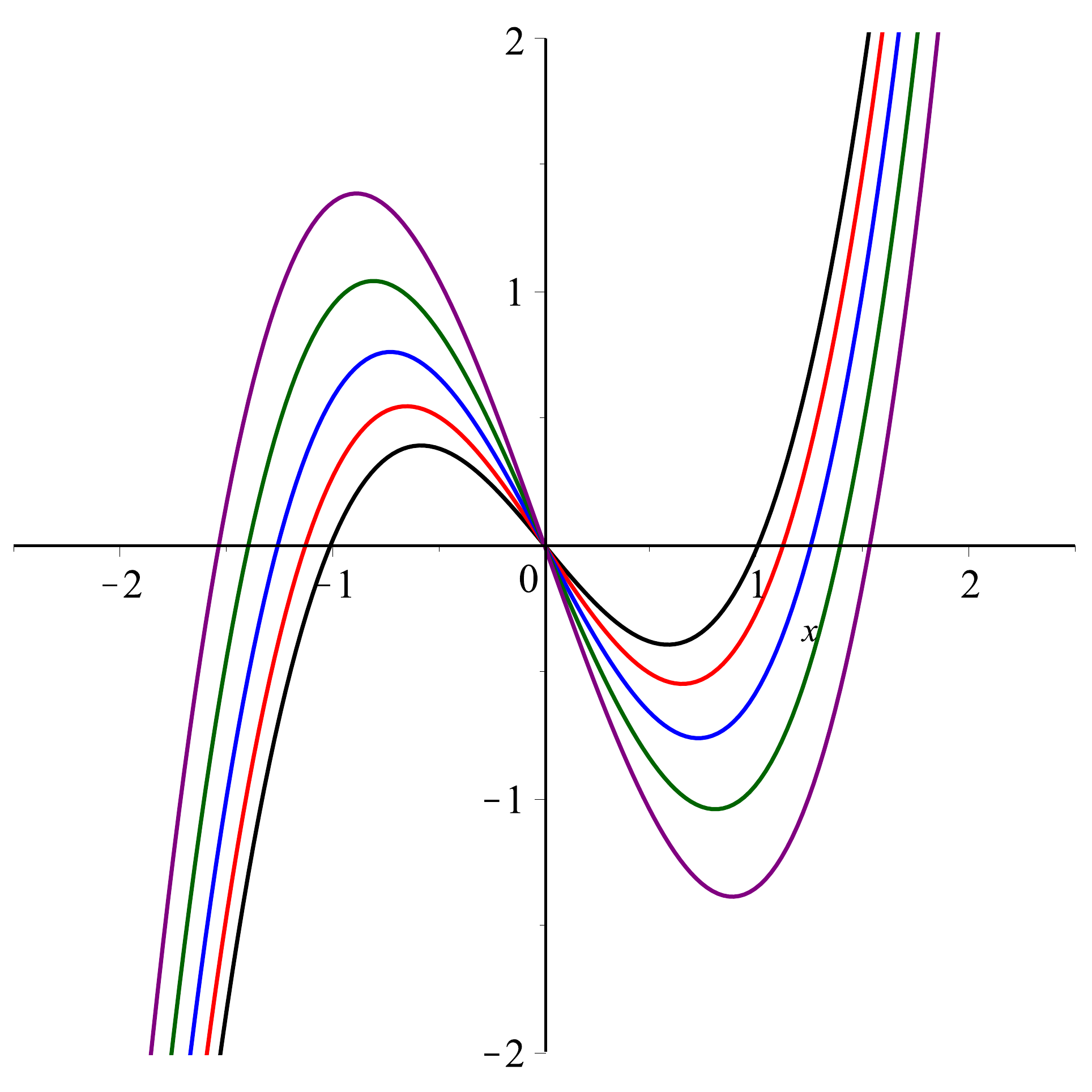} & \fig{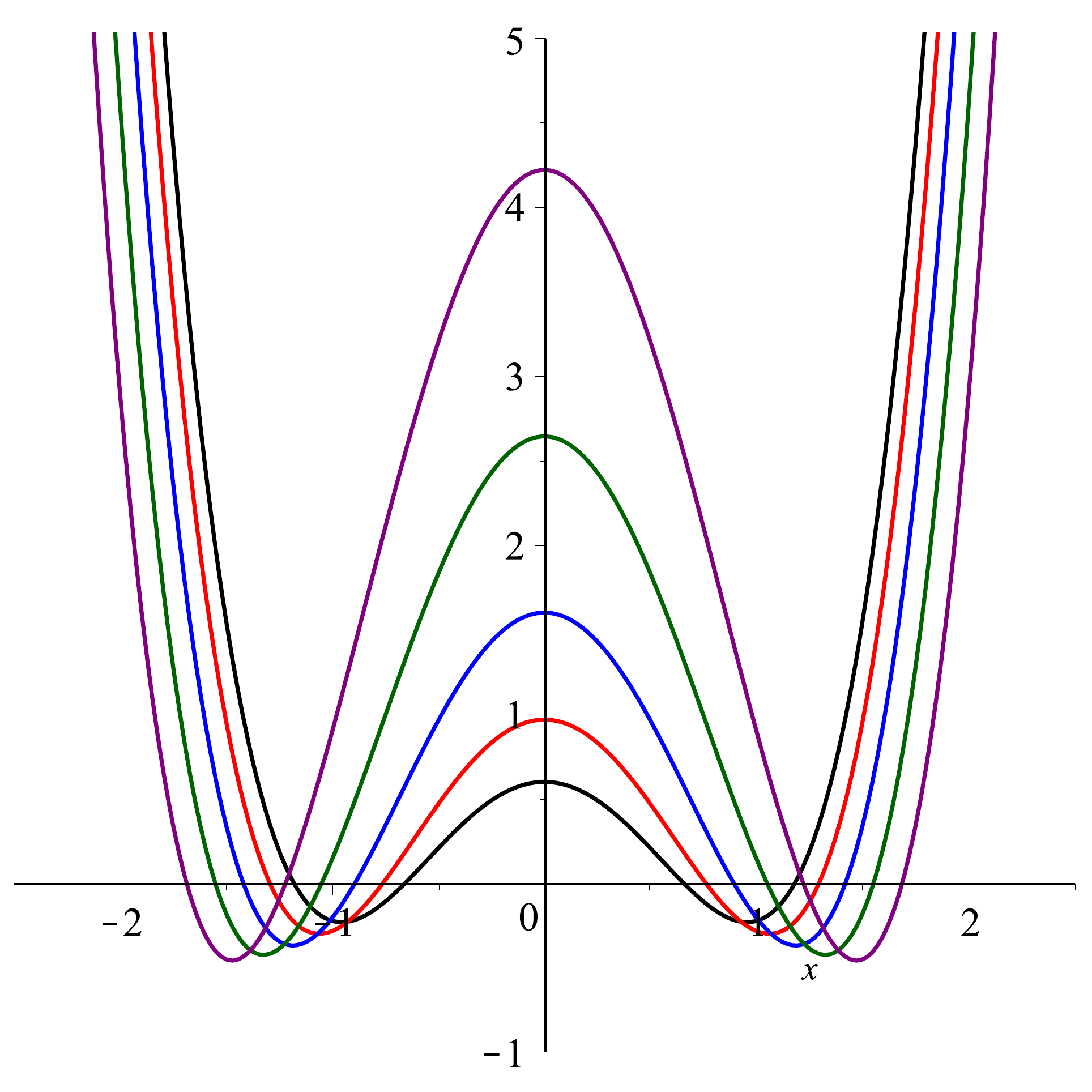} \\
S_3(x;t) & S_4(x;t)\\[5pt]
\fig{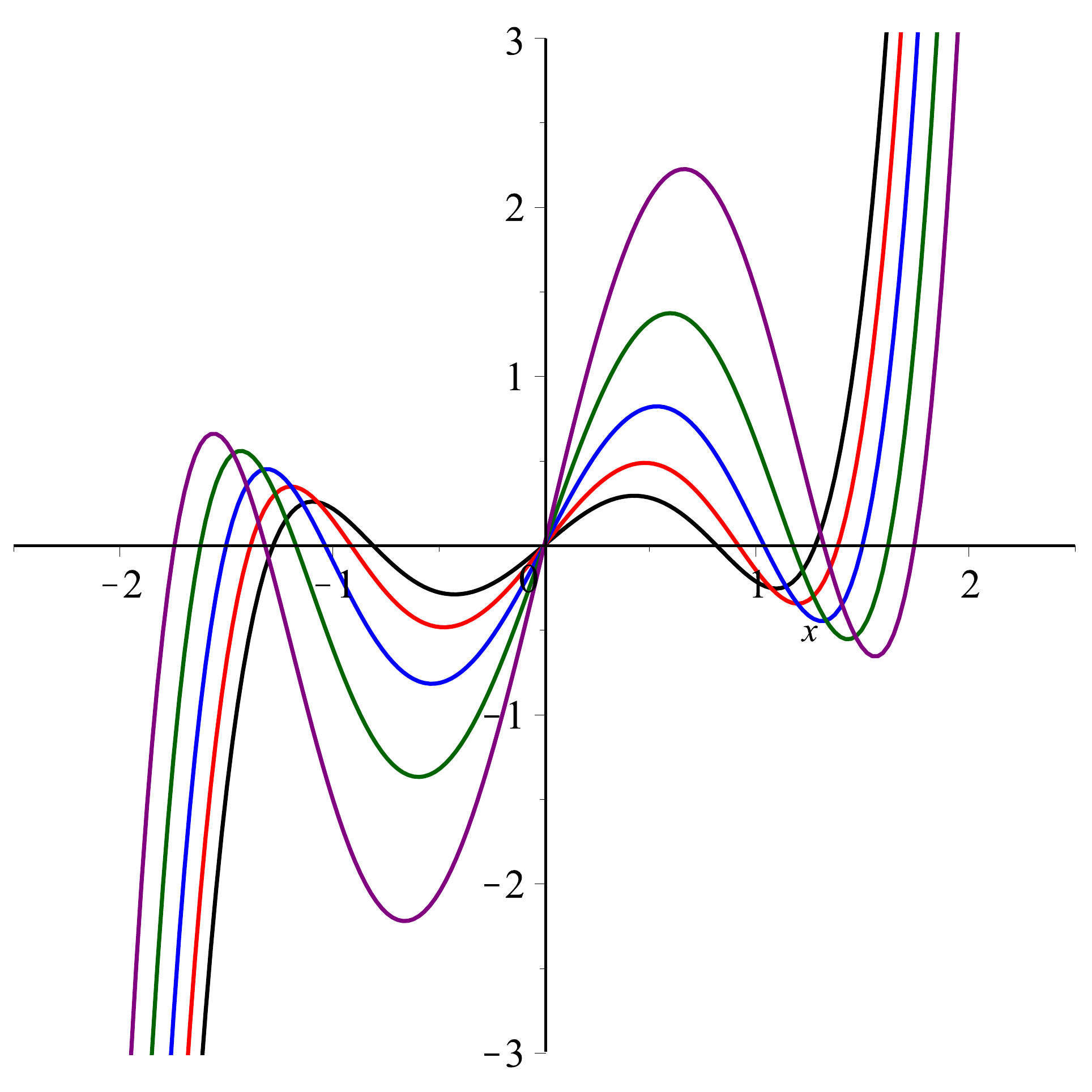} & \fig{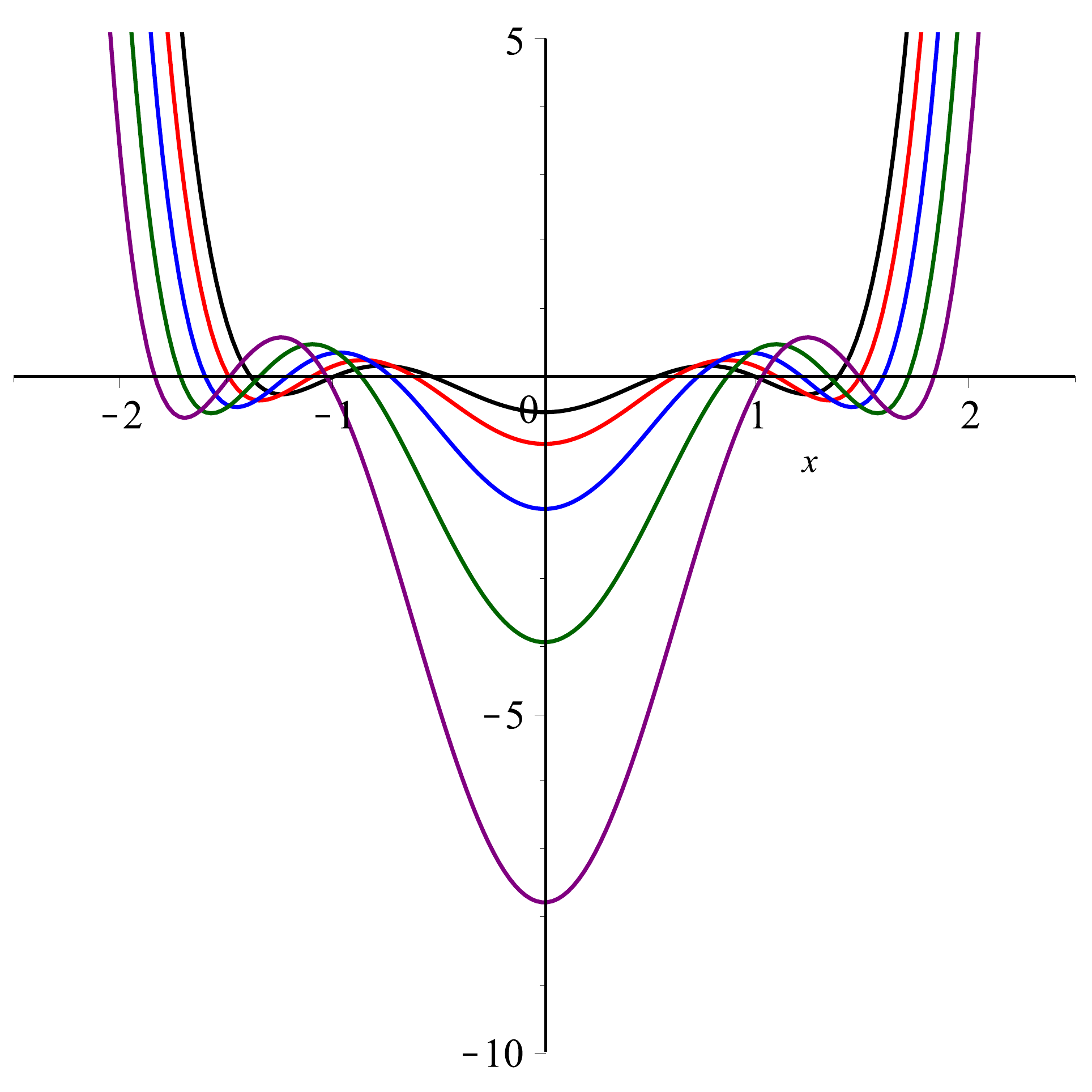} \\
S_5(x;t) & S_6(x;t)\\[5pt]
\fig{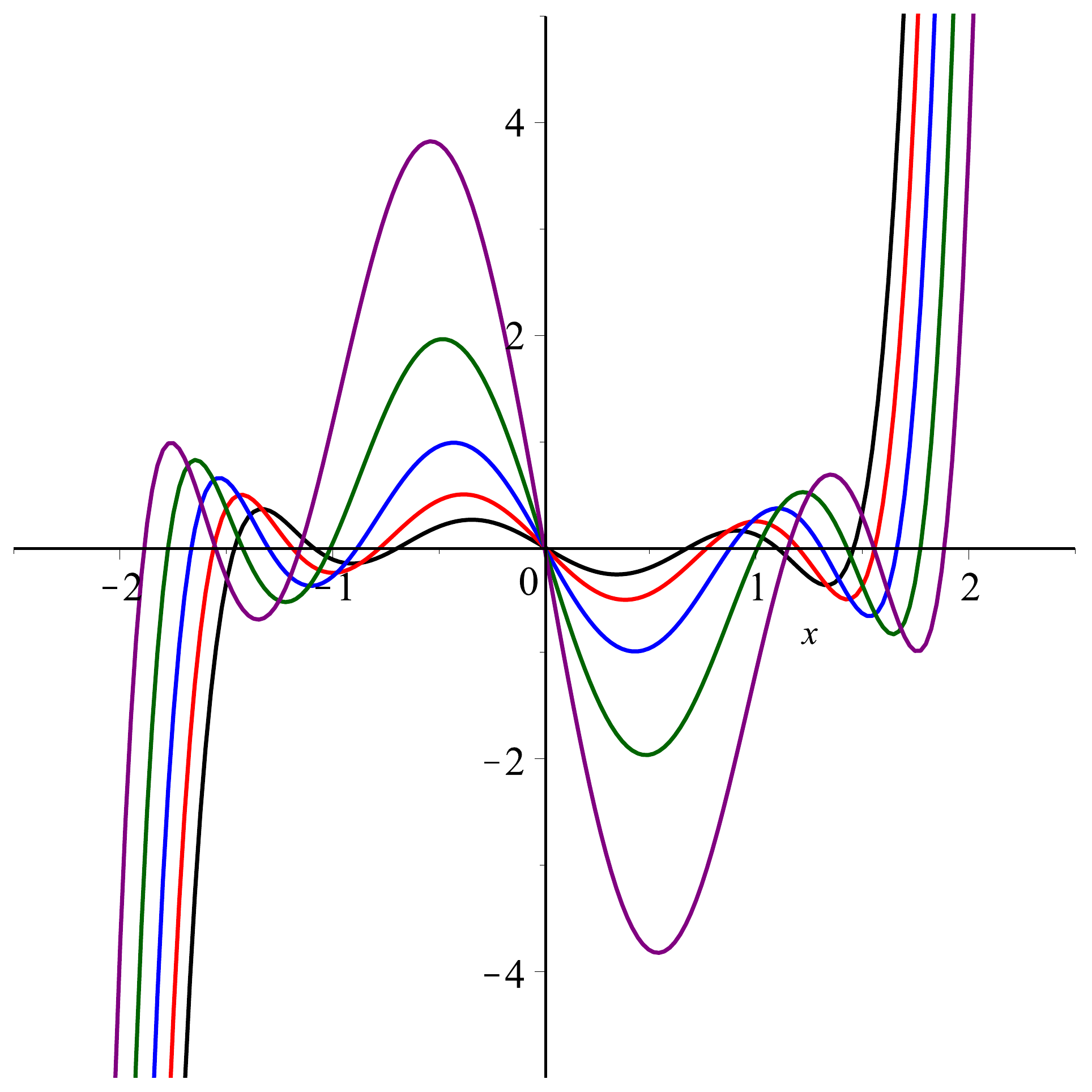} & \fig{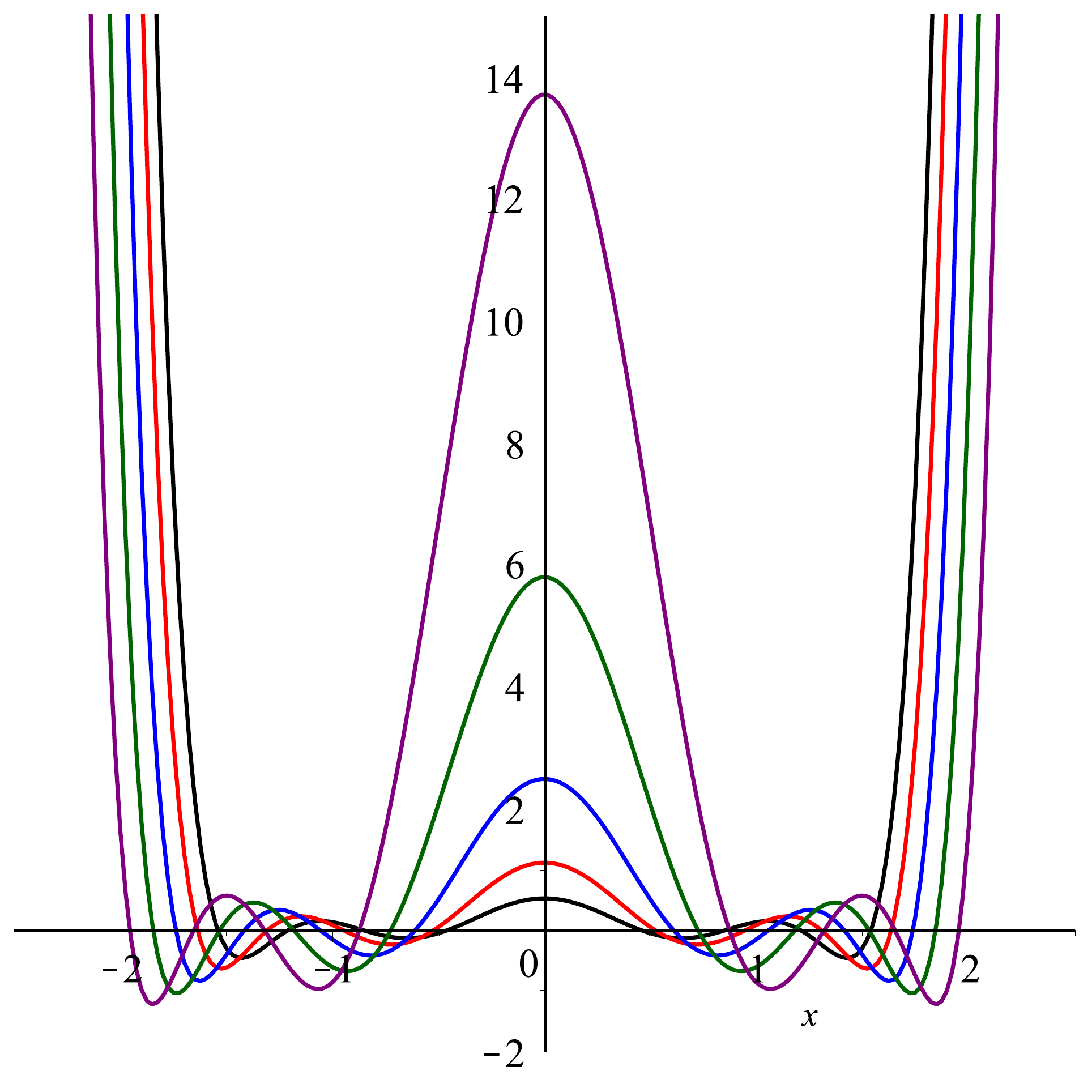}\\
S_7(x;t) & S_8(x;t)
\end{array}\]
\caption{\label{fig:Sn1}Plots of the polynomials $S_n(x;t)$, $n=3,4,\ldots,8$ for $t=0$ (black), $t=1$ (red), $t=2$ (blue), $t=3$ (green) and $t=4$ (purple), with $\la=\tfrac12$.}
\end{figure}

\begin{figure}[ht]
\[\begin{array}{c@{\quad}c}
\fig{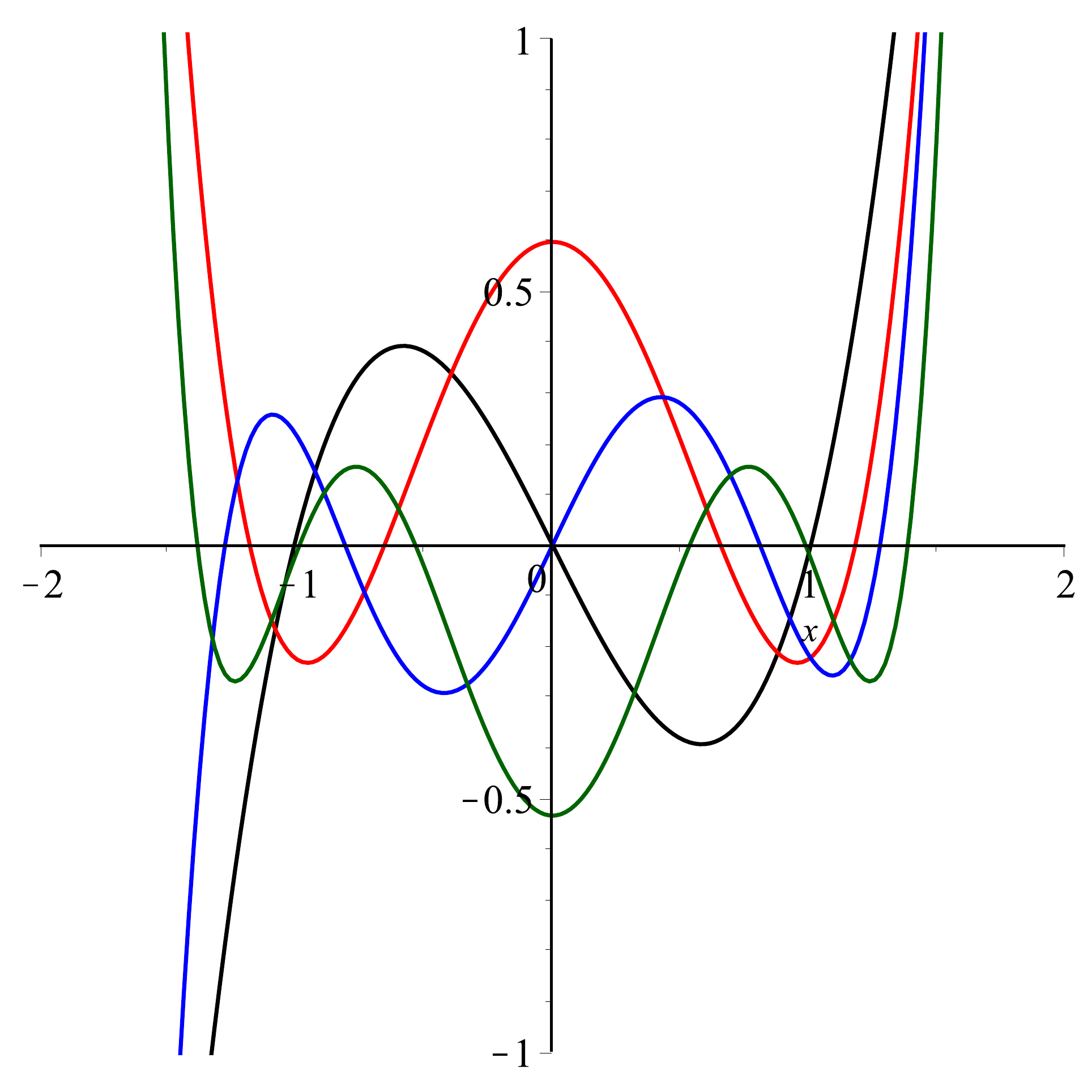}
& \fig{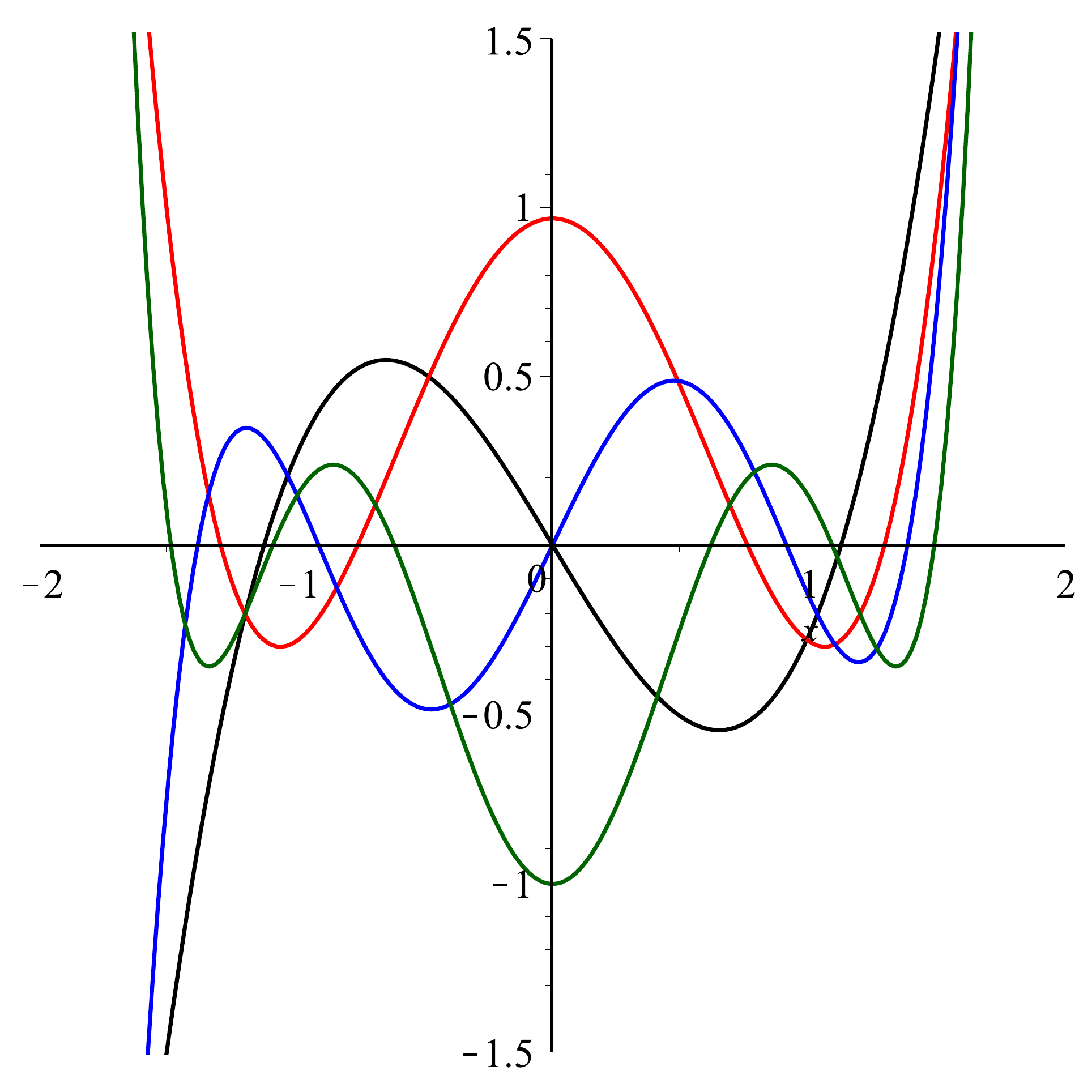} \\
t=0 & t=1 \\[5pt]
\fig{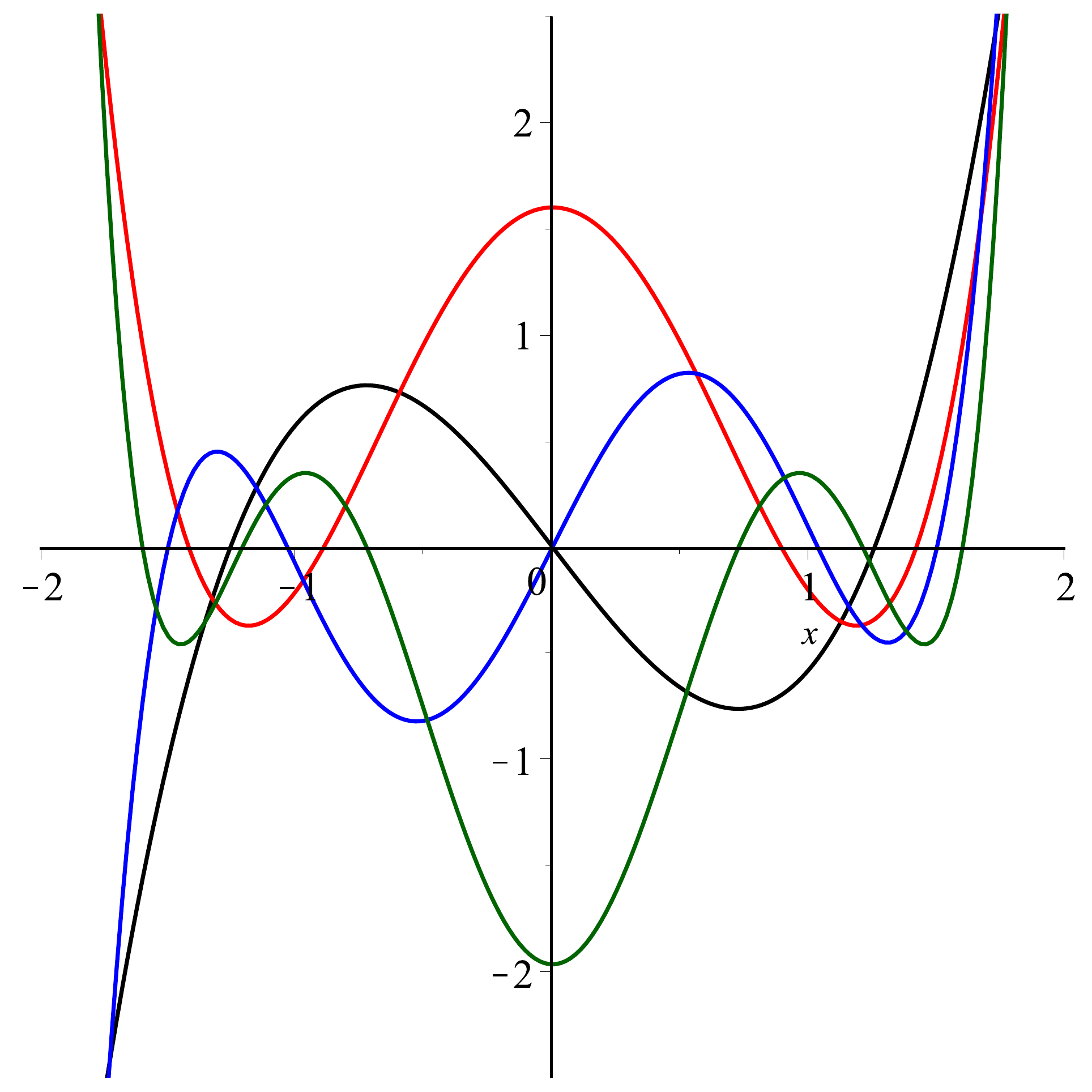}
& \fig{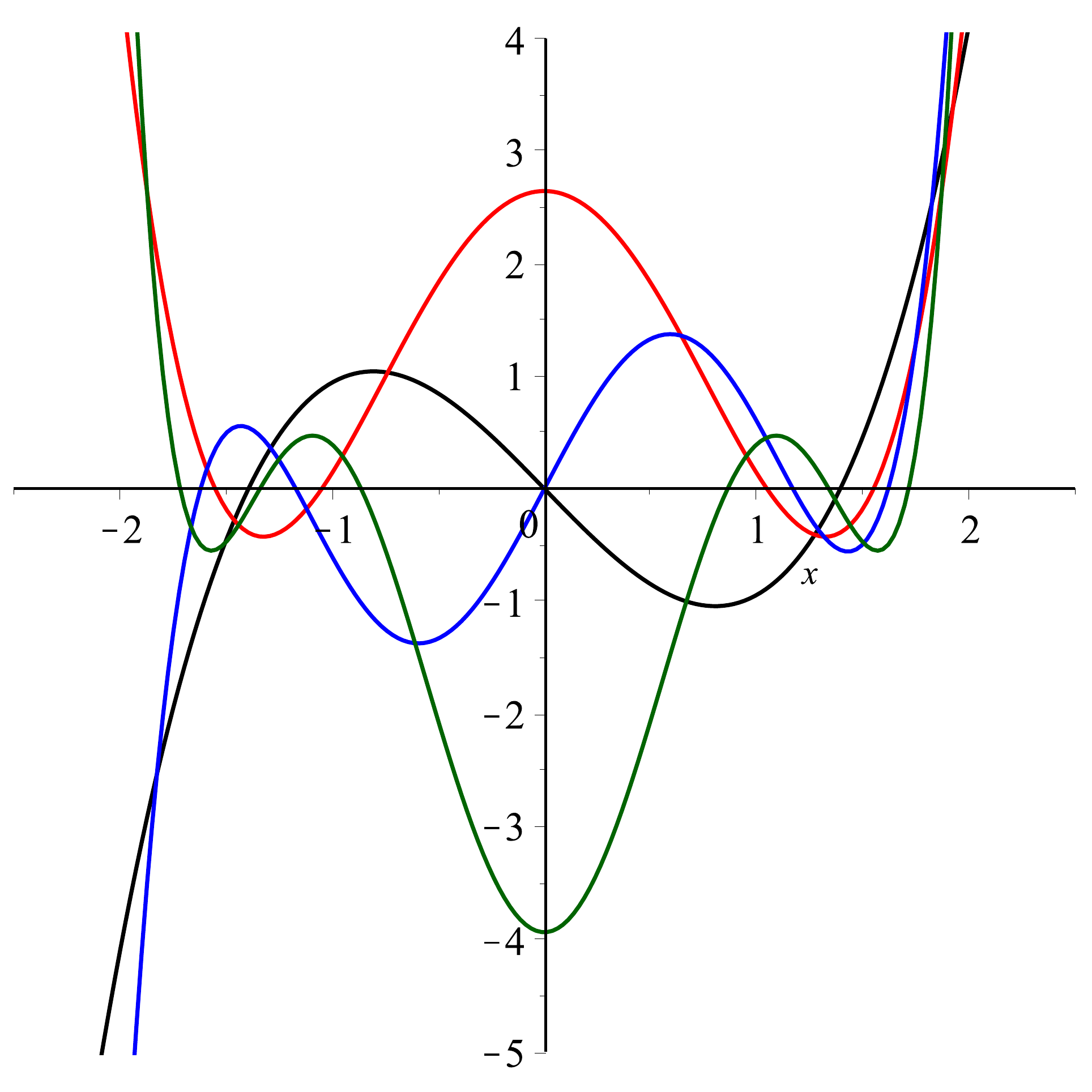} \\
t=2 & t=3 \\[5pt]
\fig{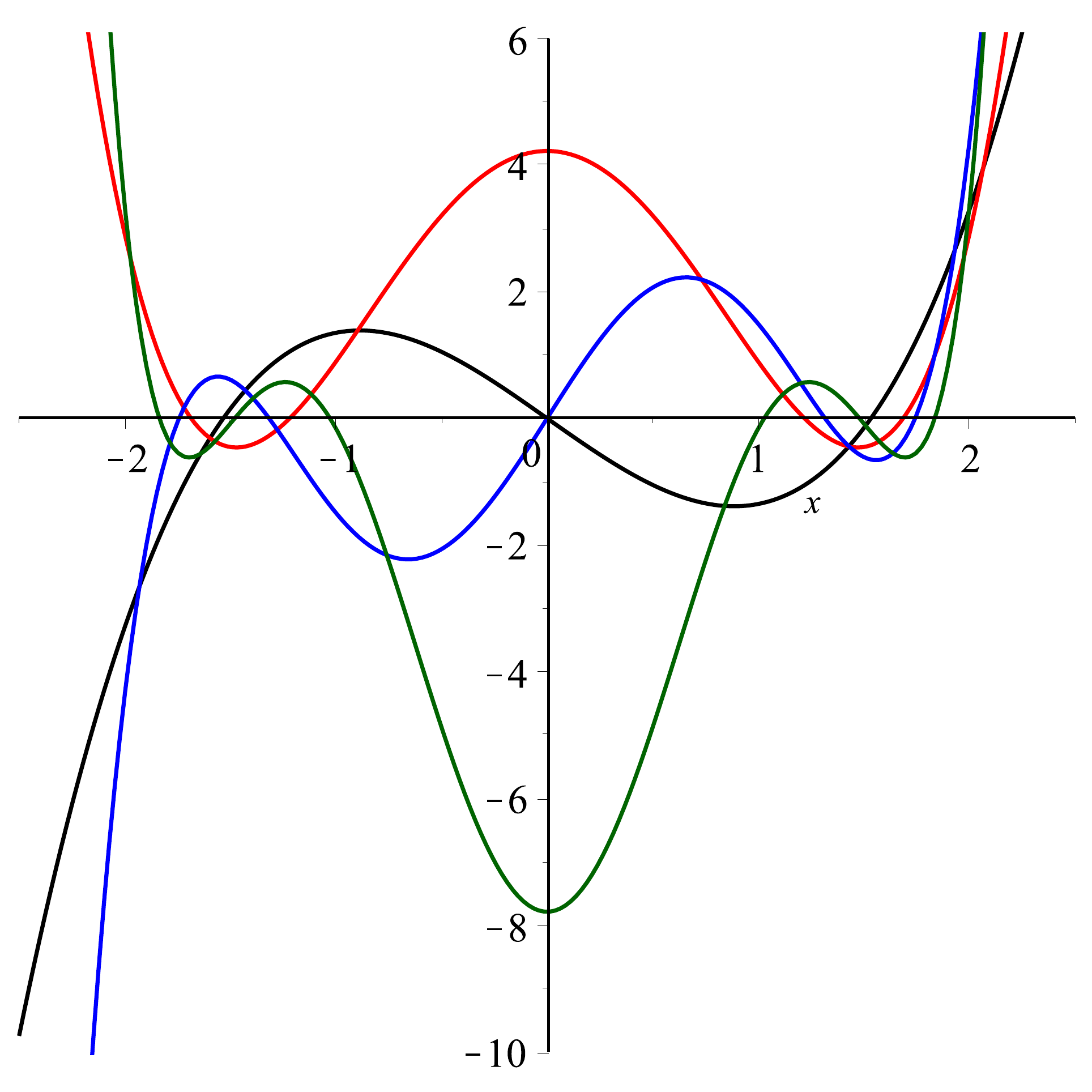}
& \fig{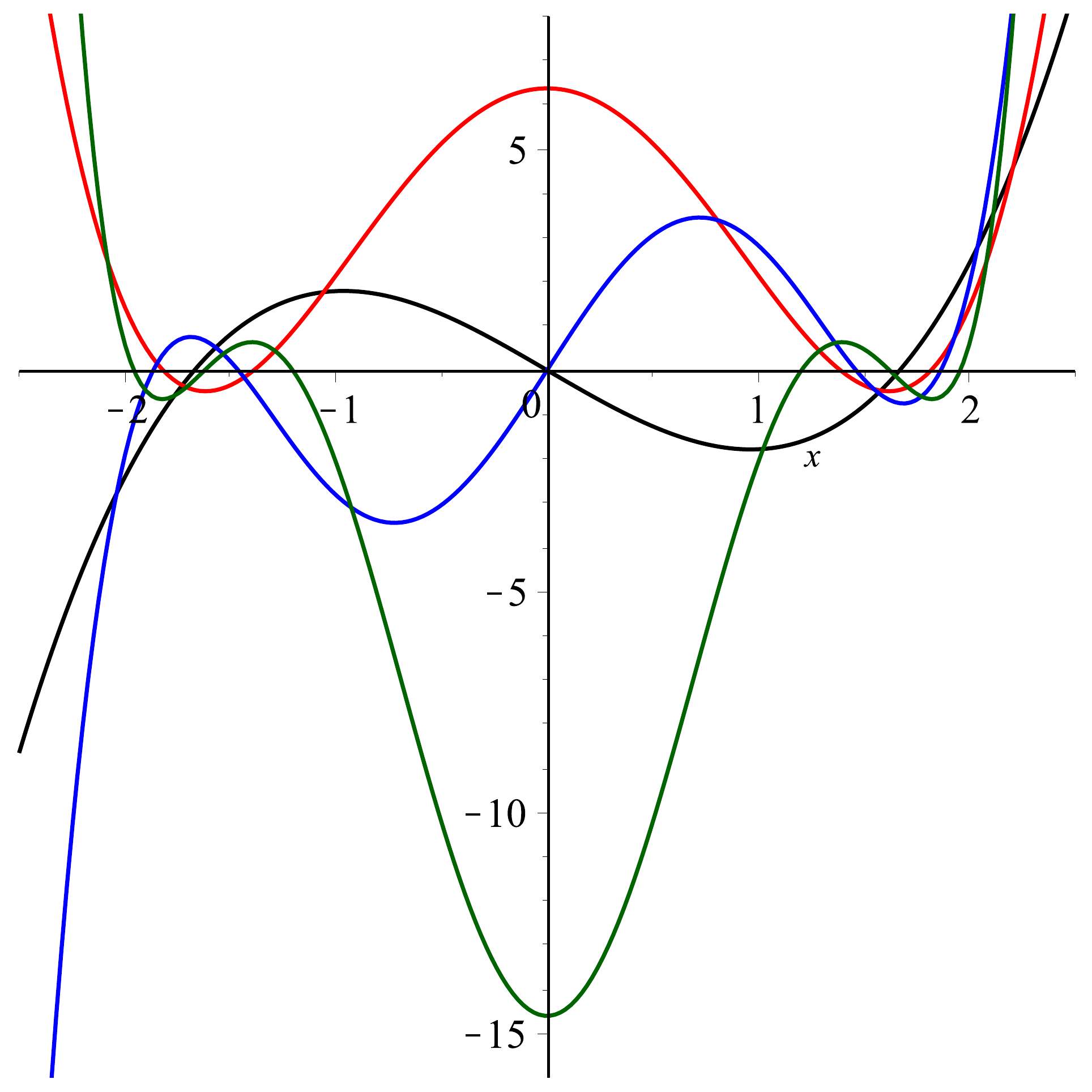}\\
 t=4 & t=5
\end{array}\]
\caption{\label{fig:Sn2}Plots of the polynomials $S_3(x;t)$ (black), $S_4(x;t)$ (red), $S_5(x;t)$ (blue), $S_6(x;t)$ (green) for $t=0,1,\ldots,5$, with $\la=\tfrac12$.}
\end{figure}

\section{\label{sec3}Asymptotic properties of generalized Freud Polynomials as \boldmath{$t\to\infty$}\label{a1}}
In this section we are concerned with the behaviour of the generalized Freud polynomials $S_n(x;t)$ as $t\to\infty$.
From Lemma~\ref{betaasymp} we see that
\[ \lim_{t\to\infty}\b_{2n}(t;\la) = 0, \qquad \lim_{t\to\infty}\b_{2n+1}(t;\la) = \tfrac12t,\]
i.e.
\[ \lim_{t\to\infty}\b_{n}(t;\la) =\tfrac14[1-(-1)^n]t.\] 

\begin{lemma}{Suppose that the monic polynomials $\SS_n(x;t)$ are generated by the three-term recurrence relation
\begin{equation}\label{ttrr:Q}x\SS_n(x;t)=\SS_{n+1}(x;t)+\bb_n(t)\SS_{n-1}(x;t), \end{equation}
where $\bb_n(t)=\tfrac14[1-(-1)^n]t$, with $\SS_0(x;t)=1$. Then
\begin{equation}\label{def:Snn} \SS_{2n}(x;t)=(x^2-\tfrac12t)^n,\qquad \SS_{2n+1}(x;t)=x(x^2-\tfrac12t)^n.\end{equation}
}\end{lemma}
\begin{proof}
From the three-term recurrence relation (\ref{ttrr:Q}) we have
\begin{align*}\SS_{2n+1}(x;t)&=x\SS_{2n}(x;t),\\
\SS_{2n+2}(x;t)&=x\SS_{2n+1}(x;t)-\tfrac12t\SS_{2n}(x;t) 
= (x^2-\tfrac12t)\SS_{2n}(x;t).
\end{align*}
Since $\SS_0(x;t)=1$ then the result immediately follows.
\comment{\begin{equation*}
\SS_{2n+1}(x;t)=x(x^2-\tfrac12t)^n,\qquad
\SS_{2n}(x;t)=(x^2-\tfrac12t)^n,
\end{equation*}
as required.}\end{proof}

In the limit as $t\to\infty$, we expect that the generalized Freud polynomials $S_n(x;t)$ will tend to the polynomials $\SS_{n}(x;t)$, see Remark~\ref{thm:Sn} below.
To show this we first define the polynomials $Q_n(y;t)$ and $\QQ_n(y)$ as follows
\begin{subequations}\label{def:Rn}\begin{align}
Q_n(y;t)&=(\tfrac12t)^{-n/2} S_n\big((\tfrac12t)^{1/2}\,y;t\big),\\
\QQ_n(y)&=(\tfrac12t)^{-n/2} \SS_n\big((\tfrac12t)^{1/2}\,y;t\big), 
\end{align}\end{subequations}
so from \eqref{def:Snn} we have
\begin{equation}\label{def:QPn}
\QQ_{2n+1}(y)=y(y^2-1)^n,\qquad
\QQ_{2n}(y)=(y^2-1)^n.\end{equation}
Plots of polynomials $Q_n(y;t)$, $n=3,4,\ldots,12$, for $t=20$ are given in Figure~\ref{fig:Qn}.

\begin{figure}[ht]
\[\begin{array}{c@{\quad}c}
\includegraphics[width=6cm]{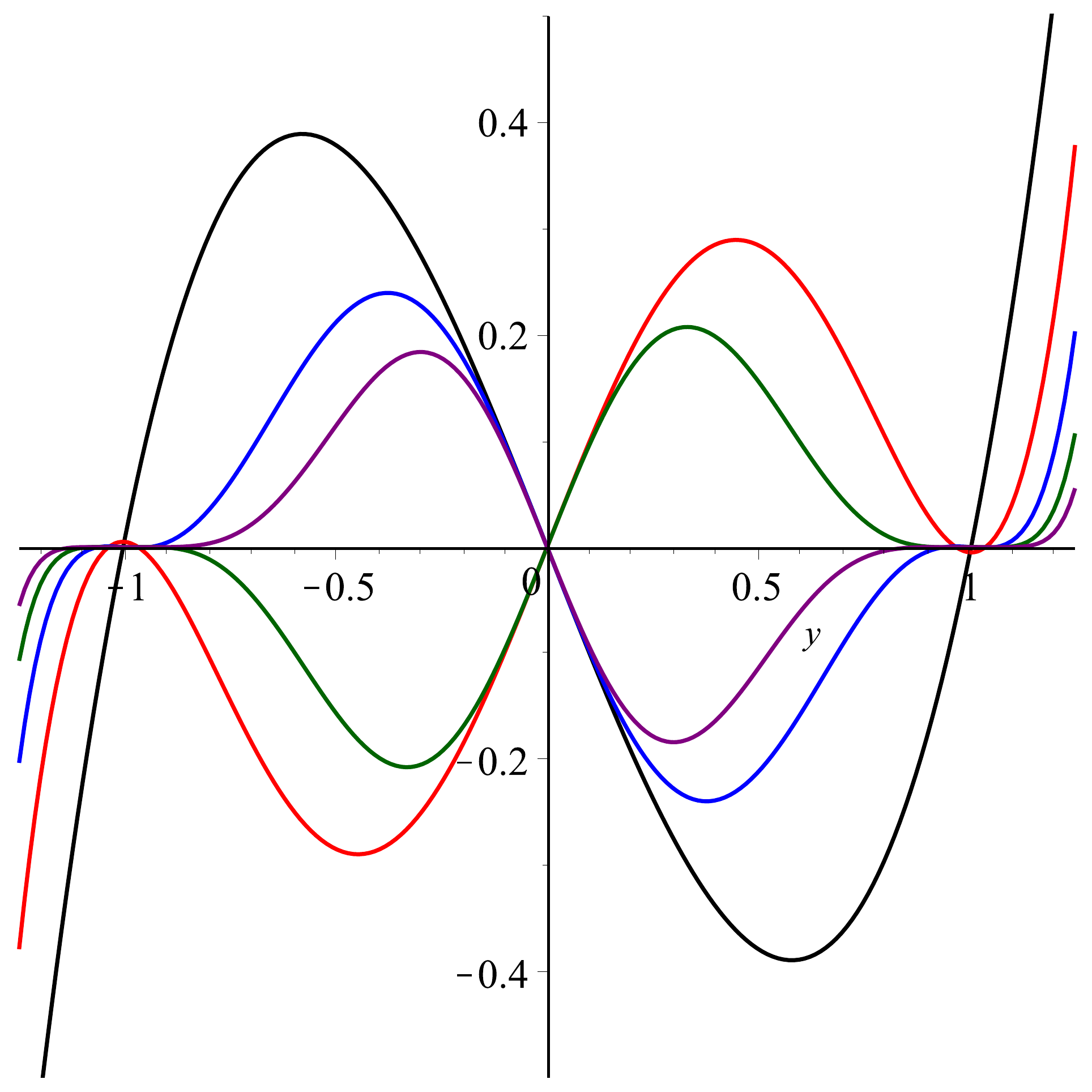} & \includegraphics[width=6cm]{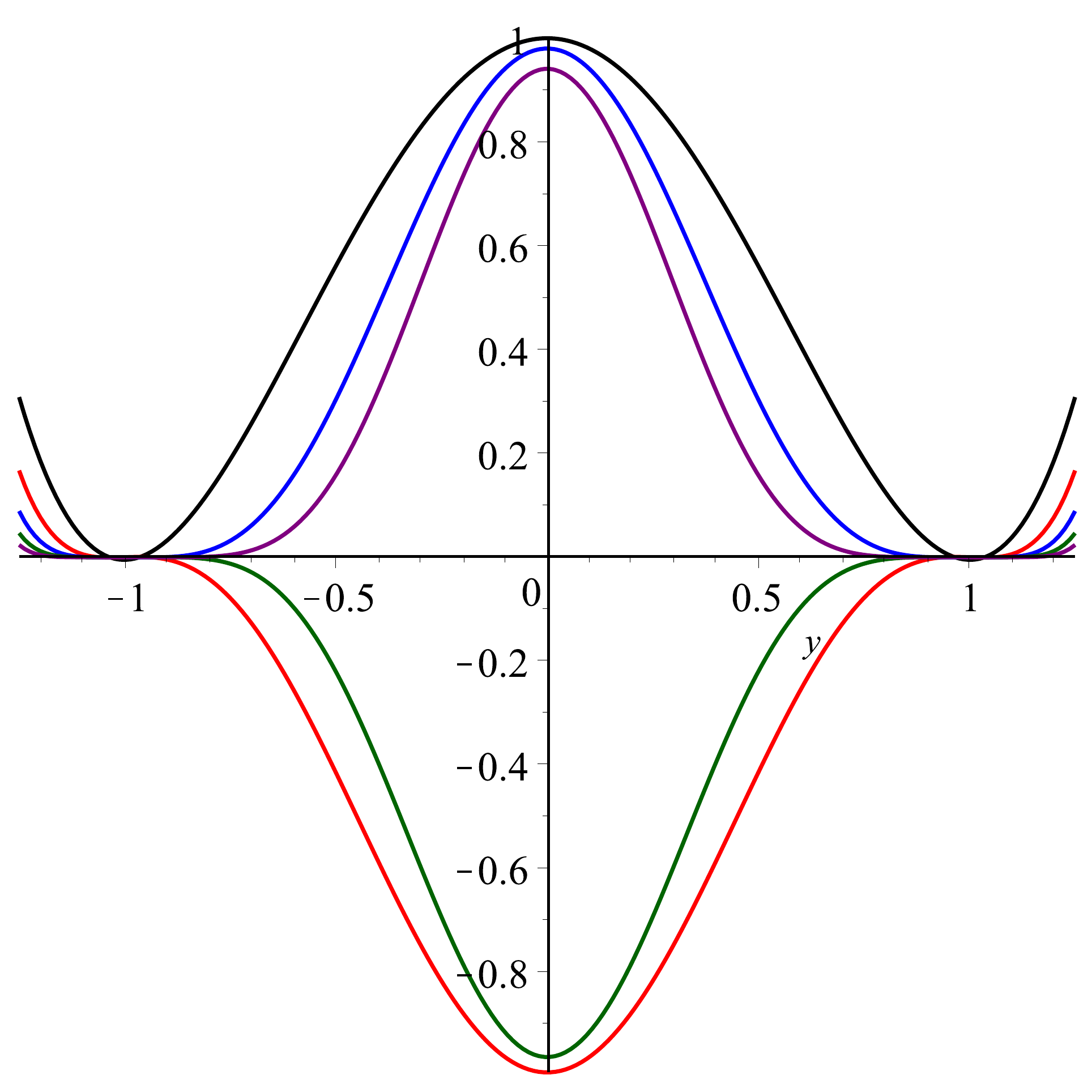}\\
Q_{2n-1}(y;t) & Q_{2n}(y;t) \\[5pt]
\end{array}\]
\caption{\label{fig:Qn}Plots of polynomials $Q_{2n-1}(y;t)$ and $Q_{2n}(y;t)$, for $n=1$ (black), $n=2$ (red), $n=3$ (blue), $n=4$ (green) and $n=5$ (purple), when $t=20$, with $\la=\tfrac12$.}
\end{figure}

\begin{theorem}{\label{thm32}Suppose that $Q_{2n}(y;t)$ and $\QQ_{2n}(y)$ are given by \eqref{def:Rn} and \eqref{def:QPn}, respectively, then as $t\to\infty$
\begin{subequations}\label{eq:36} \begin{align}\label{eq:36a}
 {Q_{2n}(y;t)} &={\QQ_{2n}(y)} -\frac{2n\la\QQ_{2n-2}(y)+n(n-1)\QQ_{2n-4}(y)}{t^2}+\O\left(t^{-4}\right),\\
 {Q_{2n+1}(y;t)} &={\QQ_{2n+1}(y)} -\frac{2n(\la+1)\QQ_{2n-1}(y)+n(n-1)\QQ_{2n-3}(y)}{t^2} +\O\left(t^{-4}\right).\label{eq:36b}
 \end{align}\end{subequations}
}\end{theorem}
\begin{proof}As $S_n(x;t)$ satisfies the three-term recurrence relation \eqref{eq:gfrr},
then using \eqref{def:Rn} we see that $Q_n(y;t)$ satisfies the three-term recurrence relation
\[ Q_{n+1}(y;t)=yQ_n(y;t)-\frac{2\b_n(t;\la)}{t} Q_{n-1}(y;t).\]
We shall prove \eqref{eq:36} by induction. By definition $Q_1(y;t)=y$, $Q_0(y;t)=1$, $\QQ_1(y)=y$ and
\[ \begin{split} Q_2(y;t) &=yQ_1(y;t)-\frac{2\b_1(t;\la)}{t} Q_{0}(y;t) = y^2-\frac{2\Ph(t)}{t},\\
{\QQ_{2}}(y)&=y^2-1,\end{split}\]
therefore
\[ \begin{split}
{Q_{2}} &= {\QQ_{2}} + 1-\frac{2\Ph}{t}\\
&= {\QQ_{2}} + 1-\frac{2}{t}\left\{\frac{t}{2}+\frac{\la}{t}-\frac{2\la(\la-1)}{t^3}+\O\left(t^{-5}\right)\right\}\\
&= {\QQ_{2}} -\frac{2\la}{t^2} +\O\left(t^{-4}\right)\\
&= {\QQ_{2}} -\frac{2\la\QQ_{0}}{t^2} +\O\left(t^{-4}\right)
\end{split}\] which shows \eqref{eq:36a} is true for $n=1$.
Also, by definition
\[\begin{split} Q_3(y;t) &=yQ_2(y;t)-\frac{2\b_2(t;\la)}{t} Q_{1}(y;t)\\
&=y \left\{y^2-\frac{2\Ph(t)}{t}\right\} -\frac{2y}{t}\left\{\frac{t}{2} -\Ph(t)+\frac{\la+1}{2\Ph(t)} \right\}\\
&= y(y^2-1)-\frac{(\la+1)y}{t\Ph(t)}
\end{split}\] and so, since $\QQ_{3}=y(y^2-1)$, then
\[ \begin{split}
Q_{3}&=\QQ_{3} -\frac{(\la+1)y}{t}\left\{ \frac{2}{t}-\frac{4\la}{t^3}+\O\left(t^{-5}\right)\right\}\\
&=\QQ_{3} -\frac{2(\la+1)y}{t^2} +\O\left(t^{-4}\right)\\
&=\QQ_{3} -\frac{2(\la+1)\QQ_{1} }{t^2} +\O\left(t^{-4}\right)
\end{split}\] which shows \eqref{eq:36b} is true for $n=1$.

Next suppose that \eqref{eq:36a} is true. Since
\[ \b_{2n+1} =\frac{t}{2}+\frac{\la-n}{t}+\O\left(t^{-3}\right),\]
then
\[ \begin{split}
{Q_{2n+2}}&=y {Q_{2n+1}} - \frac{2\b_{2n+1}}{t} Q_{2n}\\
&=y\left\{{\QQ_{2n+1}} -\frac{2n(\la+1)\QQ_{2n-1}+n(n-1)\QQ_{2n-3}}{t^2} +\O\left(t^{-4}\right)\right\}\\
&\qquad-\left\{1+\frac{2(\la-n)}{t^2}+\O\left(t^{-4}\right) \right\} 
\left\{{\QQ_{2n}} -\frac{2n\la\QQ_{2n-2}+n(n-1)\QQ_{2n-4}}{t^2}+\O\left(t^{-4}\right)\right\}\\
&=y\QQ_{2n+1}-\QQ_{2n}
-\frac{2(\la-n)\QQ_{2n}+2n(\la+1)y\QQ_{2n-1}-2n\la\QQ_{2n-2}}{t^2}\\ &\qquad
-\frac{n(n-1)[y\QQ_{2n-3}-\QQ_{2n-4}]}{t^2} +\O\left(t^{-4}\right)\\
&=\QQ_{2n+2}
-\frac{2(\la-n)\QQ_{2n}+2n(\la+1)[\QQ_{2n}+\QQ_{2n-2}]-2n\la\QQ_{2n-2}}{t^2} 
-\frac{n(n-1)\QQ_{2n-2}}{t^2} +\O\left(t^{-4}\right)\\
&=\QQ_{2n+2} -\frac{2(n+1)\la \QQ_{2n}+n(n+1)\QQ_{2n-2}}{t^2} +\O\left(t^{-4}\right)
\end{split}\]
which is \eqref{eq:36a} for $n\to n+1$, where we have used the relation \[y\QQ_{2m-1}=\QQ_{2m}+\QQ_{2m-2}.\]

Now we suppose that \eqref{eq:36b} is true. From Lemma~\ref{betaasymp} we have
\[ \b_{2n+2} = \frac{n+1}{t}+\O\left(t^{-3}\right),\] and so
\[ \begin{split}
Q_{2n+3}&={yQ_{2n+2}}-\frac{2\b_{2n+2}}{t} Q_{2n+1}\\
&= y \left\{ \QQ_{2n+2} -\frac{2(n+1)\la \QQ_{2n}+n(n+1)\QQ_{2n-2}}{t^2} +\O\left(t^{-4}\right) \right\}
- \frac{2(n+1)}{t^2}
\left\{{\QQ_{2n+1}} +\O\left(t^{-2}\right) \right\}\\
&=y\QQ_{2n+2} -\frac{2(n+1)\QQ_{2n+1}+2(n+1)\la y\QQ_{2n}+n(n+1)y\QQ_{2n-2}}{t^2} 
+\O\left(t^{-4}\right) \\
&=\QQ_{2n+3} -\frac{2(n+1)(\la+1) \QQ_{2n+1}+n(n+1)\QQ_{2n-1}}{t^2} +\O\left(t^{-4}\right)
\end{split}\]
which is \eqref{eq:36b} for $n\to n+1$, where we have used the relation $y\QQ_{2m}=\QQ_{2m+1}$.
Hence the result follows by induction.
\end{proof}
\comment{\begin{align*}
 {Q_{2n}(y;t)} &={\QQ_{2n}} -\frac{2n\la\QQ_{2n-2}+n(n-1)\QQ_{2n-4}}{t^2}+\O\left(t^{-4}\right),\\
 {Q_{2n+1}(y;t)} &={\QQ_{2n+1}} -\frac{2n(\la+1)\QQ_{2n-1}+n(n-1)\QQ_{2n-3}}{t^2} +\O\left(t^{-4}\right).
 \end{align*}}

\begin{figure}[ht]
\[\begin{array}{cc}
\fig{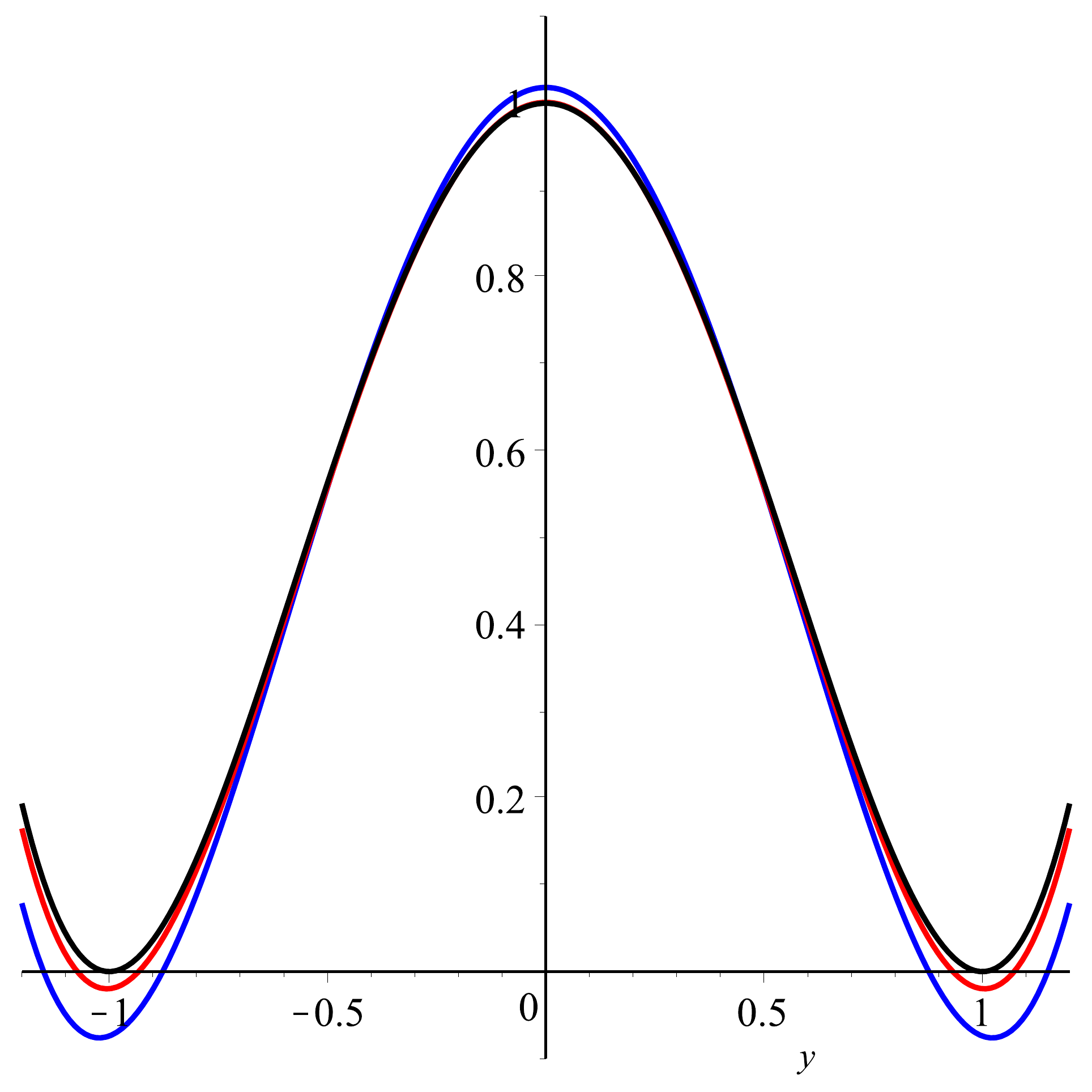} & \fig{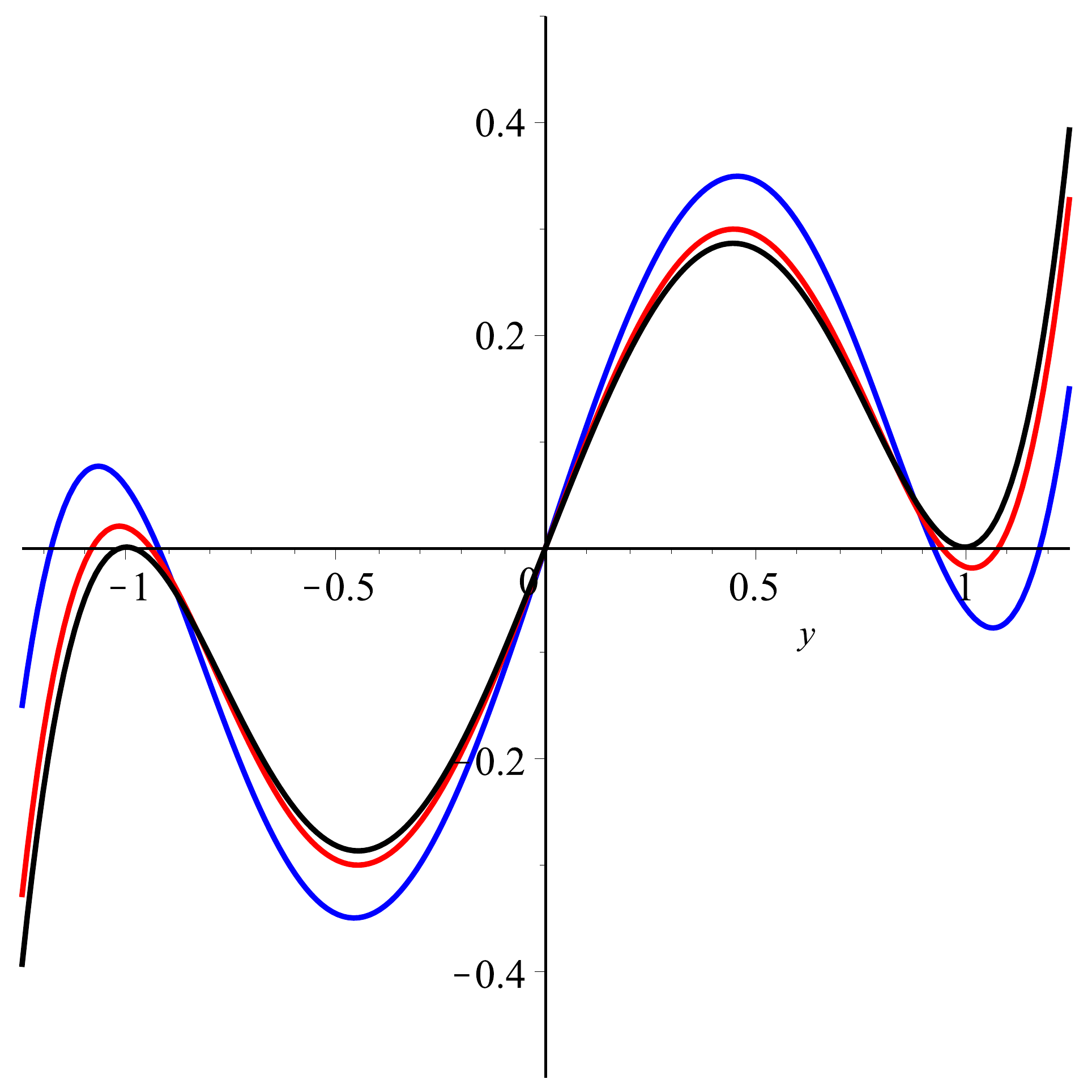} \\
n=4 & n=5 \\[5pt]
\fig{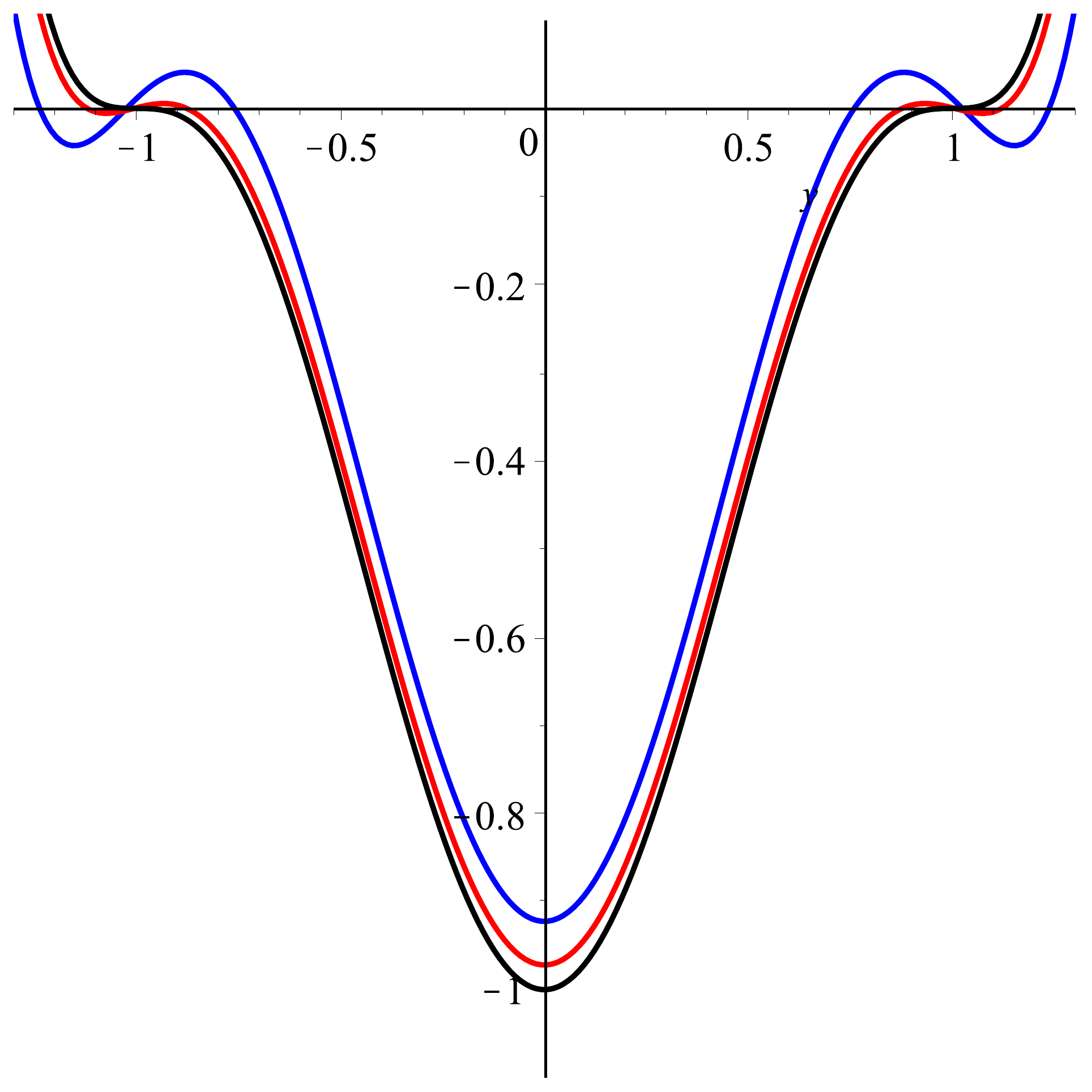} &\fig{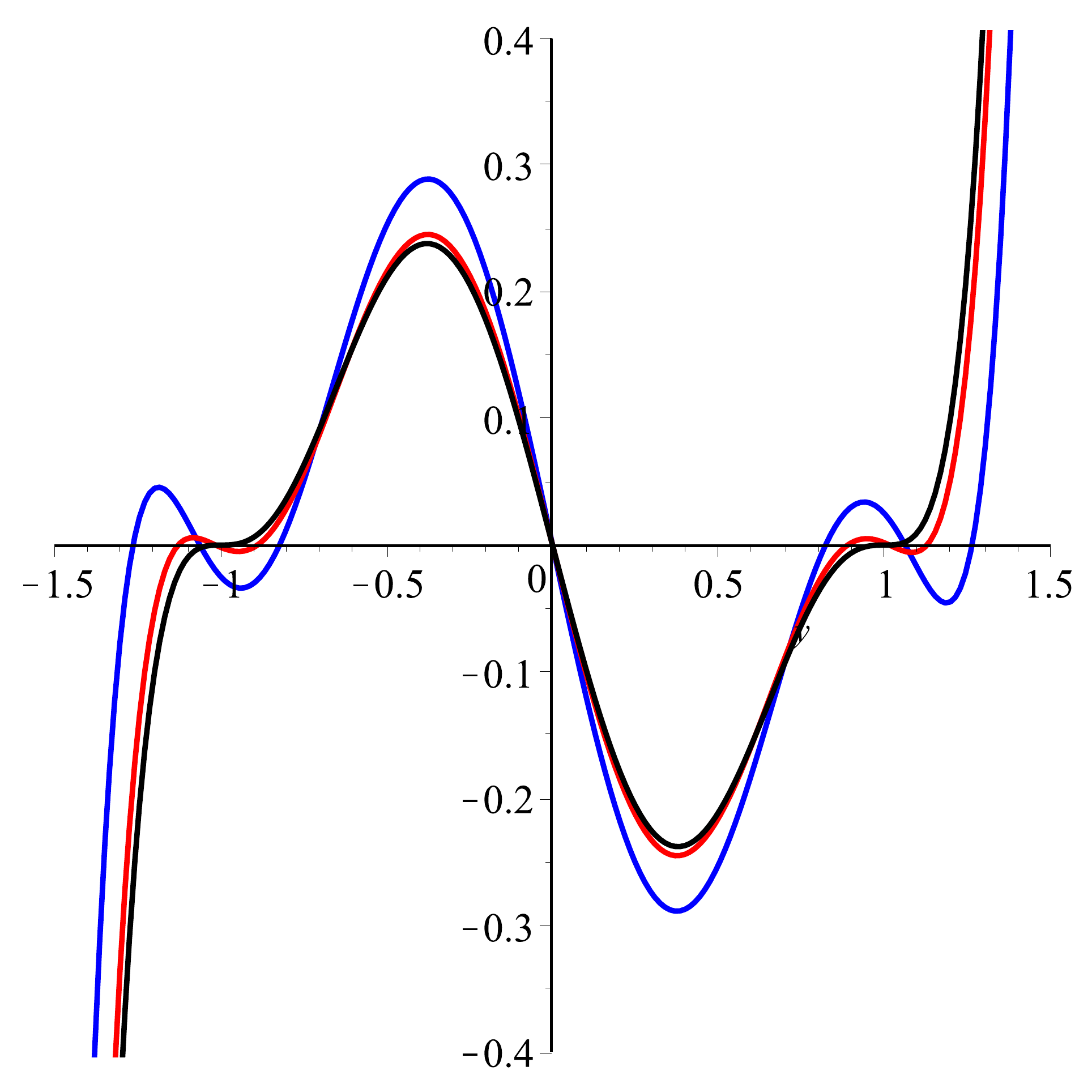} \\
n=6 & n=7 \\[5pt]
\fig{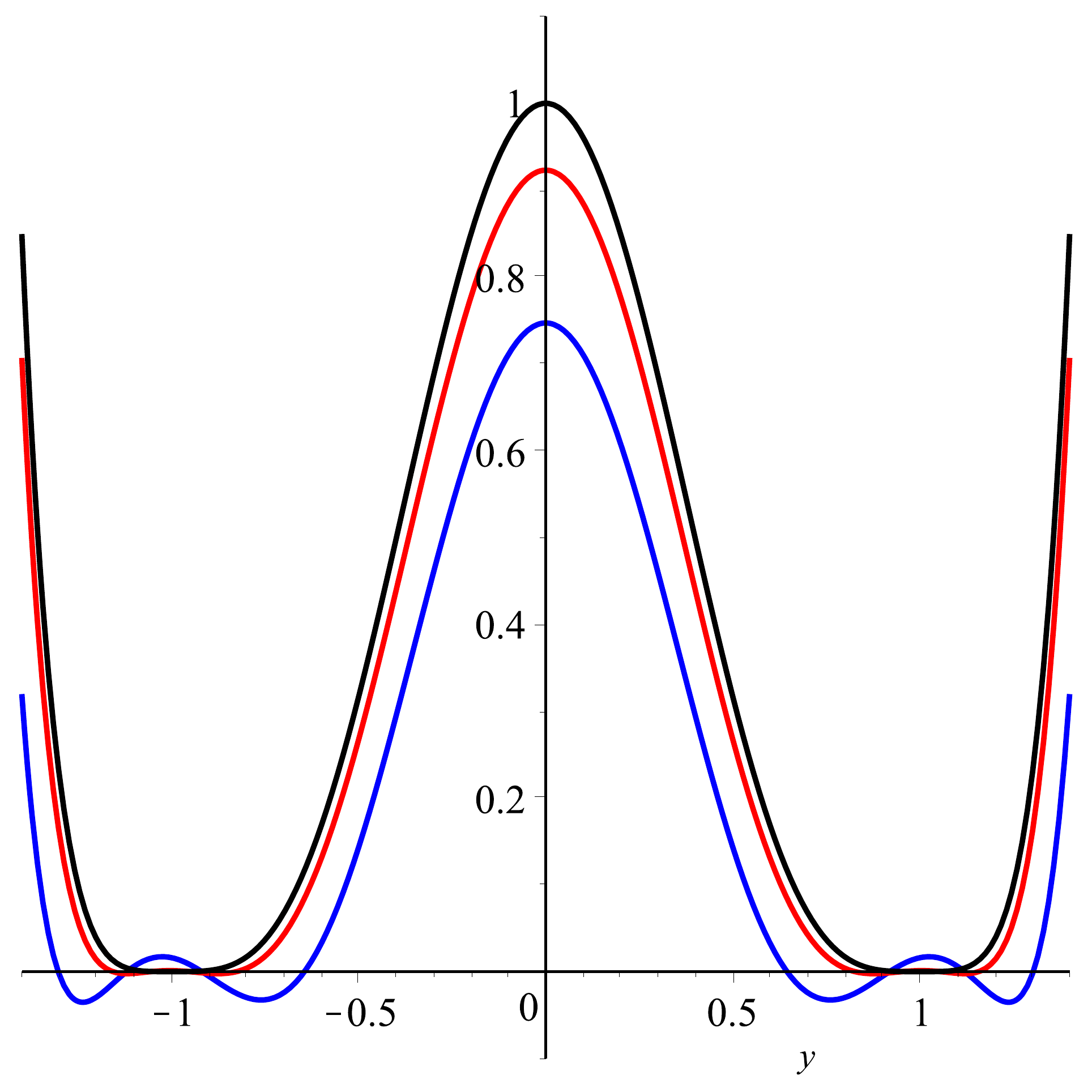} & \fig{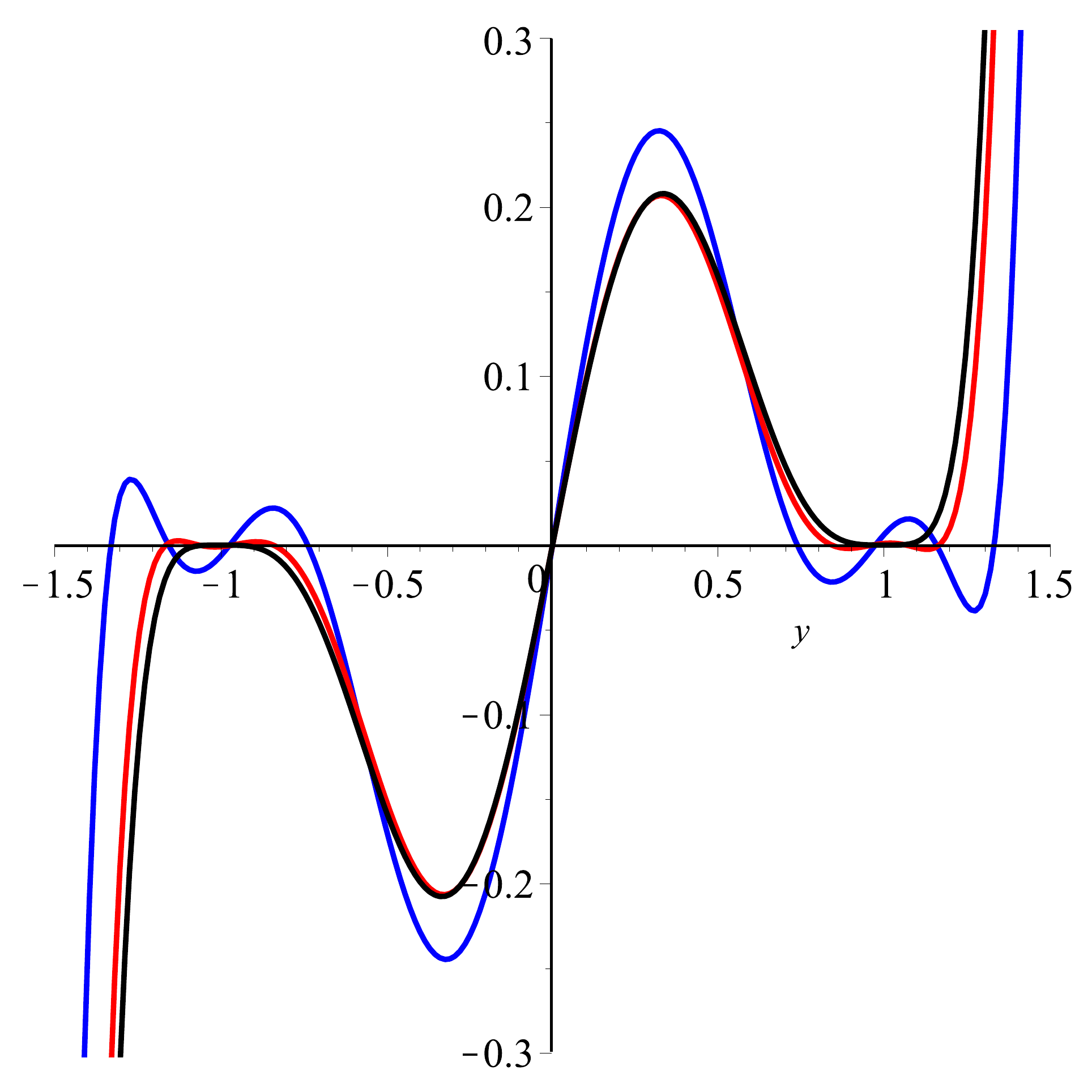}\\
n=8 & n=9
\end{array}\]
\caption{Plots of the polynomials $Q_n(y;5,\tfrac12)$ (blue), $Q_n(y;10,\tfrac12)$ (red) and $\QQ_n(y)$ (black) for $n=4,5,\ldots,9$.}
\end{figure}

\comment{Consequently we have the following result
\begin{theorem}{\label{thm:Sn}Suppose that the monic polynomials $\SS_n(x;t)$ are given by \eqref{def:Snn} and the monic polynomials $S_n(x;t)$ are generated by the three-term recurrence relation (\ref{eq:gfrr}). Then
in the limit as $t\to\infty$
\begin{subequations}
\begin{align*}
S_{2n}(x;t)&\to (x^2-\tfrac12t)^n=\SS_{2n}(x;t),\\ S_{2n+1}(x;t)&\to x(x^2-\tfrac12t)^n=\SS_{2n+1}(x;t).
\end{align*}
\end{subequations}
}\end{theorem}}

\begin{remark}{\label{thm:Sn}Suppose that the monic polynomials $\SS_n(x;t)$ are given by \eqref{def:Snn} and the monic polynomials $S_n(x;t)$ are generated by the three-term recurrence relation (\ref{eq:gfrr}). Then formally, as $t\to \infty$,

\begin{subequations}
\begin{align*}
S_{2n}(x;t)&\to (x^2-\tfrac12t)^n=\SS_{2n}(x;t),\\ S_{2n+1}(x;t)&\to x(x^2-\tfrac12t)^n=\SS_{2n+1}(x;t).
\end{align*}
\end{subequations}
In other words, if the positive zeros of $S_{n}(x,t)$ are denoted by $ x_{n,k}(\lambda,t)$ for $k=1,2,\dots, \lfloor n/2\rfloor $, where $\lfloor m\rfloor $ is the largest integer smaller than $m$, we have, 
$$\lim_{t\to\infty}x_{n,k}(\lambda,t)=\tfrac12{t},\qquad k=1,2,\dots,\lfloor n/2\rfloor .$$ 
Since the zeros are symmetric with respect to the origin, the negative zeros of $S_n(x;t)$ approach $-\tfrac{1}{2}t$ in the limit as $t\to\infty$. 
}\end{remark}

\section{Asymptotic properties of generalized Freud Polynomials as \boldmath{$n\to\infty$}}\label{sec4}
\subsection{Asymptotics for the recurrence coefficient $\b_n(t;\la)$ as $n\to\infty$}
The asymptotic expansion of $\b_n(t;\la)$ in \eqref{eq:gfrr} satisfying \eqref{nonlineardiffer} for the special case when $t=0$ and $\la=-\frac 12$ was studied by Lew and Quarles \cite{refLewQuarles}, see also \cite{refNevai1983, refNP}. 
The asymptotics for the more general case when $t\in\R$ and $\la=-\frac 12$ were given by Clarke and Shizgal \cite{refClarkeShizgal} in the context of bimode polynomials. 
In the next theorem we provide the asymptotic expansion of $\b_n(t;\la)$ in \eqref{eq:gfrr} as $n\to\infty$, for $t,\la\in\R$.

\begin{theorem}\label{thm41}
Let $t,\la \in \R$, then as $n\to\infty$, the recurrence coefficient $\b_n$ associated with monic generalized Freud polynomials
satisfying the nonlinear discrete equation \eqref{nonlineardiffer}, i.e.
\begin{equation*}
 \b_{n} \left( \b_{n+1} + \b_{n} + \b_{n-1} -\tfrac12{t} \right)= \tfrac14[n+(2\lambda +1)\Delta_{n}],
 \end{equation*}
where $\Delta_n=\tfrac12[1-(-1)^n]$, has the asymptotic expansion
\begin{align*}
\b_{n}(t;\la) &= \frac{\sqrt{3}\,n^{1/2}}{6}\left\{1+\frac{\sqrt{3}\,t}{6n^{1/2}}+\frac{t^2+12(2\la+1)\Delta_{n}}{24n} \right.\\
&\quad\left.
-\ \frac{t^4+24(2\la+1)\Delta_{n}t^2 + 48[3(2\la+1)^2\Delta_{n}^2-1]}{1152n^2} 
+ \mathcal{O}(n^{-5/2})\right\},\end{align*}
i.e.
\comment{\begin{subequations}\label{bnasymp}
\begin{align}
\b_{n}(t;\la) = \frac{\sqrt{3}\,n^{1/2}}{6}\left\{1+\frac{\sqrt{3}\,t}{6n^{1/2}}+\frac{t^2}{24n}-\frac{t^4-48}{1152n^2}
+\frac{\sqrt{3}\,t}{144n^{5/2}}+\mathcal{O}(n^{-3})\right\},
\end{align} when $n$ is even and
\begin{align}\nonumber
\b_{n}(t;\la) = \frac{\sqrt{3}\,n^{1/2}}{6}&\left\{1+\frac{\sqrt{3}\,t}{6n^{1/2}}+\frac{t^2+12(2\la+1)}{24n}
- \frac{t^4+24(2\la+1)t^2 + 96 (6 \la^2 + 6\la+1)}{1152n^2} \right.\\
&\quad\left. +\ \frac{\sqrt{3}\,t}{144n^{5/2}}+\mathcal{O}(n^{-3})\right\}\end{align} when $n$ is odd.
\end{subequations}}
\begin{subequations}\label{bnasymp}
\begin{align}
\b_{2n}(t;\la) &= \frac{\sqrt{6}\,n^{1/2}}{6}\left\{1+\frac{\sqrt{6}\,t}{12n^{1/2}}+\frac{t^2}{48n}-\frac{t^4-48}{4608n^2}
+\mathcal{O}(n^{-5/2})\right\},\\ \nonumber
\b_{2n+1}(t;\la) &= \frac{\sqrt{3}\,(2n+1)^{1/2}}{6}\left\{1+\frac{\sqrt{3}\,t}{6(2n+1)^{1/2}}+\frac{t^2+12(2\la+1)}{24(2n+1)}\right.\\
&\left.\quad - \frac{t^4+24(2\la+1)t^2 + 96 (6 \la^2 + 6\la+1)}{1152(2n+1)^2} +\mathcal{O}(n^{-5/2})
\right\}.\end{align}
\end{subequations}
\end{theorem}

\begin{proof} We begin by finding the first term in the asymptotic expansion.
The recurrence coefficient $\b_n$ associated with \eqref{genFreud} is positive and diverges as $n\to\infty$ (cf.~\cite{refCJK}) which suggests that as $n\to\infty$ \begin{equation}\label{basic2}
\b_n \sim B n^{r},
\end{equation}
with $r>0$ and $B$ a constant. Substituting the asymptotic form \eqref{basic2} into \eqref{nonlineardiffer} we obtain
\begin{equation*}
 3B^2n^{2r}-\tfrac12tBn^r\sim \tfrac14[n+ (2\lambda +1)\Delta_{n}].
 \end{equation*}
Since we require this equation to hold for all $n=1,2,\dots$, it follows that $r=\frac 12$, $B=\tfrac16\sqrt{3}$ and so the leading behaviour is given by \[\b_n\sim \tfrac16 \sqrt{3}\, n^{1/2}.\]

Next we suppose that as $n\to\infty$
\begin{equation}\label{bnexp}\b_n=\frac{\sqrt{3}\,n^{1/2}}{6}\sum_{k=0}^{\infty}\frac{b_k}{n^{k/2}},\end{equation}
with 
$b_0=1$. Substituting this together with
\begin{align*}\b_{n\pm1}&=\frac{\sqrt{3}\,(n\pm1)^{1/2}}{6}\sum_{k=0}^{\infty}\frac{b_k}{(n\pm1)^{k/2}}\\
&=\frac{\sqrt{3}\,n^{1/2}}{6}\left\{1+\frac{b_1}{n^{1/2}}+\frac{2b_2\pm1}{2n} +\frac{b_3}{n^{3/2}}
+\frac{8b_4\mp 4b_2-1}{8n^2} +\O(n^{-5/2}) \right\},
\end{align*}
as $n\to\infty$, which are obtained by letting $n\to n\pm1$ in \eqref{bnexp} and doing an asymptotic expansion,
into \eqref{nonlineardiffer}, doing an asymptotic expansion 
and equating powers of $n$ gives
\begin{subequations}\begin{align*}
b_1&-\tfrac16\sqrt{3}\,t=0,\\
b_2&+\tfrac12b_1^{2}-\tfrac16\sqrt{3}\,tb_1-\tfrac12\Delta_n=0,\\
b_3&+b_1b_2-\tfrac16\sqrt{3}\,tb_2=0,\\
b_4&+b_1b_3+\tfrac12b_2^{2}-\tfrac16\sqrt{3}\,tb_3-\tfrac1{24}=0,
\end{align*}\end{subequations}
which have solution
\begin{equation*}\begin{split}
b_1&=\frac{\sqrt {3}t}{6},\quad b_2=\frac{{t}^{2}+12\Delta_n}{24},\quad b_3=0,
\quad b_4=\frac{t^{4}+24\Delta_nt^2+48(3\Delta_n^2-3)}{1152},
\end{split}\end{equation*}
Hence, letting $\Delta_n=\tfrac12[1-(-1)^n]$, we formally obtain the asymptotic expansions \eqref{bnasymp}. 
From the nonlinear discrete equation \eqref{nonlineardiffer} we obtain
\[ \b_n^2-\tfrac12t\b_n+\b_{n} \left( \b_{n+1} + \b_{n-1} \right)= \tfrac14[n+(2\la +1)\Delta_{n}],\]
and so
\[ \b_n^2-\tfrac12t\b_n< \tfrac14[n+(2\la +1)\Delta_{n}],\]
since $\b_{n} \left( \b_{n+1} + \b_{n-1} \right)>0$, as $\b_n=h_n/h_{n-1}>0$, where
\[ h_n=\int_{-\infty}^{\infty}S_n^2(x;t)w_{\la}(x;t)\,dx.\]
 Therefore
\begin{align} 0<\b_n &<\tfrac14t + \tfrac12\sqrt{n+\tfrac14t^2+(2\la +1)\Delta_{n}}\label{bninequal}\\
&=\tfrac14t + \tfrac12n^{1/2}\left[1+ \frac{\tfrac14t^2+ (2\la +1)\Delta_{n}}{n} \right]^{1/2}\nonumber\\
&= \tfrac12n^{1/2}+ \tfrac14t +\frac{\tfrac14t^2+ (2\la +1)\Delta_{n}}{2n^{1/2}} +\sum_{m=2}^\infty a_m \frac{\left[\tfrac14t^2+ (2\la +1)\Delta_{n}\right]^m}{n^{m-1/2}},\nonumber
\end{align}
where\[ a_m= \frac{(-1)^{m+1}(2m-3)!}{2^{2m-1}\,m!\,(m-2)!},\]
which completes the proof.
\end{proof}

\begin{corollary}\label{keyy}
Assume that $\b_n(t;\la)$ satisfies \eqref{nonlineardiffer}. Then, for $t,\la \in \R$:
\begin{itemize}
\item[(i)] the sequence $\left\{\ds{\frac{\b_n(t;\la)}{\sqrt n}}\right\}_{n=1}^{\infty}$ is bounded;
\item[(ii)] $\ds \lim_{n\rightarrow \infty} \frac{\b_n(t;\la)}{\sqrt n} = \frac{\sqrt{3}}{6}$.
\end{itemize}
\end{corollary}
\begin{remarks}
\begin{enumerate}\item[]
\item Bleher and Its \cite{refBleherIts99} studied the limit of the recurrence coefficient $R_n$ as $n$, $N\to\infty$ when the ratio $n/N$ tends to a positive constant, for
the polynomials $P_n(x)$ orthogonal with respect to the weight 
\begin{subequations}\label{BI}\begin{equation}\label{BIa} w(x)=\exp\{-NV(x)\},\qquad x\in\mathbb{R},\end{equation}
with
\begin{equation}\label{BIb} V(x)=\tfrac14gx^4+\tfrac12 Tx^2,\qquad g>0,\end{equation}\end{subequations}
satisfying the three-term recurrence relation
\begin{equation*} xP_n(x)=P_{n+1}(x)+R_nP_{n-1}(x),\end{equation*}
where $R_n$ satisfies the Freud equation
\begin{equation} \label{BIfreud} n=N R_n \left[T+g(R_{n+1}+R_n+R_{n-1})\right],\qquad n\geq1.\end{equation}

Consider equation \eqref{nonlineardiffer} with $\la=-\tfrac12$, i.e.
\begin{equation}\label{nonln0}
 \b_{n} \left( \b_{n+1} + \b_{n} + \b_{n-1} -\tfrac12{t} \right)= \tfrac14n.
 \end{equation}
 The relationship between equations \eqref{BIfreud} and \eqref{nonln0} is given by
\begin{equation} \label{transRbeta} R_n=\frac{2\b_n}{\sqrt{gN}},\qquad T=-\sqrt{\frac{g}{N}}\,t.\end{equation}
as is easily verified. 
In \cite{refBleherIts99} it is shown that $R_n$ satisfies the inequality
\begin{equation*} \label{BIRn}
0<R_n<\frac{-T+\sqrt{T^2+4ng/N}}{2g}.
\end{equation*}
Applying the transformation \eqref{transRbeta} to this yields
\begin{equation*}
0<\b_n<\tfrac14t+\tfrac12\sqrt{n+\tfrac14t^2},
\end{equation*}
which is \eqref{bninequal} with $\la=-\tfrac12$. 

\item Nevai \cite{refNevai1973,refNevai1983} and later Freud \cite{refFreud1976} proved that the recurrence coefficient associated with the special case of the symmetric weight \eqref{genFreud} where $\la = -\frac{1}{2}$ and $t=0$ has the same limit as the one in Corollary \ref{keyy} (ii). Corollary \ref{keyy} (ii) therefore proves an extension of Freud's conjecture \eqref{Freudconj} for recurrence coefficients associated with the weight \eqref{conj} to recurrence coefficients satisfying \eqref{nonlineardiffer} associated with the weight \eqref{genFreud} for $m=4$.

\item Recently Joshi and Lustri \cite{refJL15} studied the asymptotic behaviour of the first discrete \p\ equation in the limit as $n\to\infty$. Using an asymptotic series expansion, they identified two types of solutions which are pole-free within some sector of the complex plane containing the positive real axis and used exponential asymptotic techniques to determine Stokes phenomena effects in these solutions.
\end{enumerate}
\end{remarks}

In \cite{refDamelin}, Damelin considers asymptotics of recurrence coefficients associated with weights $|x|^{\rho}\exp\{-Q(x)\}$ where $Q(x)$ is an even polynomial of fixed degree.

\begin{theorem}\label{n-asymp}
For $t,\la\in\R$, the recurrence coefficients $\b_{n}(t;\la)$ in \eqref{nonlineardiffer} satisfy

\begin{subequations}\begin{align}\label{eq415}
\frac{\b_{n+1}(t;\la)}{\b_{n}(t;\la)} &= 1+ \mathcal{O}\left(n^{-1}\right), \quad n \rightarrow \infty,\\
\frac{\b_{n}(t;\la)}{a^{2}_n(t)} &= \tfrac{1}{4}+ \mathcal{O}\left(n^{-1}\right), \quad n \rightarrow \infty.\label{eq415b}
\end{align}\end{subequations}
where
 $a_n$ is the Mhaskar-Rakhmanov-Saff number \cite{refMSaff84,refRakhmanov} which is the unique positive solution of the equation
\[ n= \frac{2}{\pi} \ds \int_{0}^{1} \frac{a_nt Q'(a_{n}t)}{\sqrt{1-t^2}}\,d{t} \]
for $Q(x)= x^{4}-tx^{2}$.

\end{theorem}
\begin{proof}A proof of \eqref{eq415} and \eqref{eq415b} can be found in \cite[Thm. 2.1]{refDamelin}.
\end{proof}
\subsection{Asymptotics for the generalized Freud polynomials as $n\to\infty$}
Linear second-order differential equations, which are at the heart of much of special function theory, can be used in various ways to obtain asymptotic approximations and inequalities. The differential equation satisfied by generalized Freud polynomials orthogonal with respect to the weight \eqref{genFreud} was obtained in \cite{refCJK}.

\begin{theorem}{\label{thm43A}
Monic orthogonal polynomials $S_{n}(x;t)$ with respect to generalized Freud weight (\ref{genFreud}) satisfy the differential equation
\begin{equation}\label{eq:Snode}
\deriv[2]{S_n}{x}(x;t)+R_n(x;t)\deriv{S_n}{x}(x;t)+T_n(x;t)S_n(x;t)=0,
\end{equation}
where
\begin{subequations}\label{dee}\begin{align}   
R_n(x;t) 
&= -4x^3+2tx+\frac{2\la+1}{x}-\frac{2x}{x^2-\tfrac12t+\b_n+\b_{n+1}},\\[5pt]
T_n(x;t) 
&= 4nx^2+4\b_n+16\b_n(\b_n+\b_{n+1}-\tfrac12t)(\b_n+\b_{n-1}-\tfrac12t)\nonumber\\
&\qquad-\frac{8\b_n x^2+(2\la+1)[1-(-1)^n]}{x^2-\tfrac12t+\b_n+\b_{n+1}}+4(2\la +1)(-1)^n\b_n\nonumber\\
&\qquad+(2\la+1)[1-(-1)^n]\left(t-\frac{1}{2x^2}\right).
\end{align}\end{subequations}
}\end{theorem}
\begin{proof} See \cite[Thm. 6]{refCJK}.
\end{proof}

\begin{remark} The differential equation \eqref{eq:Snode} for the special case where $\la=-\frac 12$ and $t$ replaced by $2t$ is given in \cite[eqn. (6)]{refAHM2016} though, in their notation, the statement on p.\ 104 needs to be corrected to read
\begin{align*}S_n^t(x)= 4a_n^2&\left[4x^2\left(a_{n-1}^2+a_n^2+a_{n+1}^2-t-\frac{2}{x^2-t+a_n^2+a_{n+1}^2}\right)
+4\left(a_n^2+a_{n+1}^2-t\right)\left(a_{n-1}^2+a_{n}^2 -t\right)+1\right].
\end{align*}
\end{remark}

Differential systems satisfied by weights \eqref{BIa}, where $V(x)$ is an even polynomial with positive leading coefficient, are discussed by Bertola \textit{et al.} \cite{refBEH03,refBEH06}.

In \cite{refBleherIts99,refBleherIts03}, Bleher and Its discuss semiclassical asymptotics of orthogonal polynomials $P_n(x)$ with respect to the weight \eqref{BI} using a combination of formal semiclassical WKB-type analysis of linear differential and nonlinear discrete equations, and rigorous asymptotics of a Riemann-Hilbert problem together with the nonlinear steepest descent method due to Deift and Zhou \cite{refDZ93,refDZ95}, where the latter technique provides a justification of the former. 
A similar rigorous asymptotic analysis of monic orthogonal polynomials $S_{n}(x;t)$ with respect to the generalized Freud weight (\ref{genFreud}) lies beyond the scope of this paper and we shall not pursue this further here. 

We shall however make some remarks about equation \eqref{eq:Snode} for $n$ large.
Since from Theorem \ref{thm41} we have $\b_n= \tfrac16\sqrt{3n}+\O(1)$ as $n\to\infty$, it follows from \eqref{dee} that
\begin{subequations}
\begin{align} R_n(x;t) &= -4x^3+2tx+\frac{2\la+1}{x}+\O(n^{-1/2}),\\
T_n(x;t) &= 
(\tfrac43n)^{3/2}+\O(n),
\end{align}\end{subequations}
and so we consider the equation
\begin{equation}\label{eqSntilde}
\deriv[2]{\widehat{S}_n}{x} - \left(4x^3-2tx-\frac{2\la+1}{x}\right)\deriv{\widehat{S}_n}{x}+
(\tfrac43n)^{3/2}\widehat{S}_n=0.
\end{equation}
Equation \eqref{eqSntilde} is equivalent to the Biconfluent Heun equation, cf.~\cite[\S31.12]{refNIST}
\begin{equation*}
\deriv[2]{u}{z}-\left(\frac{\gamma}{z}+\delta+z\right)\deriv{u}{z}+\left(\a-\frac{q}{z}\right)u=0,
\end{equation*}
through the transformation
\[ \widehat{S}_n(x;t,\la)=u(z;\a,\gamma,\delta,q),\qquad z=\tfrac12x^2,\]
with parameters
\[ \a=0,\qquad\gamma=-1-\la,\qquad \delta=-\tfrac12\sqrt{2}\,t,\qquad q=-\tfrac19\sqrt{6}\,n^{3/2}.\]
Note that if in equation \eqref{eqSntilde} we make the transformation $\widehat{S}_n(x)=w(\zeta)$, with $\zeta= (\tfrac43n)^{3/4}x$, then in the limit as $n\to\infty$ we obtain
\[ \deriv[2]{w}{\zeta} +\frac{2\la+1}{\zeta}\deriv{w}{\zeta}+w=0,\]
which has solution
\[ w(\zeta)=\left\{c_1J_{\la}(\zeta)+c_2J_{-\la}(\zeta)\right\}\zeta^{-\la},\]
with $J_{\la}(\zeta)$ the Bessel function. This suggests that there might be Mehler-Heine type asymptotic formulae for the polynomials ${S}_n(x;t)$ as $n\to\infty$, though we shall not investigate this further here. 

\section{Existence and uniqueness of positive solutions}\label{sec5}
 A natural question to ask is whether \eqref{nonlineardiffer} has many real solutions satisfying the initial condition $\b_0=0$. Several papers, including \cite{refLewQuarles,refNevai1983,refVanAssche07} provide an answer for the case where $t=0$. In a recent paper by Alsulami \textit{et al.}\ \cite{refANSVA15}, existence and uniqueness of a positive solution are discussed for the nonlinear second-order difference equations of the type
\begin{align}\label{maineq}
\b_n \left(\sigma_{n,1}\b_{n+1} + \sigma_{n,0}\b_{n} + \sigma_{n,-1}\b_{n-1}\right)+\kappa_n\b_n= \ell_n
\end{align}
with initial conditions $\b_0\in \R$, $\b_1\geq0$, $\kappa_n\in\R$ and mild conditions on the coefficients $\sigma_{n,j}$, $j=-1,0,1$. An excellent historical overview of the problem and its solution is given.

\begin{theorem}\label{Th48}
For $\la>-1$ and $\b_0 = 0$, there exists a unique $\b_1(t;\la)>0$ such that $\{\b_n(t;\la)\}_{n\in\N}$ defined by the nonlinear discrete equation
\begin{equation}\label{Th48eqn} \b_{n} \left( \b_{n+1} + \b_{n} + \b_{n-1} -\tfrac12{t} \right)= \tfrac14[n+(2\lambda +1)\Delta_{n}],
\end{equation}
with $\Delta_n=\tfrac12[1-(-1)^n]$, is a positive sequence and the solution arises when $\b_1(t;\la)$
is given by \eqref{beta1}, i.e.
\begin{equation}\label{beta1uni}\b_1(t;\la)=\Ph(t)
=\tfrac12t+\tfrac12\sqrt{2}\,\frac{\WhitD{-\la}\big(-\tfrac12\sqrt2\,t\big)}{\WhitD{-\la-1}\big(-\tfrac12\sqrt2\,t\big)}.\end{equation}
\end{theorem}
\begin{proof} The nonlinear discrete equation \eqref{Th48eqn} is the special case of \eqref{maineq} with
\[ \sigma_{n,1}= \sigma_{n,0}=\sigma_{n,-1}=1,\qquad \kappa_n=-\tfrac12t.\qquad \ell_n= \tfrac14[n+ (2\la+1) \Delta_n],\]
with $\Delta_n= \tfrac12[1-(-1)^{n}]$. The existence of $\b_1(t;\la)>0$ such that \eqref{Th48eqn} is a positive sequence follows immediately from \cite[Thm.\ 4.1]{refANSVA15}. The uniqueness of solutions of \eqref{maineq} is discussed in \cite[Thm.\ 5.2]{refANSVA15}, though the conditions in the theorem require that $t\leq 0$, $\la>-1$ and $\b_0=0$ in our case.
To show uniqueness for $t\in\R$, consider the nonlinear discrete equation \eqref{Th48eqn} with general initial conditions $\b_0=0$ and
${\b}_1=\widehat{\Phi}_{\la}(t;\th)$, where
\begin{equation}\label{betagenICs}
\widehat{\Phi}_{\la}(t;\th) =\frac{t}{2}+\frac{\sqrt{2}}{2}\left[\frac{\cos(\th)D_{-\la}\big(-\tfrac12\sqrt2\,t\big)
-\sin(\th)D_{-\la}\big(\tfrac12\sqrt2\,t\big)}
{\cos(\th)D_{-\la-1}\big(-\tfrac12\sqrt2\,t\big)+\sin(\th)D_{-\la-1}\big(\tfrac12\sqrt2\,t\big)}\right],\end{equation}
with $0\leq\th\leq\tfrac12\pi$ a parameter; if $\tfrac12\pi<\th<\pi$ then $\b_1$ 
has a pole at a finite value of $t$. Since the parabolic cylinder function $D_{\nu}(z)$ has the following asymptotics as 
$z\to\pm\infty$, cf.~\cite[\S12.9]{refNIST}
\begin{equation}
D_{\nu}(z) = \begin{cases} z^{\nu}\exp(-\tfrac14z^2)\big\{1+\O(z^{-2})\big\},&\mbox{\rm as}\quad z\to \infty,\\[2.5pt]
\ds \frac{\sqrt{2\pi}}{\Gamma(-\nu)}(-z)^{-\nu-1}\exp(\tfrac14z^2)\big\{1+\O(z^{-2})\big\}, &\mbox{\rm as}\quad z\to-\infty,
\end{cases}
\end{equation}
then as $t\to\pm\infty$, $\widehat{\Phi}_{\la}(t;\th)$ has the asymptotics
\begin{subequations}\begin{align*}
\widehat{\Phi}_{\la}(t;0) &= \begin{cases} \tfrac12 t+\O(t^{-1}), & \mbox{\rm as}\quad t\to\infty,\\[2.5pt]
\ds -\frac{\la+1}{t}+\O(t^{-3}),& \mbox{\rm as}\quad t\to-\infty,\end{cases}\\[2.5pt]
\widehat{\Phi}_{\la}(t;\th) &=\tfrac12 t+\O(t^{-1}),\qquad\qquad\! \mbox{\rm as}\quad t\to\pm\infty,
\quad\mbox{\rm if}\quad 0<\th<\tfrac12\pi,\\[2.5pt]
\widehat{\Phi}_{\la}(t;\tfrac12\pi) &=\begin{cases}
\ds -\frac{\la+1}{t}+\O(t^{-3}),& \mbox{\rm as}\quad t\to\infty,\\[2.5pt]
\tfrac12 t+\O(t^{-1}), & \mbox{\rm as}\quad t\to-\infty,\end{cases}
\end{align*}\end {subequations}
\comment{then as $t\to\infty$
\[ \widehat{\Phi}_{\la}(t;\th) = \begin{cases} \tfrac12 t+\O(t^{-1}), & \mbox{\rm if}\quad \th\not=\tfrac12\pi,\\[2.5pt]
\ds -\frac{\la+1}{t}+\O(t^{-3}),& \mbox{\rm if}\quad \th=\tfrac12\pi,
\end{cases}\]
and as $t\to-\infty$
\[ \widehat{\Phi}_{\la}(t;\th) = \begin{cases}\tfrac12 t+\O(t^{-1}), & \mbox{\rm if}\quad \th\not=0,\\[2.5pt]
\ds -\frac{\la+1}{t}+\O(t^{-3}),& \mbox{\rm if}\quad \th=0.
\end{cases}\]}
Consequently $\b_1=\widehat{\Phi}_{\la}(t;\th) >0$ for all $t\in\R$ if and only if $\th=0$, which proves the desired results. This result is illustrated in Figure \ref{fig:beta1} where $\b_1=\widehat{\Phi}_{\la}(t;\th)$, is plotted for various values of $\th$.
\end{proof}
\begin{figure}[ht]
\[\includegraphics[width=10cm]{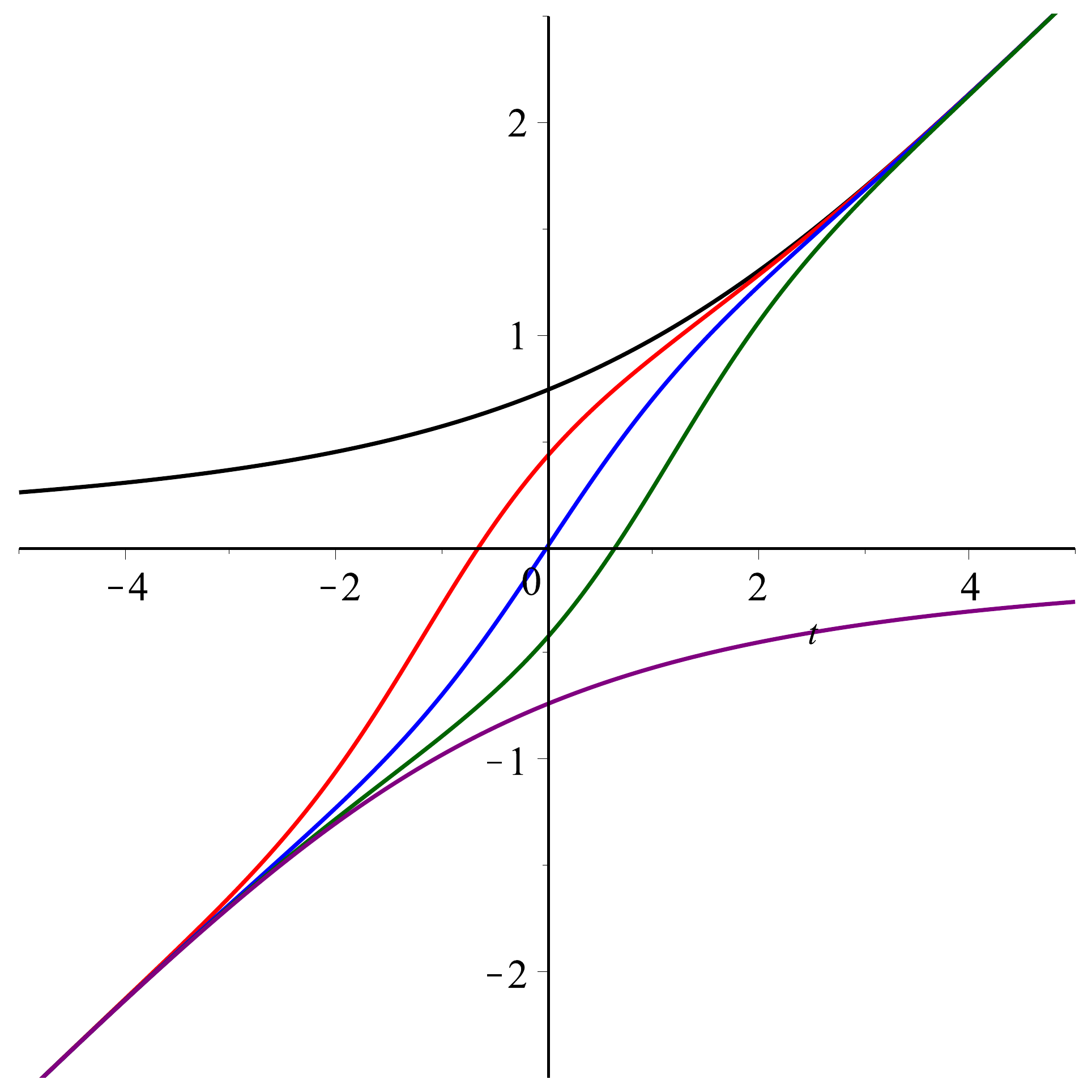}\]
\caption{\label{fig:beta1}
Plots of ${\b}_1=\widehat{\Phi}_{\la}(t;\th)$ as given by \eqref{betagenICs}, with $\la=\tfrac12$,
for $\th=0$ (black), $\th=\tfrac1{12}\pi$ (red), $\th=\tfrac1{4}\pi$ (blue), $\th=\tfrac5{12}\pi$ (green) and $\th=\tfrac1{2}\pi$ (purple).}
\end{figure}

\begin{remarks}
\begin{enumerate}\item[]
\item The rationale for considering $\b_1$ given by \eqref{betagenICs} is that
$\Ph(t)$ given by \eqref{def:Phi} satisfies the Riccati equation
\begin{equation}\label{Req:Phi}
\deriv{\Phi}{t}=-\Phi^2+\tfrac12t\Phi+\tfrac12(\la+1),
\end{equation}
which has general solution $\widehat{\Phi}_{\la}(t;\th)$. 
Letting $\Phi(t)=\ph'(t)/\ph(t)$ in \eqref{Req:Phi} gives
\[ \deriv[2]{\ph}{t}-\tfrac12t\deriv{\ph}{t}-\tfrac12(\la+1)\ph=0,\]
which has general solution
\[ \ph(t)=\left\{c_1D_{-\la-1}\big(-\tfrac12\sqrt2\,t\big)+c_2D_{-\la-1}\big(\tfrac12\sqrt2\,t\big)\right\}\exp(\tfrac18t^2),\]
with $c_1$ and $c_2$ arbitrary constants. Since only the ratio of $c_1$ and $c_2$ is relevant then we set $c_1=\cos\th$ and $c_2=\sin\th$.
\item The solution of the nonlinear discrete equation \eqref{Th48eqn} with initial conditions $\b_0=0$ and $\b_1$ given by \eqref{beta1uni} appears to depend on the sign of $t$.
In Figure \ref{fig:nbn} the points $(n,\b_n)$ are plotted for various choices of $t$. 
These show that $\b_n(t;\la)$ approaches a limit as $n\to\infty$ in different ways depending on whether $t<0$ or $t>0$.
If $t<0$ then the behaviour is similar irrespective of the value of $t$ and the plots suggest that $\{\b_{2n}\}_{n\in\N}$ and $\{\b_{2n+1}\}_{n\in\N}$ are both monotonically increasing sequences. However if $t>0$, the plots suggest that $\{\b_{2n}\}$ and $\{\b_{2n+1}\}$ are both monotonically increasing sequences for $n>n^*$, for some $n^*$ dependent on $t$.
The plots were generated in MAPLE using $250$ digits accuracy.

\def\figg#1{\includegraphics[width=5cm]{#1}}
\begin{figure}[ht]
\[\begin{array}{c@{\quad}c@{\quad}c}
\figg{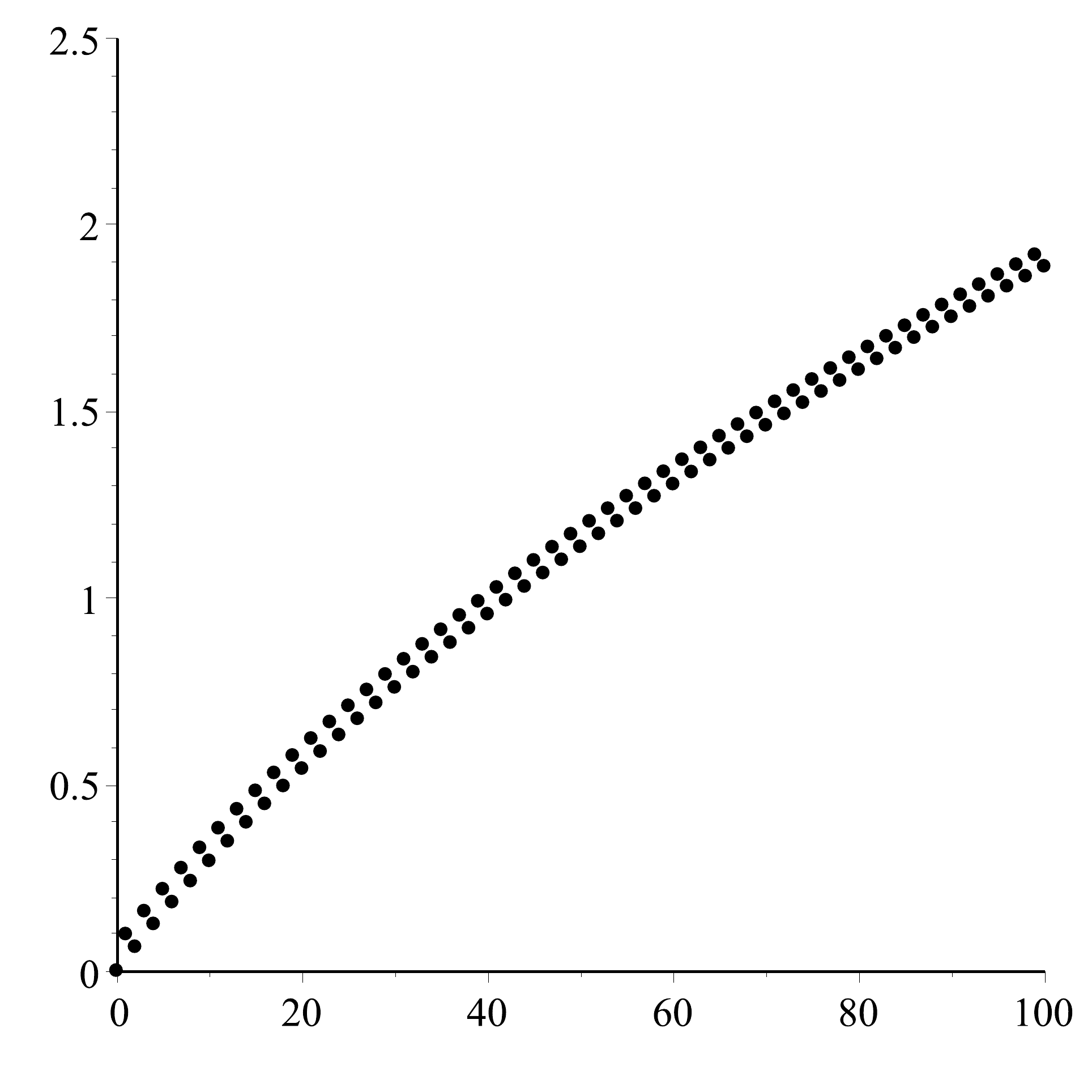} & \figg{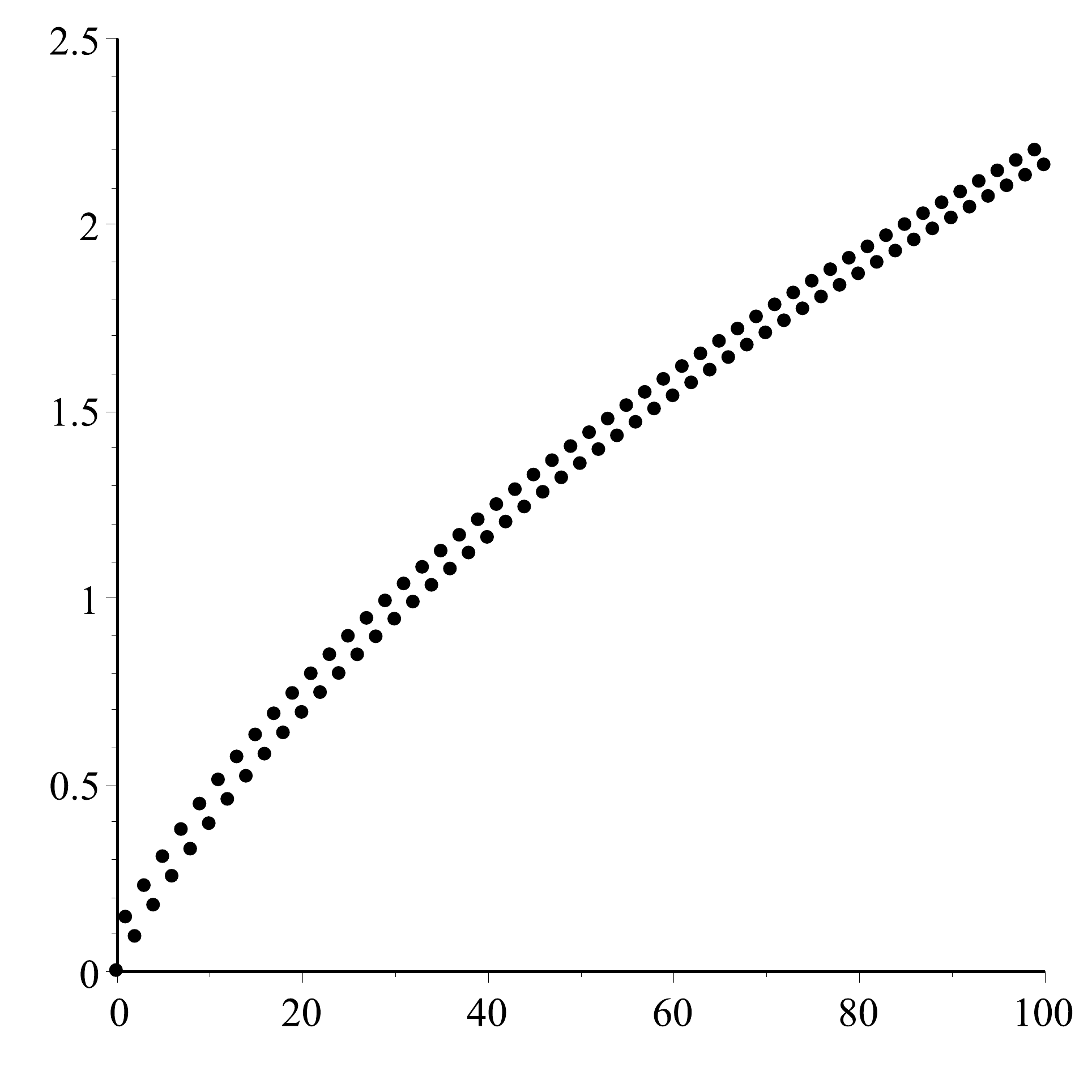} & \figg{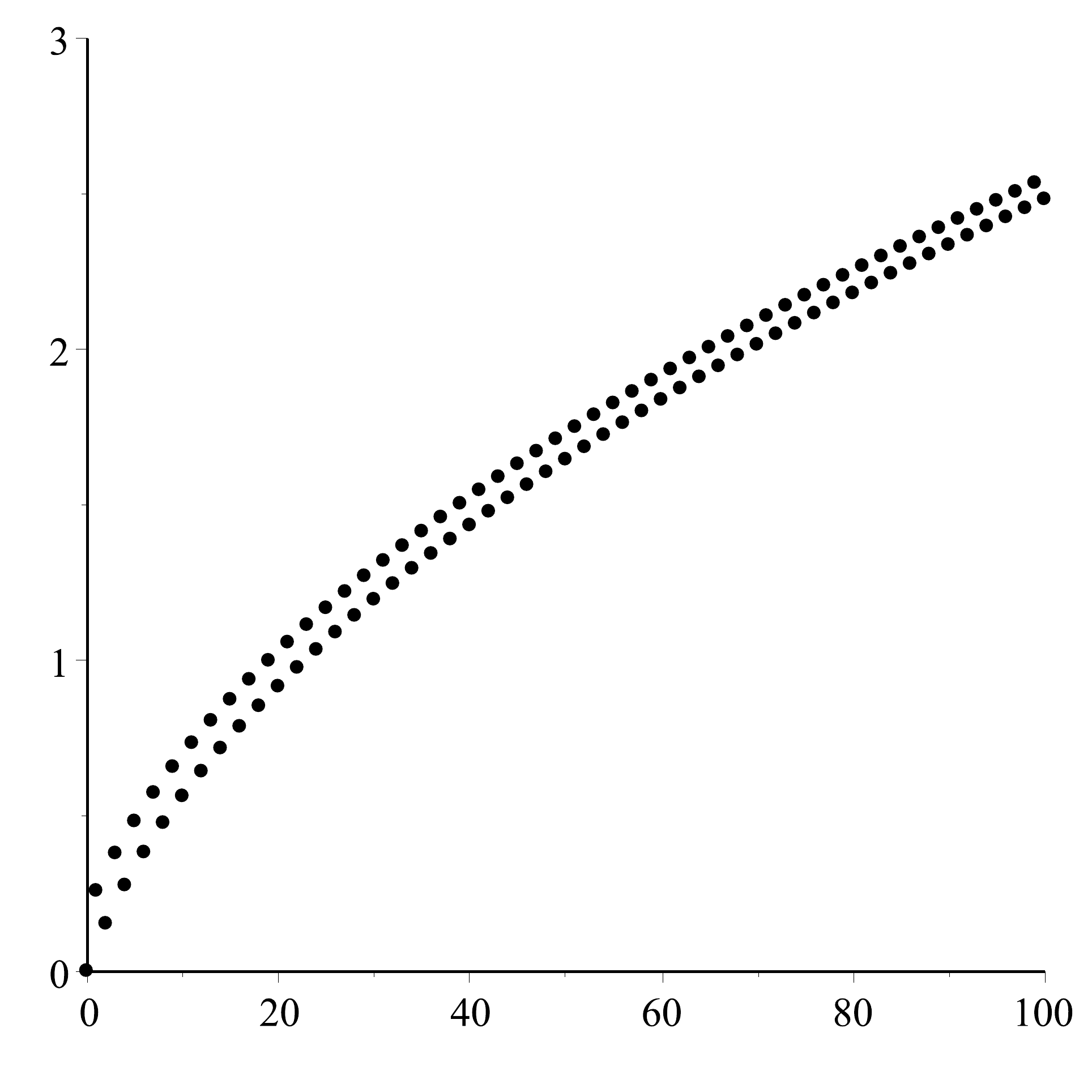}\\
t=-15 & t=-10 & t=-5 \\[10pt]
\figg{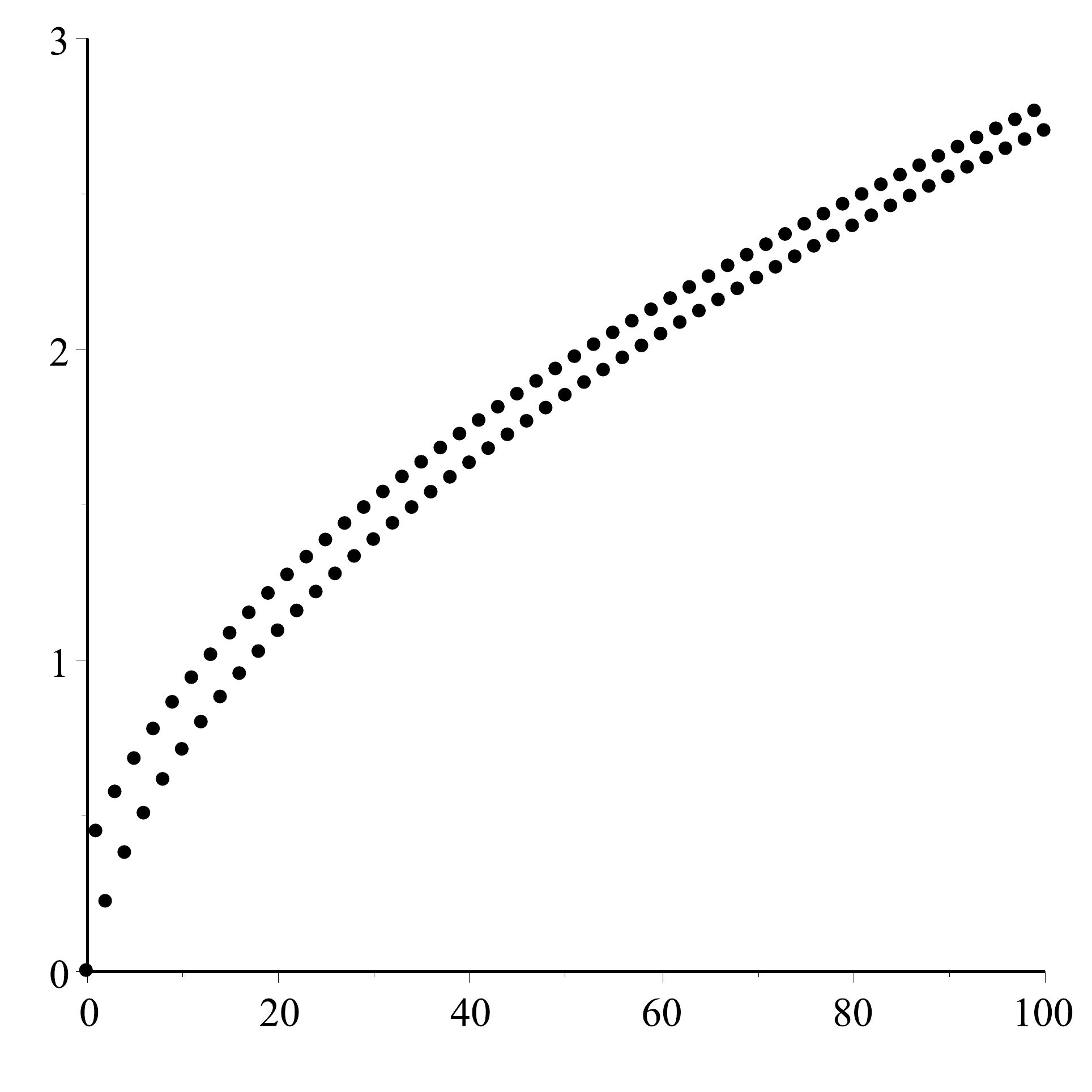} & \figg{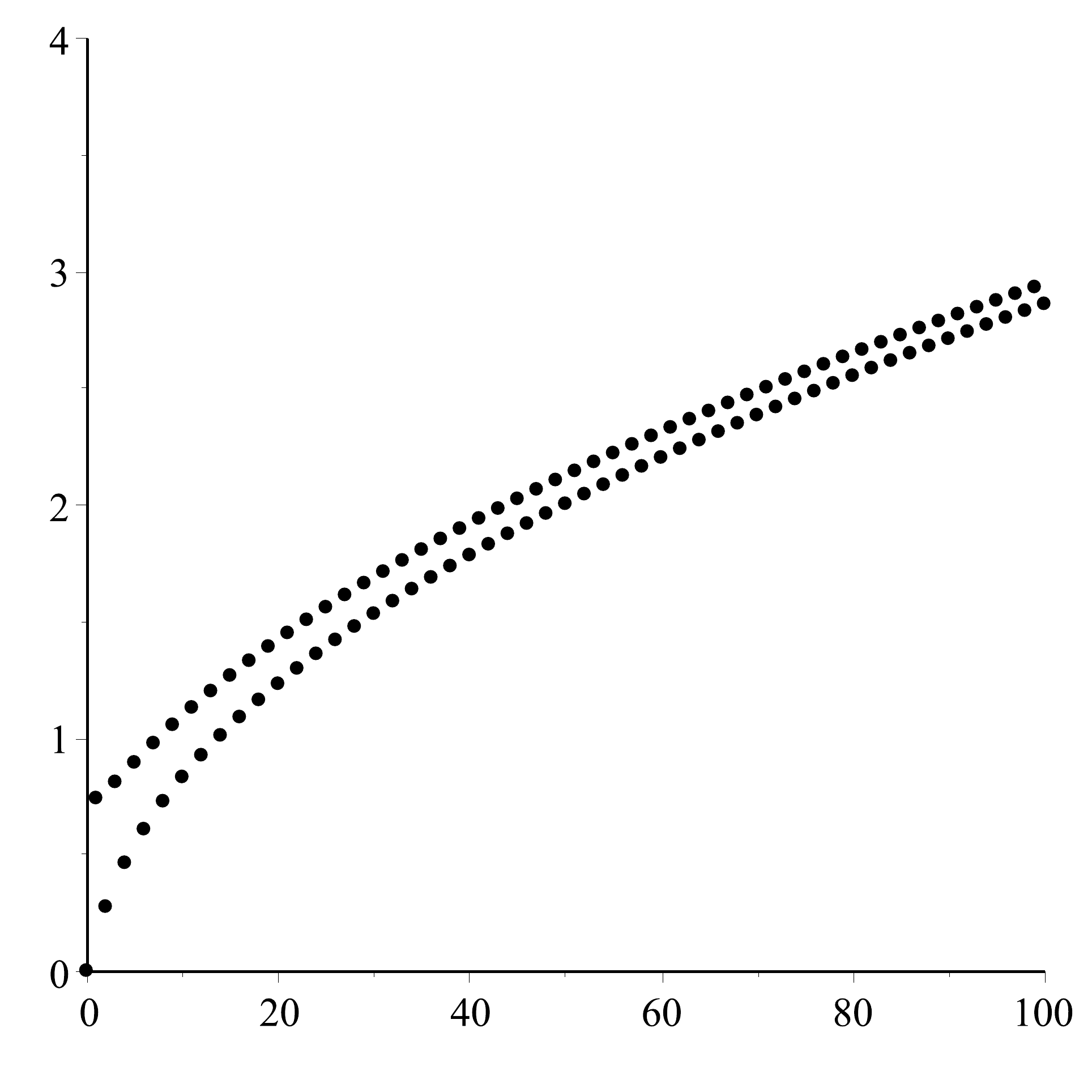} & \figg{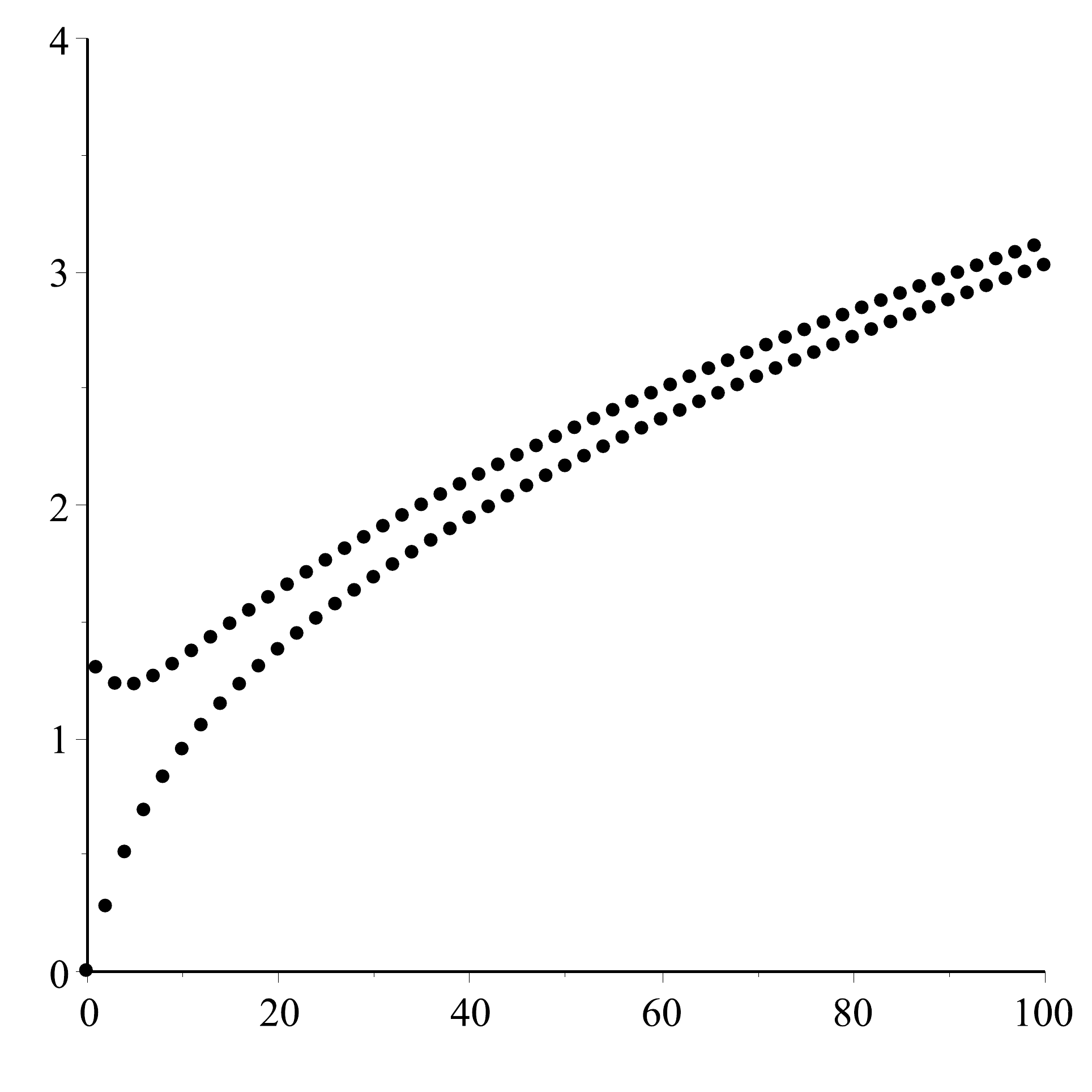} \\
t=-2 &t=0 & t=2 \\[10pt]
\figg{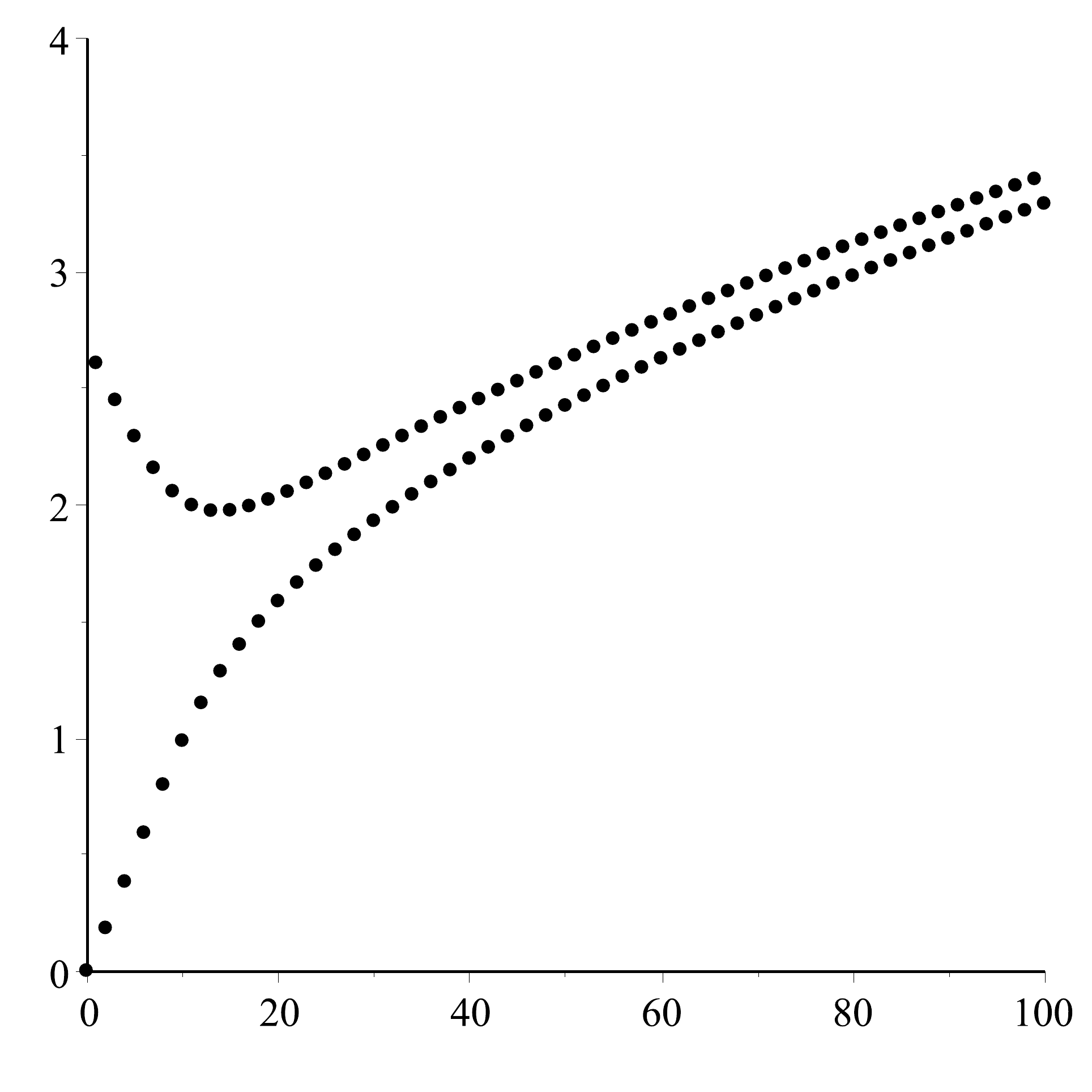} & \figg{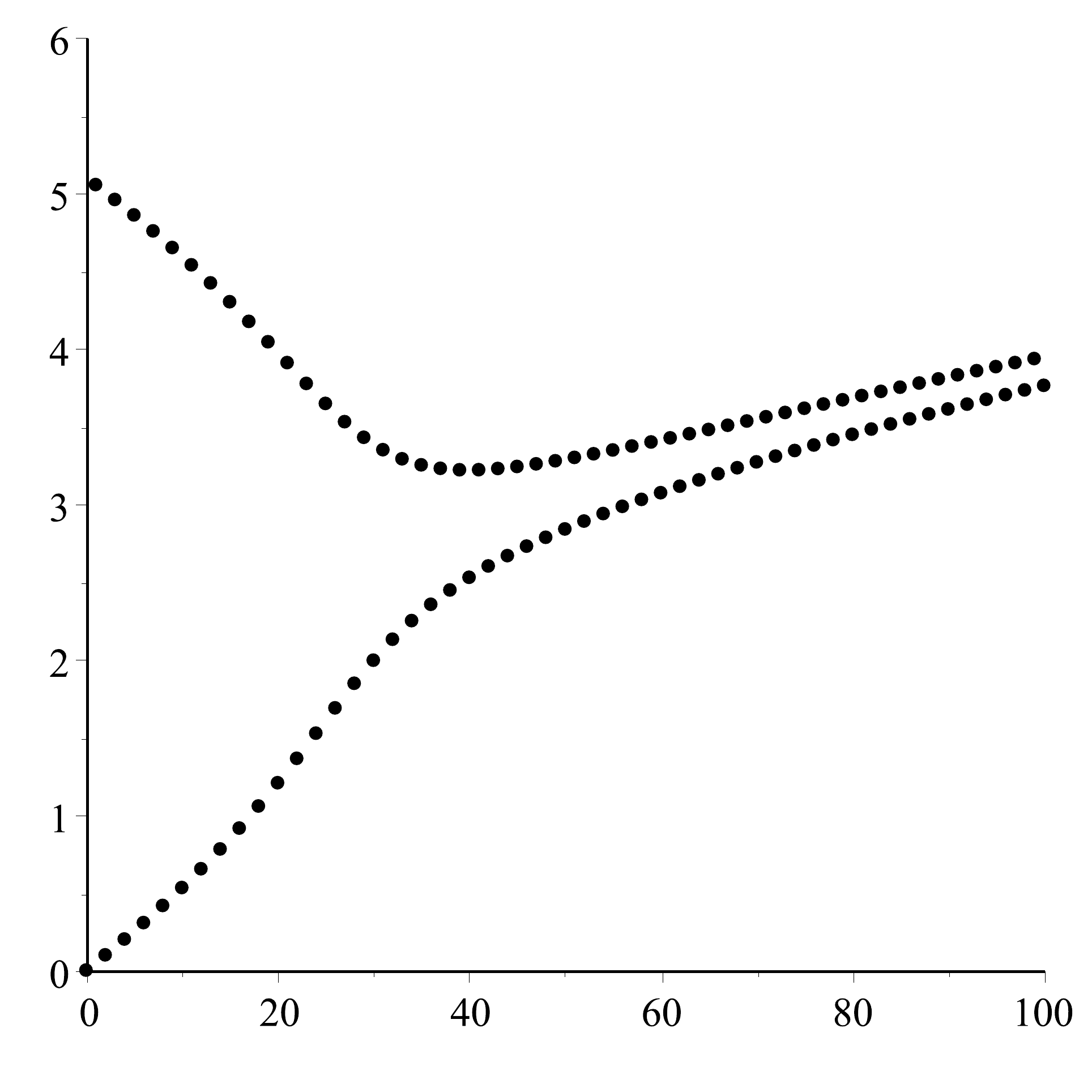} & \figg{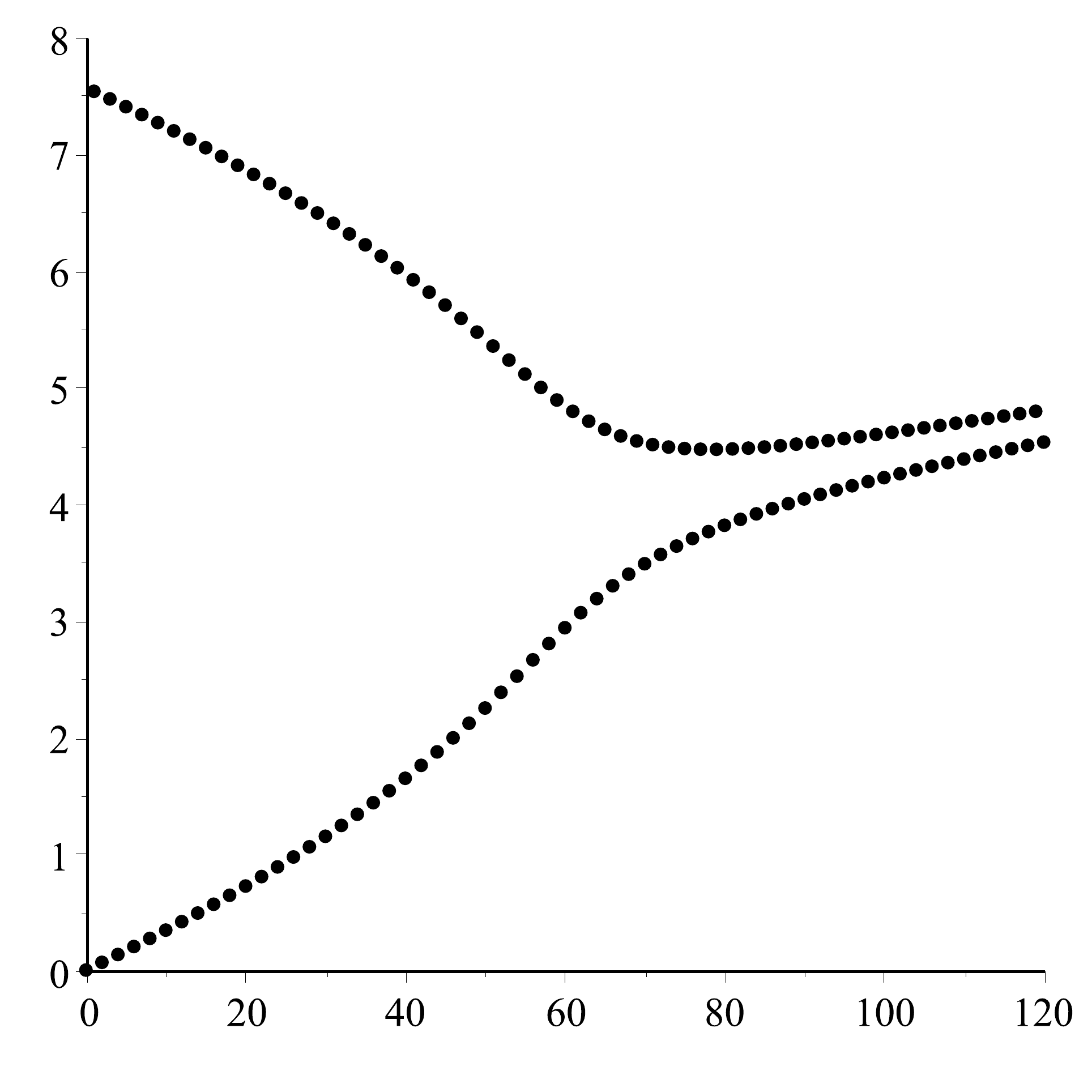}\\
t=5 & t=10 & t=15
\end{array}\]
\caption{\label{fig:nbn}Plots of the points $(n,\b_n)$ where $\b_n$ satisfies \eqref{Th48eqn} with initial conditions $\b_0=0$ and $\b_1$ given by \eqref{beta1uni}, for various choices of $t$, with $\la=\tfrac12$.}\end{figure}

\item The solution of the nonlinear discrete equation \eqref{Th48eqn} is highly sensitive to the initial conditions.
This is illustrated in Figure \ref{fig:bnsensitive} where the points $(n,\b_n)$ are plotted for the initial conditions
\[ b_0=0,\qquad \b_1=\Ph(t)+\ep,\]
where $\Ph(t)$ is given by \eqref{def:Phi}, and $\ep\in\{0,10^{-4},-10^{-4}\}$, for various choices of $t$. The plots clearly show that a small change in $\b_1$ gives rise to very different behaviour for $\b_n$.
The plots were also generated in MAPLE using $250$ digits accuracy.
\begin{figure}[ht]
\[\begin{array}{c@{\quad}c@{\quad}c}
\figg{betan0} & \figg{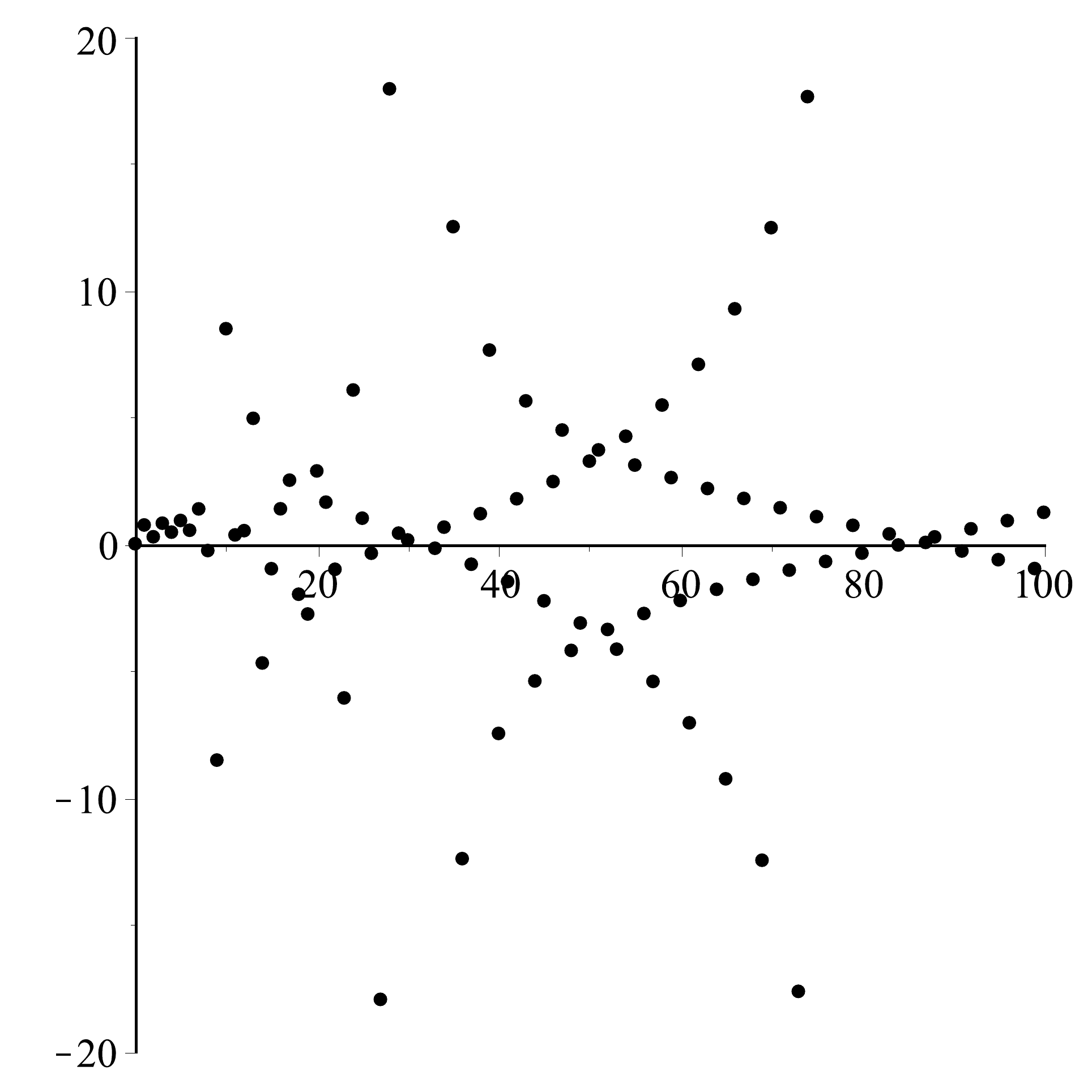} & \figg{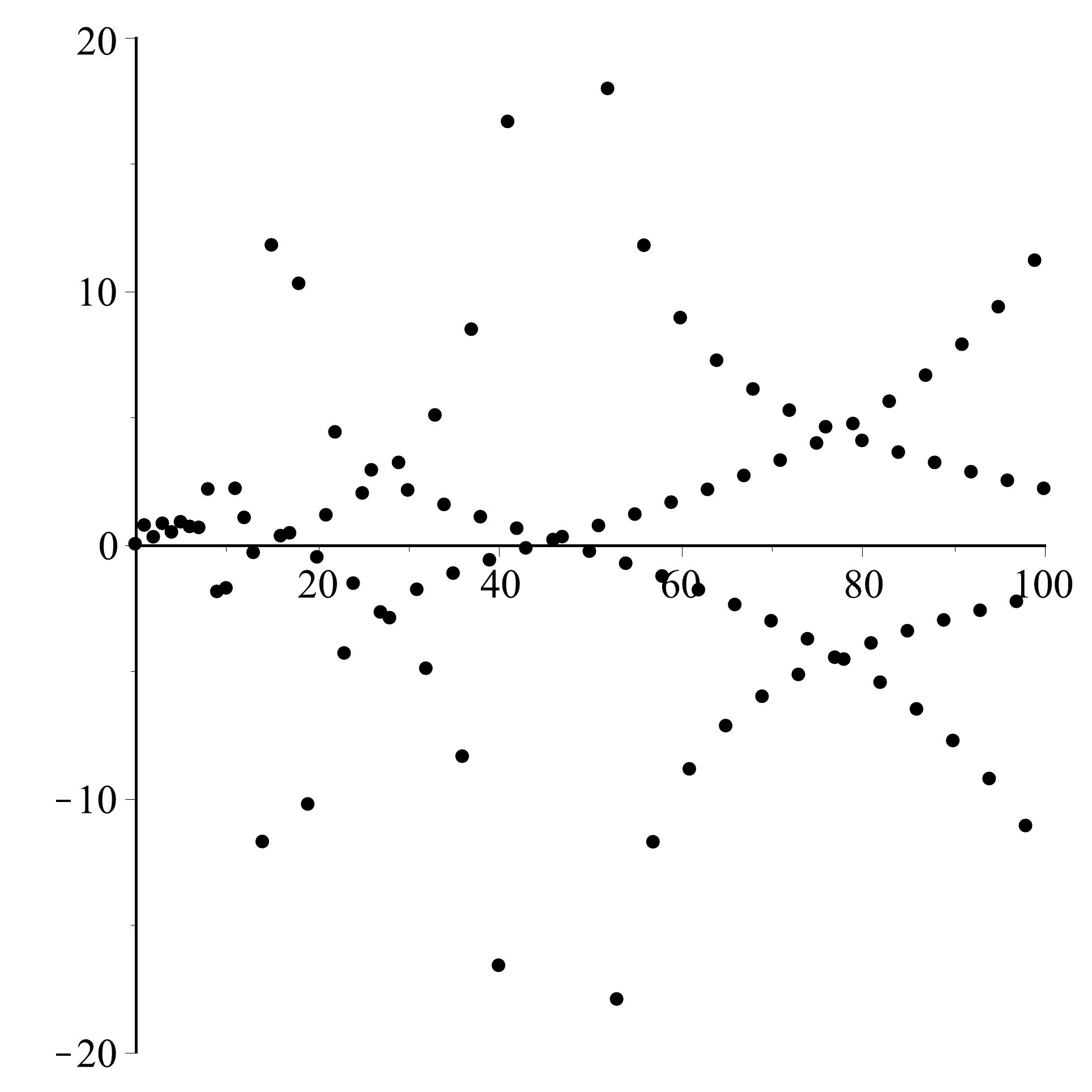}\\
\b_1=\Ph(0) & \b_1=\Ph(0)+10^{-4}& \b_1=\Ph(0)-10^{-4}\\[10pt]
\figg{betan5} & \figg{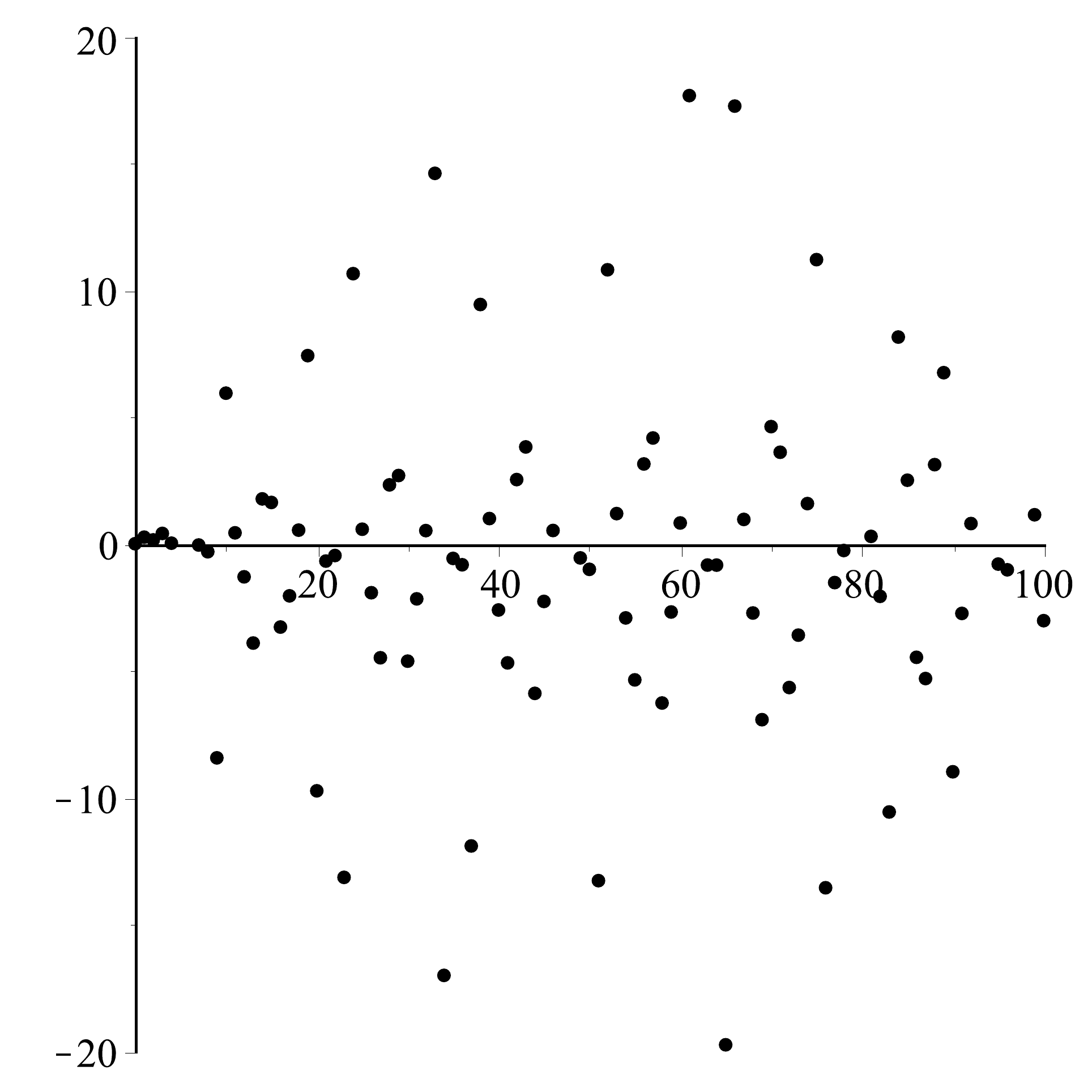} & \figg{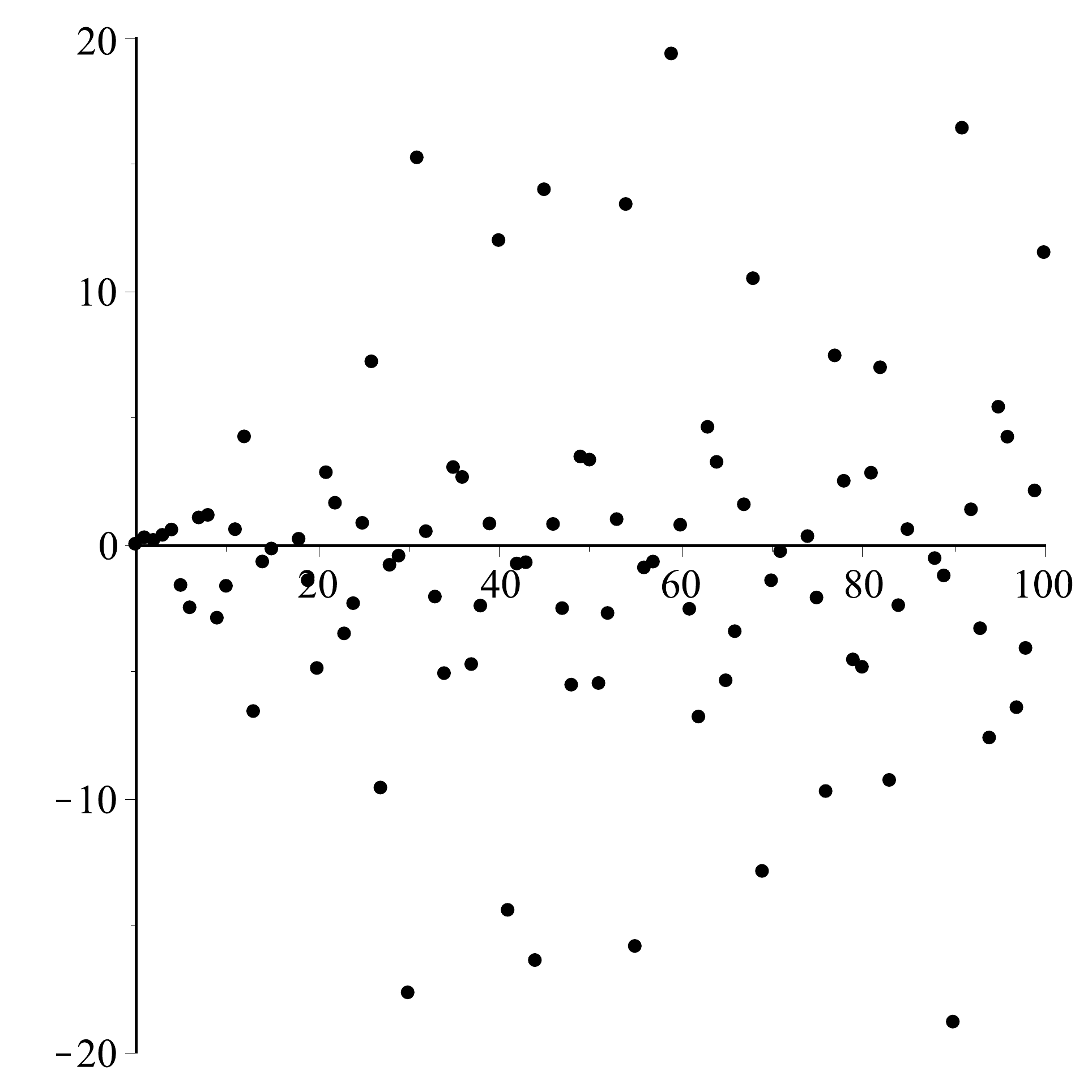}\\
\b_1=\Ph(5) & \b_1=\Ph(5)+10^{-4}& \b_1=\Ph(5)-10^{-4}\\[10pt]
\figg{betanm5} & \figg{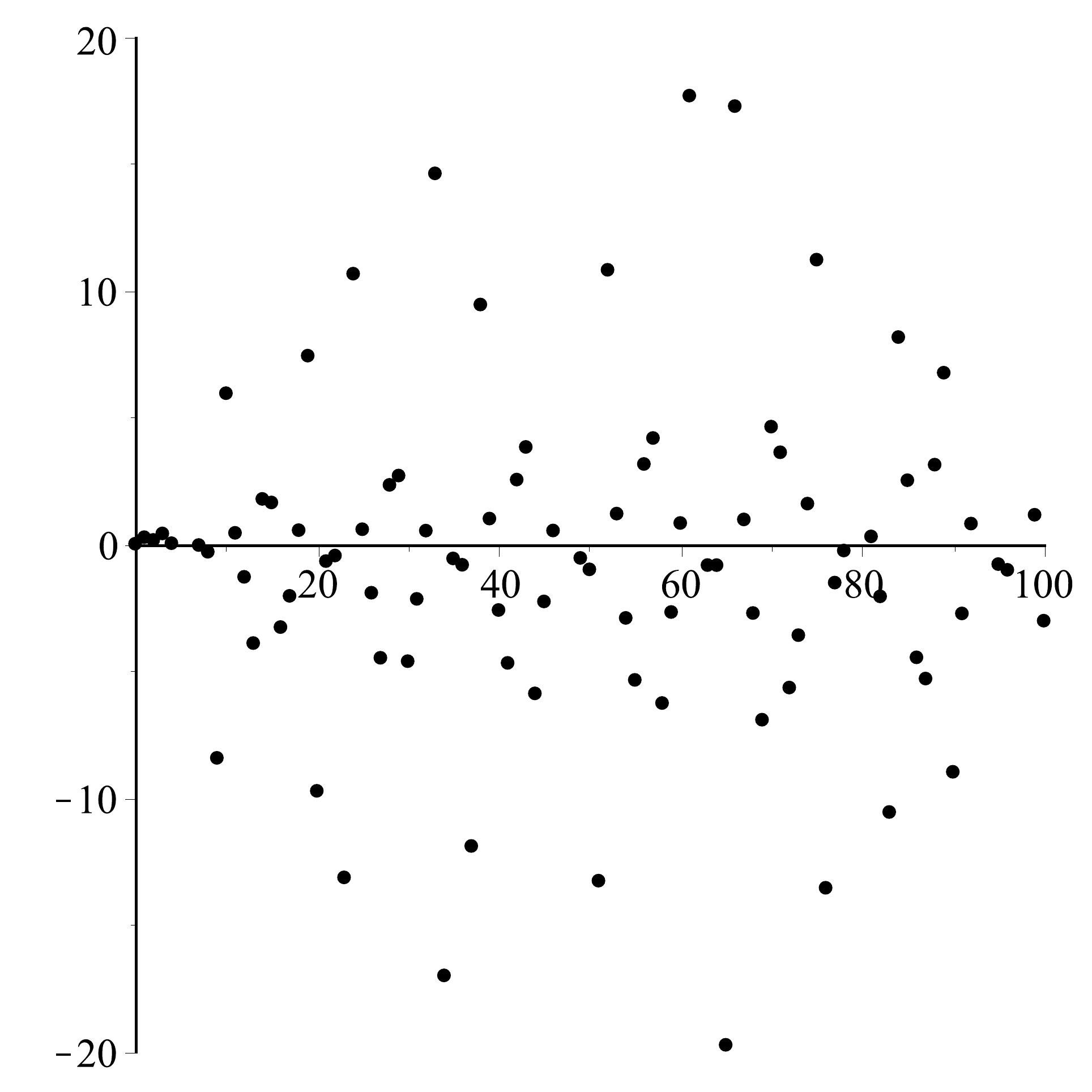} & \figg{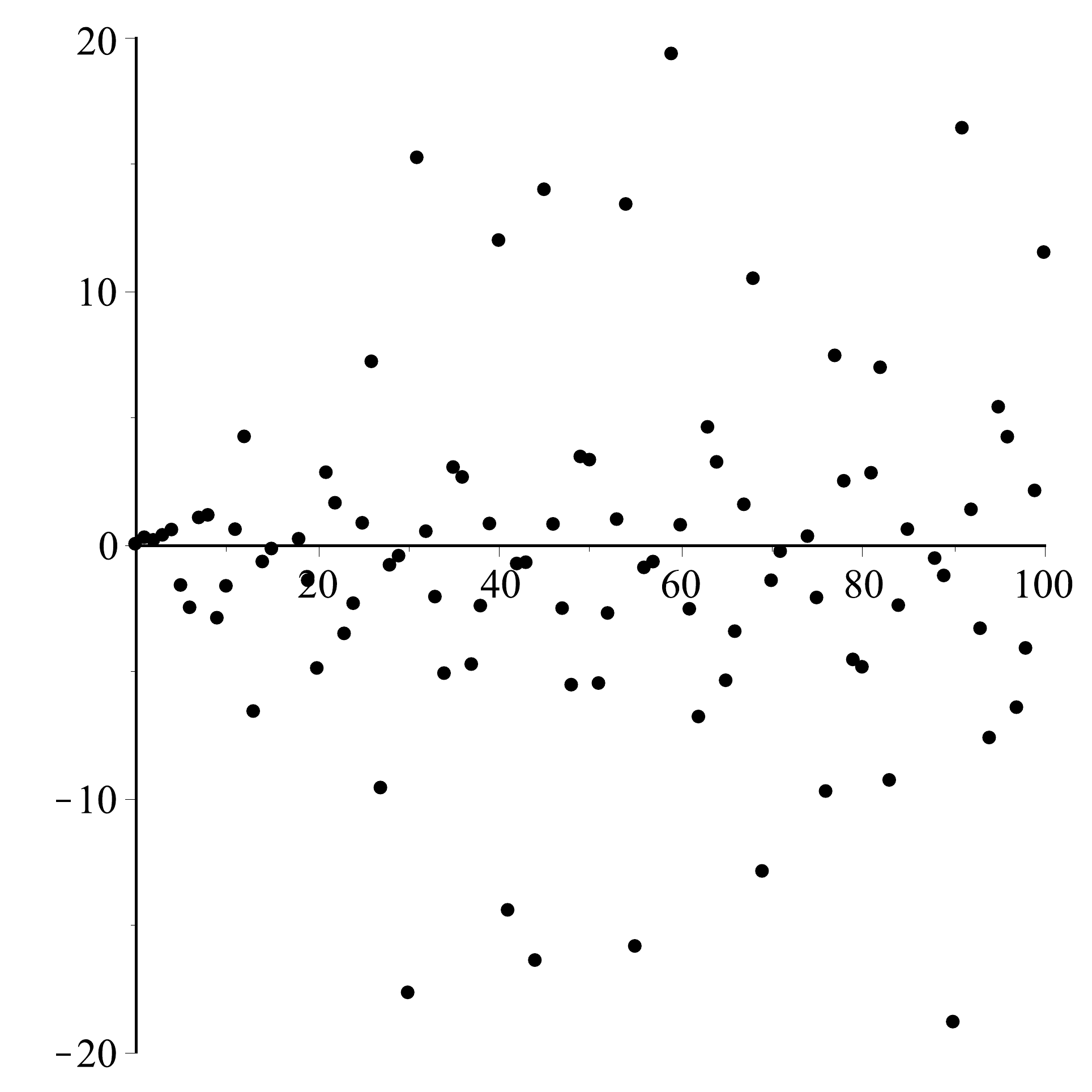}\\
\b_1=\Ph(-5) & \b_1=\Ph(-5)+10^{-4}& \b_1=\Ph(-5)-10^{-4}
\end{array}\]
\caption{\label{fig:bnsensitive}Plots of the points $(n,\b_n)$ where $\b_n$ satisfies \eqref{Th48eqn} with initial conditions $\b_0=0$ and $\b_1=\Ph(t)+\ep$, with $\Ph(t)$ given by \eqref{def:Phi}, and $\ep=0, \pm 10^{-4}$, for $t=0,\pm5$ and $\la=\tfrac12$. }\end{figure}
\end{enumerate}\end{remarks}

\section{Properties of the zeros of generalized Freud polynomials}\label{sec6}
In this section we begin by proving some properties for the zeros of semiclassical Laguerre polynomials (cf.~\cite{refCJ}) and then extend this to obtain analogous results for the zeros of monic generalized Freud polynomials, which arise from a symmetrization of semiclassical Laguerre polynomials (cf.~\cite{refCJK,refFvAZ}). For a discussion of the asymptotic behaviour as $n\to\infty$ for the recurrence coefficients and orthogonal polynomials with respect to the Laguerre-type weight
\[ w(x)=x^{\la}\exp \{-Q(x)\}, \qquad \la>-1,\quad x\in \R^{+},\]
with $Q(x)$ a polynomial with positive leading coefficient, see Vanlessen \cite{refVanlessen}.

As shown in \cite{refCJ}, the monic semiclassical Laguerre polynomials $L_n^{(\la)}(x;t)$, orthogonal with respect to the weight
\begin{equation}\label{lagw}
\ww{x} = x^{\la}\exp (-x^2 +tx), \qquad \la>-1,\quad x\in \R^{+}
\end{equation}
satisfy the three-term recurrence relation
\begin{equation}\label{lagrr}
L_{n+1}^{(\la)}(x;t)=[x-\widetilde{\a}_n(t)] L_{n}^{(\la)}(x;t)-\widetilde{\b}_{n}(t)L_{n-1}^{(\la)}(x;t)\end{equation} where
\begin{subequations}\label{eq:scLag31ab}\begin{align}\label{eq:scLag31a}
\widetilde{\a}_n(t)&=\tfrac12q_n(z)+\tfrac12t,\\ \label{eq:scLag31b}
\widetilde{\b}_n(t)&=-\tfrac18\deriv{q_n}z-\tfrac18 q_n^2(z) -\tfrac14zq_n(z)+\tfrac14\la+\tfrac12n,\end{align}\end{subequations}
with $z=\tfrac12t$. Here
\begin{equation*}q_n(z)=-2z+\deriv{}{z}\ln\frac{\Psi_{n+1,\la}(z)}{\Psi_{n,\la}(z)}\end{equation*}
satisfies \PIV\ (\ref{eq:PIV}) with parameters $(A,B)=(2n+\la+1,-2\la^2)$ and
\begin{equation*}
\Psi_{n,\la}(z)=\W\left(\psi_{\la},\deriv{\psi_{\la}}z,\ldots,\deriv[n-1]{\psi_{\la}}z\right),\qquad \Psi_{0,\la}(z)=1,
\end{equation*}
where
\begin{equation*}\label{eq:PT.SF.eq413}
\psi_{\la}(z)=\begin{cases} D_{-\la-1}\big(-\sqrt2\,z\big)\exp\big(\tfrac12z^2\big),\quad
&\mbox{\rm if}\quad \la\not\in\N,\\[2.5pt]
\ds \frac{\sqrt{\pi}}{m!\,2^{(m+1)/2}}\deriv[m]{}{z}\left\{\big[1+\erf(z)\big]\exp(z^2)\right\}, &\mbox{\rm if}\quad \la=m\in\N,
\end{cases}\end{equation*}
with $\erf(z)$ the error function, since the parabolic cylinder function $D_{-m-1}(z)$ is expressed in terms of error functions for $m\in\N$, cf.~\cite[\S12.7(ii)]{refNIST}.

\begin{theorem}\label{laguerrebounds}Let $L_n^{(\la)}(x;t)$ denote the monic semiclassical Laguerre polynomials orthogonal with respect to
\begin{equation*}
\ww{x} = x^{\la}\exp(-x^2 +tx),\quad x\in \R^{+}.
\end{equation*}

Then, for $\la>-1$ and $t\in \R$, the zeros $x_{1,n}<x_{2,n}<\dots<x_{n,n}$ of $L_n^{(\la)}(x;t)$
\begin{itemize}
\item[(i)] are real, distinct and \begin{equation}0<x_{1,n}<x_{1,n-1}<x_{2,n}<\dots<x_{n-1,n}<x_{n-1,n-1}<x_{n,n}\label{sep};\end{equation}
\item[(ii)] strictly increase with both $t$ and $\la$;
\item[(iii)] satisfy
\[a_n<x_{1,n}<\widetilde{\a}_{n-1}<x_{n,n}< b_n,\]
where
\begin{align*}
a_n=&\min_{1\leq k \leq n-1}\left\{\tfrac12(\widetilde{\a}_k+\widetilde{\a}_{k-1})-\tfrac12\sqrt{(\widetilde{\a}_k+\widetilde{\a}_{k-1})^2+4c_n\widetilde{\b}_k}\right\},\\
b_n=&\max_{1\leq k \leq n-1}\left\{\tfrac12(\widetilde{\a}_k+\widetilde{\a}_{k-1})+\tfrac12\sqrt{(\widetilde{\a}_k+\widetilde{\a}_{k-1})^2+4c_n\widetilde{\b}_k}\right\},
\end{align*}
with $c_n=4\cos^2\left(\frac{\pi}{n+1}\right)+\ep$, $\ep>0$, and
$\widetilde{\a}_n$ and $\widetilde{\b}_n$ given by \eqref{eq:scLag31ab}.
\end{itemize}
\end{theorem}
\begin{proof}\begin{itemize}\item[]
\item[(i)] The proofs for classical orthogonal polynomials, where $t=0$ (see, for example, \cite[Thm 3.3.1 and 3.3.2]{refSzego}), work without change.
\item[(ii)] Since for the semiclassical Laguerre weight \eqref{lagw}
\[\ds{\frac{\partial}{\partial \la}\ln \ww{x}=\ln x}\] is an increasing function of $x$, it follows from Markov's monotonicity theorem (cf.~\cite[Theorem 6.12.1]{refSzego} that the zeros of $L_n^{(\la)}(x;t)$ increase as $\la$ increases.
Similarly, since \[\ds{\frac{\partial}{\partial t}\ln \ww{x}=x},\] increases with $x$, it follows that the zeros of $L_n^{(\la)}(x;t)$ increase as $t$ increases.
\item[(iii)]
The inner bound $\widetilde{\a}_{n-1}$ for the extreme zeros follows from \cite[Cor. 2.2]{refDJ12} together with \eqref{lagrr} and \eqref{sep} since $\widetilde{\b}_{n-1}(t)$ does not depend on $x$. The outer bounds $a_n$ and $b_n$ for the extreme zeros $x_{1,n}$ and $x_{n,n}$ respectively, follow from the approach based on the Wall-Wetzel Theorem, introduced by Ismail and Li \cite{refIsmailLi} (see also \cite{refIsmail}) by applying their Theorems 2 and 3 to the three term recurrence relation (\ref{lagrr}).
\end{itemize}
\end{proof}

 Asymptotic properties of the extreme zeros of generalized Freud polynomials related to the weight \eqref{conj} were studied by Freud \cite{refFreud1986} and Nevai \cite{refNevai1986}. Subsequently, Kasuga and Sakai \cite{refKasugaSakai} extended and generalized these results.

Next we prove some properties of zeros of generalized Freud polynomials associated with the weight \eqref{genFreud}. The weight \eqref{genFreud} is even and it is well-known that, in this case, the zeros of the corresponding orthogonal polynomials are symmetric about the origin. This implies that the positive and the negative zeros have opposing monotonicity and, as a result of this symmetry, it suffices to study the monotonicity and bounds of the positive zeros. 
\begin{theorem}\label{zerofreud}Let $S_n(x;t)$ be the monic generalized Freud polynomials orthogonal with respect to the weight \eqref{genFreud}, i.e.
\begin{equation*}
\ww{x}=|x|^{2\la+1} \exp(-x^4+tx^2),\end{equation*}
and let $x_{n,1}(\la,t)<x_{n,2}(\la,t)<\dots<x_{n,\lfloor n/2\rfloor }(\la,t)$ denote the positive zeros of $S_n(x;t)$ (recall $\lfloor m\rfloor $ is the largest integer smaller than $m$). Then, for $\la>-1$ and $t\in \R$
 \begin{itemize}
 \item[(i)]
 the zeros of $S_n(x;t)$ are real and distinct and
 \[\begin{split}0<x_{n,1}(\la,t)&<x_{n-1,1}(\la,t)<x_{n,2}(\la,t)<\dots 
 <x_{n-1,\left[(n-1)/2\right]}(\la,t)<x_{n,\left[ n/2\right]}(\la,t);\end{split}\]
 \item[(ii)] the $\nu$th zero $x_{n,\nu}(\la,t)$, for a fixed value of $\nu$, is an increasing function of both $\la$ and $t$;
 \item[(iii)] the largest zero satisfies the inequality
\begin{equation}x_{n,\left[{n}/{2}\right]}(\la,t)<\max_{1\leq k\leq n-1}\sqrt{c_n\b_k(t;\la)},\label{lzb}\end{equation}
where $c_n=4\cos^2\left(\frac{\pi}{n+1}\right)+\ep$, $\ep>0$, and
$\b_n(t;\la)$ is given by \eqref{betathreeAAab}.
\end{itemize}
\end{theorem}
\begin{proof}
\begin{itemize}\item[]
\item[(i)]
This follows from Theorem~\ref{laguerrebounds}(i) using the relation (cf.~\cite{refChihara,refCJK,refFvAZ})
\begin{align*}
S_{2n}(x;t) =& L_n^{(\la)}(x^2;t),\qquad
S_{2n+1}(x;t) = xL_n^{(\la+1)}(x^2;t).\end{align*}
between the semiclassical Laguerre polynomials $L_n^{(\la)}(x;t)$ and the generalized Freud polynomials $S_n(x;t)$.
\item[(ii)]
The monotonicity of the positive zeros with respect to the parameters $\la$ and $t$ follows from \cite[Theorem 2.1]{refJWZ} since for the generalized Freud weight \eqref{genFreud}
\[\ds{\frac{\partial}{\partial \la}\ln \ww{x}=2\ln x},\quad\ds{\frac{\partial}{\partial t}\ln \ww{x}=x^2},\]
both increase with $x$.
\item[(iii)] The inequality (\ref{lzb}) for the largest zero $x_{n,\lfloor n/2\rfloor }(\la,t)$ follows by applying \cite[Theorem 2 and 3]{refIsmailLi} to the three term recurrence relation \eqref{eq:gfrr}.
\end{itemize}
\end{proof}


\section{Conclusion}
In this paper we have analysed the asymptotic behaviour of generalized Freud polynomials, orthogonal with respect to the generalized Freud weight \eqref{genFreud}, in two different contexts. Firstly, we obtained asymptotic results for the polynomials when the parameter $t$ involved in the semiclassical perturbation of the weight tends to $\pm \infty$. Next, we considered the strong asymptotics of the coefficients $\b_n$ in the three-term recurrence relation \eqref{eq:gfrr} satisfied by the generalized Freud polynomials $S_n(x;t)$ as the degree $n$ tends to infinity and investigated the asymptotic behaviour of the polynomials themselves as the degree increases. We showed that unique, positive solutions of the nonlinear discrete equation \eqref{nonlineardiffer} satisfied by the recurrence coefficients exist for all $t\in\R$ but that these solutions are very sensitive to the initial conditions. We also proved various properties of the zeros of generalized Freud polynomials. The closed form expressions for the recurrence coefficients obtained in \cite{refCJK} allowed the investigation of the properties of generalized Freud polynomials in this paper. A natural extension of this work would be an investigation of asymptotic properties using limiting relations satisfied by the polynomials as the parameter $\la$ tends to $\infty$.

\section*{Acknowledgments}
We thank the African Institute for Mathematical Sciences, Muizenberg, South Africa, for their hospitality during our visit when some of this research was done, supported by their ``Research in Pairs" programme. We also thank the referees for helpful comments and additional references.
PAC thanks Alfredo Dea\~no, Alexander Its, Ana Loureiro and Walter Van Assche for their helpful comments and illuminating discussions and also the Department of Mathematics, National Taiwan University, Taipei, Taiwan, for their hospitality during his visit where some of this paper was written.
KJ thanks Abey Kelil for pointing out several useful references and insightful discussions.
The research by KJ was partially supported by the National Research Foundation of South Africa under grant number 103872.

\def\p{Painlev\'{e}}
\def\JPA{J. Phys. A}
\def\bibitm{\vspace{-5pt}\bibitem}

\def\refjl#1#2#3#4#5#6#7{
\bibitm{#1}\textrm{\frenchspacing#2}, \textrm{#3},
\textit{\frenchspacing#4}, \textbf{#5}\ (#7)\ #6.}

\def\refjltoap#1#2#3#4#5#6#7{
\bibitm{#1} \textrm{\frenchspacing#2}, \textrm{#3},
\textit{\frenchspacing#4}, 
#6.}

\def\refbk#1#2#3#4#5{
\bibitm{#1} \textrm{\frenchspacing#2}, \textit{#3}, #4, #5.}

\def\refcf#1#2#3#4#5#6#7{
\bibitm{#1} \textrm{\frenchspacing#2}, \textrm{#3},
in: \textit{#4}, {\frenchspacing#5}, #6, pp.\ #7.}

 \end{document}